\documentclass[a4paper,UKenglish,pdfa,cleveref,autoref,numberwithinsect]{lipics-%
v2021}

\nolinenumbers

\usepackage{microtype}

\title{Continuous rational functions\\ are deterministic regular}

\titlerunning{Continuous rational functions are deterministic regular}

\author{Olivier Carton}{Université Paris Cité, CNRS, IRIF, F-75013, Paris, France}%
{carton@irif.fr}{https://orcid.org/0000-0002-2728-6534}
{}

\author{Gaëtan Douéneau-Tabot}{Université Paris Cité, CNRS, IRIF, F-75013, Paris, France
\and Direction générale de l'armement - Ingénierie des projets, Paris, France
}{doueneau@irif.fr}{}{}

\authorrunning{O. Carton and G. Douéneau-Tabot}

\Copyright{Olivier Carton and Gaëtan Douéneau-Tabot}



\category{}

\funding{This work was supported by the DeLTA ANR project %
(\href{https://anr.fr/Projet-ANR-16-CE40-0007}{ANR-16-CE40-0007}).}

\relatedversion{}
\relatedversiondetails{Full Version}{https://arxiv.org/abs/2204.11235}

\acknowledgements{}

\EventEditors{}
\EventNoEds{0} 
\EventLongTitle{}
\EventShortTitle{MFCS 2022}
\EventAcronym{}
\EventYear{}
\EventDate{}
\EventLocation{}
\EventLogo{}
\SeriesVolume{}
\ArticleNo{}

\usepackage{tikz}
\usepackage[utf8]{inputenc}
\usetikzlibrary{automata,positioning,arrows}
\usepackage{amsfonts, amsmath, amssymb,dsfont, mathrsfs}
\usepackage{amsthm}
\usepackage{mathtools}

\usepackage{hyperref}
\usepackage[ruled]{algorithm2e}
\usepackage{stmaryrd}
\usepackage{pifont,enumitem}
\usepackage{soulutf8}
\usepackage{yhmath}
\usepackage{bm}

\setlist[itemize]{noitemsep, topsep=0pt}
\setlist[description]{noitemsep, topsep=0pt}
\setlist[enumerate]{noitemsep, topsep=0pt}

\newtheorem{sublemma}[theorem]{\bfseries Sublemma}

\newcommand{\mb}[1]{\mathbb{#1}}
\newcommand{\mc}[1]{\mathcal{#1}}
\newcommand{\mf}[1]{\mathfrak{#1}}

\newcommand{\movi}{\varepsilon}
\newcommand{\vide}{\varnothing}

\newcommand{\defined}{\coloneqq}
\renewcommand{\phi}{\varphi}
\renewcommand{\le}{\leqslant}
\renewcommand{\ge}{\geqslant}
\renewcommand{\epsilon}{\operatorname{\normalfont \textcolor{green}{%
\textsf{FORBIDDEN}}}}
\newcommand{\Nat}{\mb{N}}
\newcommand{\pref}{\sqsubseteq}
\newcommand{\prefneq}{\sqsubset}
\renewcommand{\subsetneq}{\subset}
\renewcommand{\supsetneq}{\supset}
\newcommand{\parfonc}{\rightharpoonup}
\newcommand{\fonc}{\rightarrow}

\newcommand{\trans}{\mc{T}}
\newcommand{\auto}{\mc{A}}
\newcommand{\becomes}{\leftarrow}
\newcommand{\ntrans}{\mc{N}}
\newcommand{\strans}{\mc{S}}
\newcommand{\lmove}{\triangleleft}
\newcommand{\rmove}{\triangleright}
\newcommand{{\lmark}}{\vdash}
\newcommand{\runs}[1]{\mathchoice{\xrightarrow{#1}}{\xrightarrow{\smash{\lower1pt\hbox{$\scriptstyle #1$}}}}{\xrightarrow{#1}}{\xrightarrow{#1}}}

\newcommand{\oNT}{$1$-$\operatorname{\normalfont \textsf{nT}}${}}
\newcommand{\tDT}{$2$-$\operatorname{\normalfont \textsf{dT}}${}}

\newcommand{\DSST}{$\operatorname{\normalfont \textsf{dSST}}${}}
\newcommand{\NSST}{$\operatorname{\normalfont \textsf{nSST}}${}}
\newcommand{\Dom}[1]{\operatorname{\normalfont \textsf{Dom}}(#1)}

\newcommand{\out}{\operatorname{\mf{out}}}
\newcommand{\outi}[1]{\operatorname{{\mf{out}}}_{#1}}
\newcommand{\exinf}{\operatorname{\normalfont \textsf{normalize}}}
\newcommand{\replace}{\operatorname{\normalfont \textsf{replace}}}
\newcommand{\double}{\operatorname{\normalfont \textsf{double}}}

\newcommand{\new}[1]{\overline{#1}}
\newcommand{\Regs}{\mf{R}}
\newcommand{\Reggs}{\mf{R}'}
\newcommand{\reg}{\mf{r}}
\newcommand{\regg}{\mf{s}}
\newcommand{\reggg}{\mf{t}}
\newcommand{\pc}{\prec}
\newcommand{\four}{2}
\newcommand{\Chi}{\raisebox{2pt}{\mbox{{\large $\chi$}}}}

\newcommand{\ctime}{\operatorname{\normalfont \textsf{time}}}
\newcommand{\en}[1]{\operatorname{\normalfont \textsf{end}}_{#1}}
\newcommand{\tree}[1]{\operatorname{\normalfont \textsf{tree}}(#1)}
\newcommand{\lag}{\operatorname{\normalfont \textsf{lag}}}
\newcommand{\lagm}{\operatorname{\normalfont \textsf{max-lag}}}
\newcommand{\Jf}{\operatorname{\normalfont \textsf{J}}}
\newcommand{\Cf}{\operatorname{\normalfont \textsf{C}}}
\newcommand{\last}{\operatorname{\normalfont \textsf{last}}}
\newcommand{\cov}{\operatorname{\normalfont \textsf{cover}}}
\newcommand{\used}{\operatorname{\normalfont \textsf{used}}}
\newcommand{\Good}{\operatorname{\normalfont \textsf{Good}}}
\newcommand{\first}{\operatorname{\normalfont \textsf{first}}}
\newcommand{\nnext}{\operatorname{\normalfont \textsf{next}}}
\newcommand{\Qubot}{Q_{\bot}}
\newcommand{\tnorm}[1]{\text{\normalfont {#1}}}
\newcommand{\aval}{\operatorname{\normalfont \textbf{substitute}}}

\newcommand{\simu}{\operatorname{\normalfont \textbf{simulation}}}
\newcommand{\copies}{\operatorname{\normalfont \textsf{copies}}}
\newcommand{\Comp}{\operatorname{\normalfont \textsf{Comp}}}

\newcommand{\Parts}{\operatorname{\normalfont \textsf{Parts}}}
\newcommand{\pre}[2]{\operatorname{\normalfont \textsf{pre}}_{#1}^{#2}}
\newcommand{\nb}[1]{\operatorname{\normalfont \textsf{nb}}_{#1}}
\newcommand{\val}[2]{\operatorname{\normalfont \textsf{prod}}_{#1}^{#2}}
\newcommand{\com}[2]{\operatorname{\normalfont \textsf{com}}_{#1}^{#2}}
\newcommand{\adv}[2]{\operatorname{\normalfont \textsf{adv}}_{#1}^{#2}}
\newcommand{\advm}[2]{\operatorname{\normalfont \textsf{max-adv}}_{#1}^{#2}}

\newcommand{\subst}[2]{\mc{S}_{#1}^{#2}}

\newcommand{\Hi}{\mf{H}}

\newcommand{\bound}{|Q|^{|Q|}}
\newcommand{\Bound}{\Omega}

\newcommand{\myBlue}{blue!90}
\newcommand{\tred}[1]{\textcolor{\myBlue}{#1}}
\newcommand{\crod}[1]{\mathop{\cdot} {  #1}}
\newcommand{\push}[2]{#2\crod{#1}}
\newcommand{\cro}[1]{ \tred{\bm{\llbracket} #1\bm{\rrbracket}}}

\newcommand{\old}[1]{%
}

\newcommand{\cor}[1]{
#1}

\begin{document}
\maketitle

\begin{abstract} 
A word-to-word function is rational if it can be realized by
a non-deterministic one-way transducer. Over finite words,
it is a classical result that any rational function is regular,
i.e. it can be computed by a deterministic two-way transducer,
or equivalently, by a deterministic streaming string transducer
(a one-way automaton which manipulates string registers).

This result no longer holds for infinite words,
since a non-deterministic one-way transducer can
guess, and check along its run, properties such as infinitely many
occurrences of some pattern, which is impossible
for a deterministic machine.  In this paper, we identify the class of
rational functions over infinite words which are also computable by a
deterministic two-way transducer. It coincides with the
class of  rational functions which are continuous,
and this property can thus be decided.
This solves an open question raised in a previous paper of
Dave~et~al.
  \end{abstract}



\section{Introduction}

Transducers are finite-state machines obtained by adding outputs
to finite automata. They are very useful in a lot of areas like coding, computer
arithmetic, language processing or program analysis, and more
generally in data stream processing.
In this paper, we study transducers which
compute partial functions. They are either deterministic,
or non-deterministic but unambiguous (they have
at most one accepting run on a given input).

Over finite words, a deterministic two-way transducer (\tDT{})
consists of a deterministic
two-way automaton which can produce outputs.
Such machines realize the class of 
\emph{regular functions}, which is often considered
as one of the functional counterparts of regular languages.
It coincides with the class of functions definable
by monadic second-order transductions~\cite{engelfriet2001mso},
or copyless deterministic streaming string transducers (\DSST),
which is a model of one-way automata
manipulating string registers~\cite{alur2010expressiveness}.
On the other hand, the model
of non-deterministic one-way transducers (\oNT{})
describe the well-known class of \emph{rational functions}.
It is well known that any rational function is regular,
but the converse does not hold.

\subparagraph*{Infinite words.} The class of \emph{regular functions
over infinite words} was defined in~\cite{alur2012regular}
using monadic second-order transductions. It coincides with the class
of functions realized by \tDT{} with $\omega$-regular lookahead,
or by copyless \DSST{} with some Müller conditions.
However, the use of $\omega$-regular lookaheads
(or Müller conditions for \DSST) is necessary to
capture the expressive power of monadic second-order logic
on infinite words, in order to check properties such as infinitely many
occurrences of some pattern. Similarly, the model of \oNT{}
with Büchi acceptance conditions defines the subclass
of \emph{rational functions over infinite words}.

Even if regular and rational functions give very 
natural frameworks for specification (due to their connections
with logic), not all these functions can effectively
be computed by a deterministic machine without lookaheads.
It turns out that the regular functions which can
be computed by a deterministic Turing machine
(doing an infinite computation on its infinite input)
are exactly those which are continuous for the 
Cantor topology~\cite{dave2020synthesis}.
Furthermore continuity can be decided, which
 has been known for rational functions since~\cite{prieur2001decide}.

The authors of~\cite{dave2020synthesis} conjecture
that any continuous rational (or even regular) function 
can in fact be computed by a \tDT{} (without lookahead), instead
of a Turing machine. A partial answer was obtained
in~\cite{filiot2021synthesizing}, whose results imply that
\tDT{} can be built for a subclass of rational functions
defined by \oNT{} where some forms of non-determinism
are prohibited. Their proof is based on game-theoretic techniques.

\subparagraph*{Contributions.} This paper shows that any
continuous rational function over infinite words can be extended
to a function which is computable by a \tDT{}
(without lookaheads). Since the converse also holds,
this result completely characterizes rational functions
which can be computed by \tDT{}s, up to an extension
of the domain. Furthermore, this property
is decidable and our construction of a \tDT{} is effective.

This result is tight, in the sense
that two-way moves cannot be avoided. Indeed, one-way 
deterministic transducers (describing the class of \emph{sequential functions})
cannot realize all  continuous rational functions, even when only considering
 total functions (contrary to what happens for the subclass of
 rational functions studied in \cite{filiot2021synthesizing}).

In order to establish this theorem, we first study the expressive
power of  \tDT{} over infinite words. We introduce the class of
\emph{deterministic regular functions} as the class of functions computed by \tDT{}
(as opposed to the regular functions, which are not entirely  deterministic
since they use lookaheads to guess the future).
Following the aforementioned equivalences between two-way and register
transducers, we prove that deterministic rational functions are exactly the functions
which are realized by copyless \DSST{} (without Müller conditions).
Hence our problem is reduced to showing that any continuous rational function
can be realized by a copyless \DSST{}. Building a copyless
\DSST{} is also relevant  for practical applications, since it corresponds to
a streaming algorithm over infinite strings.

\begin{figure}[h!]

    	\begin{center}
	\hspace*{-0.4cm}
        \begin{tikzpicture}{scale=0.9}
        
            \fill[fill = blue!20]  (0,0.2) ellipse (3.8cm and 1.5cm);
            \def\detreg{(-1.2,0.1) ellipse (2.4cm and 1.1cm)};
            \fill[fill = red!50]\detreg ;
            \def\rational{(1.2,0.1) ellipse (2.4cm and 1.1cm)};
            \fill[fill = blue!40]  \rational;
            
            \begin{scope}
       	    \clip \detreg;
 	       \fill[purple!40] \rational;
 	   \end{scope}
	   
            \fill[fill = purple!20] (0,-0.3) ellipse (0.85cm and 0.35cm);

            \draw(0,1.4) node  {\scalebox{1}{$\substack{\textsf{REGULAR}}$}};

            \draw(0,-0.25) node  {\scalebox{1}{$\substack{\textsf{SEQUENTIAL}}$}};
            
          \draw(-2.3,0.2) node  {\scalebox{1}{$\substack{\textsf{DETERMINISTIC}\\ \textsf{REGULAR}}$}};

            \draw(1.7,0.9) node  {\scalebox{1}{$\substack{\textsf{RATIONAL}}$}};
            \draw(0,0.55) node  {\scalebox{1}{$\substack{\textsf{RATIONAL}%
            \\  \textsf{\& CONTINUOUS}}$}};

             \draw[darkgray,dotted,very thick] (2,0.6) -- (4.05,1.3);
             \fill[darkgray] {(2,0.6) circle (0.05)};
             \node[right] at (4.1,1.3) {\normalsize{\textcolor{darkgray}{$%
             \substack{ \tnorm{Normalization in base}~2 \tnorm{.}\\
             x\in \{0,1\}^\omega~\mapsto~x \tnorm{ if } |x|_0= \infty \\
             u01^\omega~\mapsto~ u10^\omega \tnorm{ for }u \in \{0,1\}^*}$}}};

             \draw[darkgray,dotted,very thick] (0.8,0.2) -- (4.05,0.2);
             \fill[darkgray] {(0.8,0.2) circle (0.05)};
             \node[right] at (4.1,0.2) {\normalsize{\textcolor{darkgray}{$%
             \substack{0^{n_1} {a_1} 0^{n_2} {a_2} \cdots~\mapsto~%
             ({a_1})^{n_1} ({a_2})^{n_2} \cdots \\
            \tnorm{with } n_i \in \Nat ~\text{ and }~a_i \in \{1,2\}}$}}};

             \draw[darkgray,dotted,very thick] (0.3,-0.5) -- (4.05,-0.9);
             \fill[darkgray] {(0.3,-0.5) circle (0.05)};
             \node[right] at (4.1,-0.9) {\normalsize {\textcolor{darkgray}{$%
             \substack{\tnorm{Division by } 3 \tnorm{ in base } 2}$}}};

             
        \end{tikzpicture}
        \end{center}
        
    \caption{\label{fig:conclu} Classes of partial functions over infinite words studied in this paper}
\end{figure}
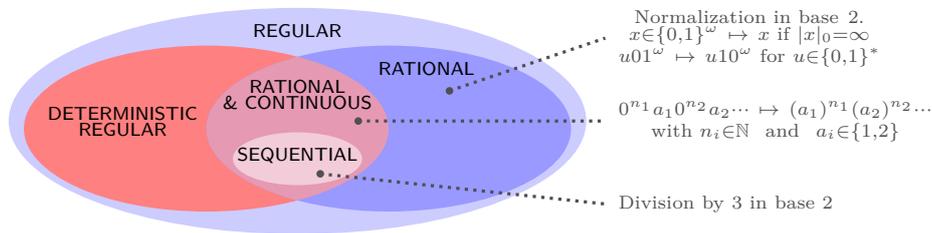

Then we introduce various new concepts in order
to transform a  \oNT{} computing a continuous function
into a \DSST{}. This determinization procedure is rather involved.
The main difficulty is that even if the \oNT{}
 is unambiguous, it might not check its guesses
after reading only a finite number of letters. In other words,
a given input can label several infinite runs, even if only one
of them is accepting. However, a deterministic machine
can never determine which run is the accepting one,
since it requires to check whether a property occurs infinitely often.
This intuition motivates our key definition of \emph{compatible sets}
among the states of a \oNT{}. Such sets are the sets of states
which have a ``common infinite future''.
The restriction of \oNT{} considered in~\cite{filiot2021synthesizing}
leads to compatible sets which are always singletons
(hence their condition defines a natural special case).
We show that when the function computed by the \oNT{} is continuous,
the outputs produced along finite runs which end
in a compatible set enjoy several combinatorial properties.

We finally describe how to build a \DSST{} which realizes
the continuous function given by a \oNT{}. Its construction is
is rather complex, and it crucially relies on the aforementioned
properties of compatible sets. These sets are manipulated by
the \DSST{} in an original tree-like fashion.
To the knowledge of the authors, this construction of this \DSST{}
is completely new (in particular, it is not based on the constructions of~\cite{dave2020synthesis} 
nor of~\cite{filiot2021synthesizing}).

\subparagraph*{Outline.}
We recall in Section~\ref{sec:rational}
the definitions of rational functions and one-way transducers.
In Section~\ref{sec:det-reg}, we present the new class
of deterministic regular functions and give 
the various transducer models which capture it.
Our main result which relates continuous rational
and deterministic regular functions is given in Section~\ref{sec:main}.
The proof is sketched in sections~\ref{sec:main}~and~\ref{sec:invariants}.

\section{Rational functions}

\label{sec:rational}

Letters $A,B$ denote alphabets, i.e. finite
sets of letters. The set $A^*$ (resp. $A^+$, $A^\omega$)
denotes the set of finite words (resp. non-empty finite words, infinite words)
over the alphabet $A$. If $u \in A^* \cup A^\omega$,
we let $|u| \in \Nat \cup \{\infty\}$ be its length. For $a \in A$,
$|u|_a$ denotes the number of $a$ in $u$.
For $1 \le i \le |u|$, $u[i] \in A$ is the $i$-th letter of  $u$.
If $1 \le i \le j \le |u|$, $u[i{:}j]$ stands for $u[i] u[i{+}1] \cdots$ until $j$.
We write $u[i{:}]$  for $u[i{:}|u|]$.
If $j > |u|$ we let $u[i{:}j] \defined u[i{:}|u|]$.
If $j < i$ we let $u[i{:}j] \defined \movi$ .
We write $u \pref v$ (resp. $u \prefneq v$)
when $u$ is a (resp. strict) prefix of $v$.
Given two words $u, v$, we let
$u \land v$ be their longest common prefix.
We say that $u,v$ are \emph{mutual prefixes} if $u \pref v$
or $v \pref u$. In this case we let $u \lor v$ be the longest
of them. A function $f$ between two sets $S,T$ is denoted
by $f: S \fonc T$. If $f$ is a \emph{partial function} (i.e. possibly with
non-total domain), it is denoted $f: S \parfonc T$.
Its domain is denoted $\Dom{f}$.

\begin{definition} A \emph{one-way non-deterministic transducer} \emph{(\oNT)}
\mbox{$\trans  = (A,B,Q, I, F, \Delta, \lambda)$ is:}
\begin{itemize}
\item a  finite input (respectively output) alphabet $A$ (respectively $B$);
\item a finite set of states $Q$  with $I \subseteq Q$ initial 
and $F \subseteq Q$ final;
\item a transition relation $\Delta  \subseteq Q \times A \times Q$;
\item an output function $\lambda: \Delta \rightarrow B^*$ (defined
for each transition).
\end{itemize}
\end{definition}
We write $q\runs{a | \alpha} q'$ whenever $(q,a,q') \in \Delta$
and $\lambda(q,a,q') = \alpha$. A \emph{run} labelled by some
$x \in A^* \cup A^\omega$ is a sequence of consecutive transitions
$\rho \defined q_0\runs{x[1] | \alpha_1}
q_1 \runs{x[2] | \alpha_2} q_2 \cdots $. The \emph{output} of $\rho$
is the word $\alpha_1 \alpha_2 \cdots \in A^* \cup A^\omega$.
If $x \in A^\omega$, we also write $q_0 \runs{x | \alpha_1 \alpha_2 \cdots} \infty$
to denote an infinite run starting in $q_0$. The
run $\rho$ is \emph{initial} if $q_0 \in I$,
\emph{final} if $x \in A^\omega$ and $q_i \in F$ infinitely 
often (Büchi condition), and \emph{accepting}
if both initial and final.
$\trans$ computes the \emph{relation} $\{(x,y):y \in B^\omega \text{ is output along
an accepting run on } x\}$. It is \emph{functional} 
if this relation is a (partial) function.
In this case, $\trans$ can be transformed
in an equivalent \emph{unambiguous} \oNT{}
(a transducer which has at most one
accepting run on each $x \in A^\omega$)~\cite[Corollary~3]{choffrut1999uniformization}.
A function $f:A^\omega \parfonc B^\omega$  is said to be \emph{rational}
if it can be computed by a (unambiguous) \oNT{}.

\begin{example}\label{ex:rational}
In Figure~\ref{fig:ex-rational}, we describe
\oNT s which compute the following functions:
\begin{itemize}
\item  $\exinf: \{0,1\}^\omega \parfonc \{0,1\}^\omega$ mapping
$x \mapsto x$ if $|x|_{0} = \infty$ and $u01^\omega \mapsto  u10^\omega$ if $u \in \{0,1\}^*$;

\item  $\replace: \{0,1,2\}^\omega \parfonc \{1,2\}^\omega$
with $\Dom{\replace} = \{x: |x|_{1} = \infty \tnorm{ or } |x|_{2} = \infty\}$
and mapping $ 0^{n_1} {a_1} 0^{n_2} {a_2} \cdots  \mapsto
{a_1}^{n_1+1} {a_2}^{n_2+1} \cdots$ if $a_i \in \{1,2\}$, $n_i \in \Nat$;

\item  $\double: \{0,1,2\}^\omega \fonc \{0, 1,2\}^\omega$ mapping
$0^{n_1} {a_1} 0^{n_2} {a_2}\cdots \mapsto
0^{a_1  n_1} {a_1} 0^{a_2  n_2} a_2 \cdots $ and \linebreak
$0^{n_1} {a_1} \cdots 0^{n_m} {a_m} 0^\omega \mapsto
0^{a_1  n_1} {a_1} \cdots 0^{a_m  n_m} {a_m} 0^\omega $
(if finitely many $1$ or $2$).
\end{itemize}
\end{example}

\begin{figure}[h!]
\centering
	\begin{subfigure}[b]{0.32\textwidth}
	\centering
		\begin{tikzpicture}{scale=0.8}
			\node [state, initial, initial text=, accepting, inner sep=3pt, minimum size=0pt]
			(q0) at (0,0) {\small $\substack{q_0}$};
			\node[state, accepting ,  inner sep=3pt, minimum size=0pt ] at (2,0.5) (q1) {$\substack{q_1}$};
			\node[state,  initial,initial right, initial text=, inner sep=3pt, minimum size=0pt ]%
			at (2,-0.5) (q2) {$\substack{q_2}$};
			\draw[->] (q0) edge[bend left = 30] node[above]{$\substack{0 | 1}$} (q1);
			\draw[->] (q0) edge[bend left = 7] node[above]{$\substack{0 | 0}$} (q2)
			        	       (q2) edge[bend left = 30] node[below]{$\substack{1 | 1}$} (q0);
			\draw[->] (q1) edge[loop above] node{$\substack{1|0}$} (q1);		
			\draw[->]      (q2) edge[loop below] node{$\substack{1|1}$} (q2);
			\draw[->]      (q0) edge[loop above] node{$\substack{0|0}$} (q0);
		\end{tikzpicture}
	\caption{\oNT{} computing $\exinf$}
	\end{subfigure}
	\begin{subfigure}[b]{0.32\textwidth}
	\centering
		\begin{tikzpicture}{scale=0.8}
			\node [state, initial, initial text=, accepting, inner sep=3pt, minimum size=0pt]
			(q0) at (0,0) {\small $\substack{q_0}$};
			\node[state,  inner sep=3pt, minimum size=0pt ] at (2,0.5) (q1) {$\substack{q_1}$};
			\node[state,  inner sep=3pt, minimum size=0pt ] at (2,-0.5) (q2) {$\substack{q_2}$};
			\draw[->] (q0) edge[bend left = 30] node[above]{$\substack{0 | 1}$} (q1)
			        	       (q1) edge[bend left = 7] node[above]{$\substack{1 | 1}$} (q0);
			\draw[->] (q0) edge[bend left = 7] node[below]{$\substack{0 | 2}$} (q2)
			        	       (q2) edge[bend left = 30] node[below]{$\substack{2 | 2}$} (q0);
			\draw[->] (q1) edge[loop above] node{$\substack{0|1}$} (q1);		
			\draw[->]      (q2) edge[loop below] node{$\substack{0|2}$} (q2);
			\draw[->]      (q0) edge[loop above] node{$\substack{1|1 \\ 2 |2}$} (q0);
		\end{tikzpicture}
	\caption{\label{fig:replace}\oNT{} computing $\replace$}
	\end{subfigure}
	\begin{subfigure}[b]{0.32\textwidth}
	\centering
		\begin{tikzpicture}{scale=0.8}
			\node [state, initial, initial text=, accepting, inner sep=3pt, minimum size=0pt]
			(q0) at (0,0) {\small $\substack{q_0}$};
			\node[state,accepting,  inner sep=3pt, minimum size=0pt ] at (2,0.5) (q1) {$\substack{q_1}$};
			\node[state,  inner sep=3pt, minimum size=0pt ] at (2,-0.5) (q2) {$\substack{q_2}$};
			\draw[->] (q0) edge[bend left = 30] node[above]{$\substack{0 | 0}$} (q1)
			        	       (q1) edge[bend left = 7] node[above]{$\substack{1 | 1}$} (q0);
			\draw[->] (q0) edge[bend left = 7] node[below]{$\substack{0 | 00}$} (q2)
			        	       (q2) edge[bend left = 30] node[below]{$\substack{2 | 2}$} (q0);
			\draw[->] (q1) edge[loop above] node{$\substack{0|0}$} (q1);		
			\draw[->]      (q2) edge[loop below] node{$\substack{0|00}$} (q2);
			\draw[->]      (q0) edge[loop above] node{$\substack{1|1 \\ 2 |2}$} (q0);
		\end{tikzpicture}
	\caption{\label{fig:double}\oNT{} computing $\double$}
	\end{subfigure}
\caption{\label{fig:ex-rational} Unambiguous, clean and trim \oNT s
computing the functions of Example~\ref{ex:rational}.}
\end{figure}
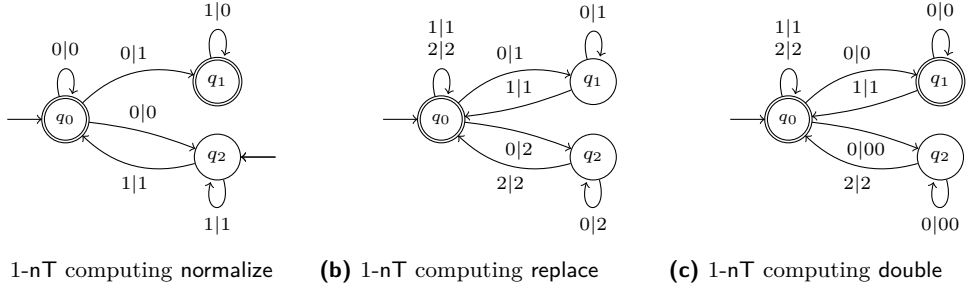%

\begin{remark} \label{rem:not-sequential}
The functions mentioned in Example~\ref{ex:rational}
are not \emph{sequential}, i.e. they cannot be computed
by deterministic one-way transducers (i.e. deterministic \oNT{}s).
\end{remark}

A \oNT{} is \emph{trim} if any state is both accessible
and co-accessible, or equivalently if it occurs in some accepting run.
It is \emph{clean} if the production along
any accepting run is infinite.

\begin{lemma} \label{lem:infinite-accepting}
A trim \oNT{} is clean if and only if for all $q \in F$,
the existence of a cycle $q\runs{u | \alpha} q $ for 
$u \in A^+$ implies $\alpha \neq \movi$.
Given an unambiguous \oNT, one can build an equivalent unambiguous,
\mbox{clean and trim \oNT{}.}
\end{lemma}

\section{Deterministic regular functions}

\label{sec:det-reg}

We now introduce the new class of  \emph{deterministic regular} functions,
which are computed by  \emph{deterministic two-way transducers}.
Contrary to \oNT s, such machines
cannot test $\omega$-regular properties of their input.
Hence they describe continuous (and computable) functions.

\begin{definition}[] \label{def:two-way}

	A \emph{deterministic two-way transducer} \emph{(\tDT)}
	$\trans = (A,B,Q,q_0,\delta,\lambda)$ is:
	\item 
	\begin{itemize}
	\item an input alphabet $A$ and an output alphabet $B$;
	\item a finite set of states $Q$
	with an initial state $q_0 \in Q$;
	\item a transition function $\delta: Q \times (A \uplus \{{\lmark}\})
	\parfonc Q \times \{\lmove, \rmove\}$;
	\item an output function $\lambda: Q \times (A \uplus \{{\lmark}\}) \parfonc B^*$
	with same domain as $\delta$.

	\end{itemize}

\end{definition}
If the input is $x \in A^\omega$, then $\trans$
is given as input the word ${\lmark} x$.
The symbol ${\lmark}$ is used to mark the beginning
of the input. We denote by $x[0] \defined{} \lmark$.
A \emph{configuration} over  ${\lmark} x$ is a tuple
$(q,i)$ where $q \in Q$ is the current state and
$i \ge 0$ is the current position of the reading head.
The \emph{transition relation} $\rightarrow$ is defined as follows.
Given a configuration $(q,i)$, let $(q',\star):= \delta(q,w[i])$.
Then $(q, i) \rightarrow (q', i')$ whenever either $\star = \lmove$
and  $i' = i-1$ (move left), or $\star = \rmove$ and $i' = i+1$ (move right).
A \emph{run} is a (finite or infinite) sequence of
configurations $(q_1,i_1) \rightarrow (q_2,i_2) \rightarrow \cdots$.
An \emph{accepting} run is an infinite run which starts in $(q_0,0)$
and such that $i_n \rightarrow \infty$ when $n \rightarrow \infty$
(otherwise the transducer repeats the same loop).

The partial function $f: A^\omega \parfonc B^\omega$ computed by
$\trans$ is defined as follows. Let $x \in A^\omega$ be such
that  there exists a (unique) accepting run
$(q^x_0,i^x_0) \rightarrow (q^x_1,i^x_1) \rightarrow \cdots$
labelled by $x$. Let $y \defined \prod_{j=1}^{\infty} \lambda(q^x_j, w[i^x_j]) \in B^* \cup B^\omega$
be the concatenation of the outputs produced along this run.
If $y \in B^\omega$, we define $f(x) \defined y$.
Otherwise $f(x)$ is undefined.

\begin{example} The function $\replace$ from Example~\ref{ex:rational}
can be computed by \tDT{}. For each $i \ge 1$, this \tDT{} crosses the block $0^{n_i}$
to determines ${a_i}$, and then crosses the block once more
and outputs ${a_i}^{n_i+1}$. The function $\double$ can be computed
using similar ideas. However, an important difference is that
the \tDT{} must output  the block $0^{n_i}$ when it crosses it for the first time, in order
to ensure that the production over $0^\omega$ is $0^\omega$.
\end{example}
There exists deterministic regular functions
which are not rational, for instance the function
which reverses (mirror image) a prefix of its input.

Over finite words, it is known that two-way transducers
are equivalent to copyless streaming string transducers
\cite{alur2010expressiveness}.
Over infinite words, a similar equivalence holds between
two-way transducers with lookahead and
copyless streaming string transducers with Müller output conditions~\cite{alur2012regular}.
These models define the class of \emph{regular functions}
over infinite words. However, lookaheads enable
two-way transducers to check $\omega$-regular properties
of their input (and thus non-computable behaviors).
Hence our deterministic regular functions form a strict subclass
of these regular functions over infinite words.

We now introduce a model of streaming string transducer
to describe deterministic regular functions, in the spirit of the
aforementioned results.
In our setting, it consists of a one-way deterministic
automaton with a finite set $\Regs$ of registers that store words
from $B^*$. We use a distinguished register $\out$ to
store the output produced when reading an infinite word.
The registers are modified using \emph{substitutions},
i.e. mappings $\Regs \rightarrow (B \uplus \Regs)^*$.  We denote by
$\subst{\Regs}{B}$ the set of these substitutions.  They can be extended
morphically from $(B \uplus \Regs)^*$ to $ (B \uplus \Regs)^*$ by preserving the
elements of $B$.  They can be composed (see~Example~\ref{ex:subst}).

\begin{example} \label{ex:subst} Let $\Regs = \{\reg,\regg\}$ and $B = \{b\}$.
Consider $\sigma_1:= \reg \mapsto b, \regg \mapsto b\reg \regg b$
and $s_2:= \reg \mapsto \reg b, \regg \mapsto \reg \regg$,
then $ \sigma_1 \circ \sigma_2 (\reg) = s_1(\reg b) = bb$ and $ \sigma_1 \circ \sigma_2 (\regg)
= \sigma_1(\reg \regg)
= bb\reg \regg b$.
\end{example}

\begin{definition} A \emph{deterministic streaming string transducer}
\emph{(\DSST)} is:
\begin{itemize}
\item a finite input (resp. output) alphabet $A$ (resp. $B$);
\item a finite set of states $Q$ with $q_0 \in Q$ initial;
\item a transition function $\delta: Q \times A \parfonc Q$;
\item a finite set of registers $\Regs$ with a distinguished output register $\out \in \Regs$;
\item an update function $\lambda: Q \times A \parfonc \subst{\Regs}{B}$
such that for all $(q,a) \in \Dom{\lambda} = \Dom{\delta}$:
\begin{itemize}
\item $\lambda(q,a)(\out) = \out \cdots$;
\item there is no other occurence of $\out$ in $\{ \lambda(q,a)(\reg): \reg \in \Regs\}$.
\end{itemize}
\end{itemize}
We denote it $\trans = (A,B, Q,q_0, \delta, \Regs, \out, \lambda)$.
\end{definition}
This machine defines a function $f: A^* \parfonc B^*$ as follows.
For $i \ge 0$ let $q^x_i:= \delta(q_0,x[1{:}i])$ (when defined).
For $i \ge 1$, we let $\lambda^x_i \defined \lambda(q^x_{i-1},x[i])$ (when defined)
and $\lambda^x_0(\reg) = \movi$ for all $\reg \in \Regs$.
For $i \ge 0$, define the substitution $\cro{\cdot}^x_i \defined \lambda^x_0 \circ \cdots
\circ \lambda^x_i$.
By construction we get $\cro{\out}^x_i \pref \cro{\out}^x_{i+1}$
(when defined).
If $\cro{\out}^x_i$ is defined for all $i \ge 0$
and $| \cro{\out}^x_i| \rightarrow + \infty$,
we let $f(x) \defined  \bigvee_i \cro{\out}^x_i$
(it denotes the unique infinite word $y$
such that $\cro{\out}^x_i \pref y$ for all $i \ge 0$).
Otherwise $f(x)$ is undefined.

We say that a substitution $\sigma \in \subst{\Regs}{B}$
is \emph{copyless} (resp. $K$-bounded) if for all $\reg \in \Regs$,
$\reg$ occurs at most once in $\{ \sigma(\regg): \regg \in \Regs\}$
(resp. for all $\reg, \regg \in \Regs$, $\reg$ occurs at most
$K$ times in $\sigma(\regg)$). 

\begin{definition}[Copy restrictions]
We say that a \DSST{}  $\trans = (A,B, Q,q_0, \delta, \Regs, \out, \lambda)$ is
\emph{copyless} (resp. \emph{$K$-bounded}) if for all $x \in A^\omega$
and $i \le j$ such that $\lambda^x_i \circ \cdots \circ \lambda^x_j$
is defined, this substitution is copyless (resp. $K$-bounded).
\end{definition}

\begin{example}
The function $\replace$ from Example~\ref{ex:rational}
can be computed by a copyless \DSST{}. For all $i \ge 1$, it crosses the block $0^{n_i}$
and computes $1^{n_i}$ and $2^{n_i}$ in two registers.
Once it sees ${a_i}$ it adds in $\out$ the register storing  ${a_i}^{n_i}$.
The function $\double$ can be computed
using similar ideas. However, an important difference is that
the \DSST{} must directly output the block $0^{n_i}$ while crossing it,
in order to ensure that the production over $0^\omega$ is $0^\omega$.
\end{example}
The proof of the next result is quite involved,
but it is largely inspired by the techniques used for regular functions
over finite or infinite words (see e.g.~\cite{dartois2016aperiodic,doueneau2020register}).

\begin{theorem} \label{theo:2dt-dsst} The following machines compute
the same class of functions $A^\omega \parfonc B^\omega$:
\begin{enumerate}
\item deterministic two-way  transducers \emph{(\tDT{})};
\item $K$-bounded deterministic streaming string transducers \emph{(}$K$-bounded \DSST\emph{)};
\item copyless deterministic streaming string transducers \emph{(}copyless \DSST\emph{)}.
\end{enumerate}
Furthermore, all the conversions are effective.
\end{theorem}

\begin{remark} Even if this result is a variant of existing
results over finite or infinite words, it requires a proof on its own.
Indeed, the authors are not aware of a direct proof which would
enable to deduce it from the existing similar results.
\end{remark}

Let us now describe the domains
of deterministic regular functions.
We say that a language is \emph{Büchi deterministic}
if  it is accepted by a deterministic Büchi automaton~\cite{perrin2004}.

\begin{proposition} \label{prop:dom-reg} If $f$ is deterministic
regular, then $\Dom{f}$ is Büchi deterministic.
\end{proposition}

We finally give a closure property of deterministic regular functions
under pre-composition.

\begin{definition}
A \emph{restricted} \oNT{} is a \oNT{} whose states all are final.
\end{definition}
The semantics of a restricted \oNT{} $\ntrans = (A,B,Q,I, \Delta, \lambda)$
is defined so that it  \emph{always} computes 
a  function $f: A^\omega \parfonc B^\omega$.
The domain $\Dom{f}$ is the set of $x \in A^\omega$ such that
$\ntrans$ has a unique accepting run labelled by $x$,
and such that the output along this unique run is infinite. In this case,
we let $f(x)$ be the output of along this run.
Intuitively, such a transducer expresses the ability to make
non-deterministic guesses, as long as these guesses
can be verified after reading a finite number of letters
(i.e. there are no two possible infinite runs).

\begin{theorem} \label{theo:prec-reg}
Given a restricted \oNT{} computing a function $f: A^\omega \parfonc B^\omega$
and a deterministic regular function $g: B^\omega \parfonc C^\omega$,
$g \circ f$ is (effectively) deterministic regular.
\end{theorem}

\section{Continuous rational functions are deterministic regular}

\label{sec:main}

We now state the main result of this paper, which shows that
a rational function  can be extended to a deterministic regular function.
Using an extension of the original function is necessary
since  not all $\omega$-regular languages
are Büchi deterministic (see \cref{prop:dom-reg}). Note that \cref{theo:main}
is in fact an equivalence,
in the sense that a rational function which can be
extended to a deterministic regular function is
obviously  continuous.

We recall that a function $f: A^\omega \parfonc B^\omega$ is
\emph{continuous} if and only if for all $x \in \Dom{f}$
and $ n \ge 0$, there exists $ p \ge 0$ such that $\forall y \in \Dom{f}$,
$|x \wedge y| \ge p \Rightarrow |f(x) \wedge f(y)| \ge n$.

\begin{example} The functions $\replace$ and $\double$
are continuous, but $\exinf$ is not.
\end{example}

\begin{theorem} \label{theo:main} Given a continuous rational
function $f: A^\omega \parfonc B^\omega$, one can
build  a deterministic regular function $f'$ which extends $f$
(i.e. for all $x \in \Dom{f}$, $f(x) = f'(x)$).
\end{theorem}
To prove \cref{theo:main}, it is enough
by theorems~\ref{theo:2dt-dsst} and~\ref{theo:prec-reg} to show that $f'$ can be computed
as a composition of a restricted \oNT{} 
and a $K$-bounded \DSST{}~(see Subsection~\ref{ssec:compo-f},
the construction will in fact give a $1$-bounded transducer).

\subsection{Properties of continuous rational functions}

\label{sec:struct-prop}

We first describe some structural properties
of \oNT{} computing continuous functions.
In this subsection, we let $\trans = (A,B,Q, I, F, \Delta, \lambda)$
be an unambiguous, clean and trim  \oNT{} computing a
continuous function $f: A^\omega \parfonc B^\omega$.
It is well known  that $\trans$ verifies
\cref{lem:continuity-loops}. This property is
in fact equivalent to the continuity of $f$
(see e.g.~\cite{prieur2001decide}
or~\cite{dave2020synthesis}).

\begin{lemma}
\label{lem:continuity-loops}
For all $q_1, q_2 \in I$, $q'_1 \in F$, $q'_2 \in Q$,
$u \in A^*$, $u' \in A^+$, $\alpha_1, \alpha'_1, \alpha_2, \alpha'_2 \in B^*$
 such that $q_i \runs{u | \alpha_i} q'_i \runs{u' | \alpha'_i} q'_i$
for $i \in \{1,2\}$ we have (note that $\alpha'_1 \neq \movi$ since $\trans$ is clean):
\begin{itemize}
\item if $\alpha_2' \neq \movi$, then $\alpha_1 {\alpha'_1}^\omega = \alpha_2 {\alpha'_2}^\omega$;
\item if $\alpha'_2 = \movi$, $x \in A^\omega$, $\beta \in B^\omega$
and $q'_2 \runs{x|\beta} \infty$ is final, then
$ \alpha_1 {\alpha'_1}^\omega  = \alpha_2 \beta$.
\end{itemize}
\end{lemma}
Empty cycles $q \runs{u  | \movi} q$ for $q \not \in F$ cannot be avoided
in a \oNT{}. However, we shall see in \cref{lem:make-productive} that
such cycles can be avoided if the function is continuous.
Formally, we say that the clean $\trans$ is \emph{productive}
if the hypotheses of \cref{lem:continuity-loops} imply $\alpha'_2 \neq \movi$.

\begin{lemma} \label{lem:make-productive}
Given $\trans$, one can build an equivalent unambiguous,
trim and~productive~\oNT{}.
\end{lemma}

\subparagraph*{Compatible sets and steps.}
We now introduce the key notion of a \emph{compatible set}
which is a set of states having a ``common
future'' and such that one of the future runs is accepting.

\begin{definition}[Compatible set] \label{def:compat}
We say that a set of states $C \subseteq Q$
is \emph{compatible} whenever there exists $x \in A^\omega$
and infinite runs $\rho_q$ for each $q \in C$ labelled by $x$ such that:
\begin{itemize}
\item $\forall q \in C$, $\rho_q$ starts from $q$;
\item $\exists q \in C$ such that $\rho_q$ is final.
\end{itemize}
\end{definition}
Let $\Comp$ be the set of compatible sets.
If $S \subseteq Q$, let $\Comp(S)$ be the
set $2^S \cap \Comp$.

\begin{definition}[Pre-step] We say that $C,u,D$ is a \emph{pre-step}
if $C,D \in \Comp$, $u \in A^*$ and
for all $q \in D$, there exists a unique
state $\pre{C,D}{u}(q) \in C$ such that 
$\pre{C,D}{u}(q)  \runs{u} q$.
\end{definition}
Note that for all $D' \in \Comp(D)$, we have $\pre{C,D}{u}(D') \in \Comp$.

\begin{definition}[Step] We say that a pre-step $C,u,D$ is a \emph{step}
if $\pre{C,D}{u}$ is surjective.
\end{definition}
Given $q \in D$, let $\val{C,D}{u}(q)$ be the output
$\alpha \in B^*$ produced along the run
$\pre{C,D}{u}(q)  \runs{u | \alpha } q$.
We say that a (pre-)step is \emph{initial} whenever $C \subseteq I$.
We first claim that the productions
along the runs of an initial step are
mutual prefixes. \cref{lem:continuity-loops} 
is crucial here.

\begin{lemma} \label{lem:mutual}
Let $J,u,C$ be an initial step.
Then $\val{J,C}{u}(q)$ for $q \in C$
are mutual prefixes.
\end{lemma}

\begin{example} In Figure~\ref{fig:replace}, if a step is initial,
it is of the form $\{q_0\}, u, \{q_i\}$ for some $i \in \{0,1,2\}$.
In Figure~\ref{fig:double}, $\{q_0\}, 0^n, \{q_1,q_2\}$ is an initial
step for all $n \ge 0$.
 \end{example}

\begin{definition}[Common, advance] 
\label{def:del} Let $J,u,C$ be an initial step. We define:
\begin{itemize}
\item the \emph{common} $\com{J,C}{u}  \in B^*$ as the longest common prefix
$\bigwedge_{q \in C} \val{J,C}{u}(q)$;
\item for all $q \in C$, its \emph{advance} $\adv{J,C}{u}(q) \in B^*$
as $(\com{J,C}{u})^{-1} \val{J,C}{u}(q)$;
\item the \emph{maximal advance} $\advm{J,C}{u}$ 
as the longest advance, i.e. $\bigvee_{q \in C} \adv{J,C}{u}(q)$.
\end{itemize}
\end{definition}
Definition~\ref{def:del} makes sense by \cref{lem:mutual},
and furthermore $\val{J,C}{u}(q) = \com{J,C}{u} \adv{J,C}{u}(q)$ for all $q \in C$.
Now let $M \defined \max(10,\max_{q,q' \in Q, a \in A} |\lambda(q,a,q')|)$
and $\Bound \defined M \bound$. We say that a compatible
set $C$ is \emph{separable} if there exists an initial
step which ends in $C$, and such that the lengths of the productions along two of its
runs differ of at least $\Bound$.

\begin{definition}[Separable set] \label{def:separable-set}
Let $C \in \Comp$, we say that $C$ is \emph{separable}
if there exists an initial step $J,u,C$ and $p,q \in C$ such that
$\left||\adv{J,C}{u}(p)| - |\adv{J,C}{u}(q)|\right| > \Bound$.
\end{definition}

\begin{remark}
In other words, it means that $|\advm{J,C}{u}| > \Bound$.
\end{remark}
It is easy to see (by a pumping argument)
that one can decide if a set is separable.
We now show that the productions
along the initial steps which end in a separable
set are forced to ``iterate'' some value $\theta$
if the step is pursued. The following lemma
is the key ingredient for showing that a rational
function is deterministic regular (see Section~\ref{sec:main}).

\old{
\begin{lemma}[Looping futures] \label{lem:sep-theta} Let  $C \in \Comp$  be separable
and $J,u,C$ be an initial step
(not necessarily the one which makes $C$ separable). There exists 
$\tau, \theta \in B^*$ with $|\tau| \le 3\Bound$,
and $|\theta| = \Bound!$ which can be uniquely determined
from $C$  and $\adv{J,C}{u}(q)$ for $q \in C$,
such that:
\begin{itemize}
\item $\tau \pref \advm{J,C}{u} \pref \tau  \theta^\omega$;
\item for all step $C,v,D$ and $q \in D$,
$\val{C,D}{v}(q) \pref  (\adv{J,C}{u}(p))^{-1} \tau  \theta^\omega$
with $p \defined \pre{C,D}{v}(q)$.
\end{itemize}
\end{lemma}}%
\cor{%
\begin{lemma}[Looping futures] \label{lem:sep-theta} Let  $C \in \Comp$  be separable
and $J,u,C$ be an initial step
(not necessarily the one which makes $C$ separable). There exists 
$\tau, \theta \in B^*$ with $|\tau| \le \Bound!$,
and $|\theta| = \Bound!$ which can be uniquely determined
from $C$  and $\adv{J,C}{u}(q)$ for $q \in C$,
such that:
\begin{itemize}
\item $ \advm{J,C}{u} \pref \tau  \theta^\omega$;
\item for all step $C,v,D$ and $q \in D$,
$\val{C,D}{v}(q) \pref  (\adv{J,C}{u}(p))^{-1} \tau  \theta^\omega$
with $p \defined \pre{C,D}{v}(q)$.
\end{itemize}
\end{lemma}
}%

\begin{remark}\label{rem:sep-theta}
Since $\adv{J,C}{u}(p) \pref \advm{J,C}{u} \pref \tau \theta^\omega$,
the second item makes sense.
\end{remark}

\begin{example} In Figure~\ref{fig:double},
the compatible set $C \defined \{q_1, q_2\}$ is separable.
For all step $C, v, D$ we have $D = C$ thus $v = 0^n$,
$\val{C,D}{0^n}(q_1) = 0^n$ and $\val{C,D}{0^n}(q_2) = 0^{2n}$.
 \end{example}

\subsection{Composition of a restricted \oNT{} and a $1$-bounded \DSST{}}

\label{ssec:compo-f}

In the rest of this paper, we let $\trans = (A,B, Q, I,F,\Delta,\lambda)$
be an unambiguous, productive and trim \oNT{}
computing a continuous $f: A^\omega \parfonc B^\omega$.
Our goal is to rewrite $f$ as the composition
of a restricted \oNT{}  and a $1$-bounded \DSST{}.
We first build the restricted \oNT{}, which computes an over-approximation
of the accepting run of $\trans$ in terms of compatible sets.

\begin{lemma} \label{lem:pre-compat}
One can build a restricted \oNT{} $\ntrans$ computing
$g:A^\omega \parfonc (\Comp  \uplus A)^\omega$ such that
$\Dom{f} \subseteq \Dom{g}$, and for all $x \in \Dom{g}$,
$g(x) = C_0 x[1] C_1  x[2] C_2 \cdots$
where:
\begin{itemize}
\item $C_0 \subseteq I$ and for all $i \ge 0$, $C_i, x[i{+}1],C_{i+1}$ is a pre-step;
\item if $x \in \Dom{f}$ then $\forall i\ge 0$, $q^x_i \in C_i$,
where $q^x_0 \runs{x[1]} q^x_1 \runs{x[2]} \cdots$ is the
accepting run of $\trans$.
\end{itemize}
\end{lemma}

Given $x \in \Dom{f}$, we denote by $C^x_0, C^x_1, \dots$ the sequence
of compatible sets produced by $\ntrans$ in \cref{lem:pre-compat}.
We now describe a $1$-bounded \DSST{} $\strans$ which, when given as
input $g(x) \in (\Comp  \uplus A)^\omega$ for $x \in \Dom{f}$, outputs $f(x)$
(this description is continued in Section~\ref{sec:invariants}).

\subparagraph*{Tree of compatibles.} Given $C \in \Comp$,
we define  $\tree{C}$ as a finite set of words over
$\Comp(C)$, which describes the decreasing chains for $\subsetneq$.
It can be identified with the set of all root-to-node paths of a tree
labelled by elements of $\Comp(C)$, as shown in Example~\ref{ex:tree-compat}.

\begin{definition}[Tree of compatibles] Given $C \in \Comp$, we 
denote by $\tree{C}$ the set of words $\pi = C_1 \cdots C_n \in (\Comp(C))^+$
such that $C_1 = C$ and for all $1 \le i \le n{-}1$,  $C_i \supsetneq C_{i+1}$.
\end{definition}

\begin{example} \label{ex:tree-compat} 
If $C = \{1,2,3\}$ and $\Comp(C) = \{\{1,2,3\}, \{1,2\}, \{2,3\}, \{1\}, \{2\}, \{3\}\}$,
then we have $\tree{C} = \{\{1,2,3\}\{1,2\}\{1\}, \{1,2,3\}\{1,2\}\{2\}, \{1,2,3\}\{2,3\}\{2\},  \{1,2,3\}\{2,3\}\{3\},$
\linebreak
$ \{1,2,3\}\{1\}, \{1,2,3\}\{2\}, \{1,2,3\}\{3\} \}$. Its view as a tree is depicted in Figure \ref{fig:tree-compat}.
\end{example}
\begin{figure}[h!]

\centering
\begin{tikzpicture}{scale=1}

	\node[above] at (4,5.2) {$\{1,2,3\}$};	
	\draw (4,5) -- (4,5.2);		
	\draw (0,5) -- (8,5);
		\draw (0,4.8) -- (0,5);
		\node[above] at (0,4.3) {$\{1,2\}$};	
		\draw (-0.5,4.1) -- (0.5,4.1);
		\draw (0,4.1) -- (0,4.3);	
			\draw (-0.5,3.9) -- (-0.5,4.1);
			\node[above] at (-0.5,3.4) {$\{1\}$};	
			\draw (0.5,3.9) -- (0.5,4.1);
			\node[above] at (0.5,3.4) {$\{2\}$};	
		\draw (2,4.8) -- (2,5);
		\node[above] at (2,4.3) {$\{2,3\}$};
		\draw (1.5,4.1) -- (2.5,4.1);	
		\draw (2,4.1) -- (2,4.3);	
			\draw (1.5,3.9) -- (1.5,4.1);
			\node[above] at (1.5,3.4) {$\{2\}$};	
			\draw (2.5,3.9) -- (2.5,4.1);
			\node[above] at (2.5,3.4) {$\{3\}$};	
		\draw (4,4.8) -- (4,5);
		\node[above] at (4,4.3) {$\{1\}$};	
		\draw (6,4.8) -- (6,5);
		\node[above] at (6,4.3) {$\{2\}$};	
		\draw (8,4.8) -- (8,5);
		\node[above] at (8,4.3) {$\{3\}$};	
		
\end{tikzpicture}

\caption{\label{fig:tree-compat}  The tree of compatibles obtained from Example~\ref{ex:tree-compat}.}

\end{figure}
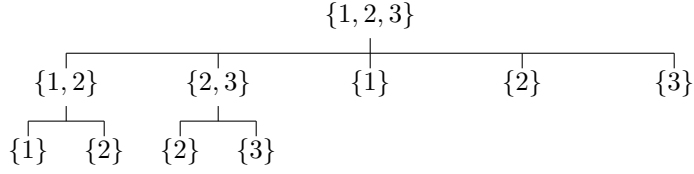

\subparagraph*{Information stored.}  The states of the \DSST{} $\strans$
are partitioned in two categories: the sets of the \emph{separable mode}
and the sets of of the \emph{non-separable mode}.
A configuration of the \DSST{} $\strans$ will always keep track
of the following information:
\begin{itemize}
\item the content of a register $\out$;
\item two sets $\Jf, \Cf \in \Comp$ and a function $\pre{}{}: \Cf \fonc \Jf$ (stored in the state);
\old{\item a function $\lag: \Cf \fonc B^*$ such that $|\lag(q)| \le 3\Bound$
for all $q \in \Cf$ (stored in the state).}
\cor{
\item a function $\lag: \Cf \fonc B^*$ such that $|\lag(q)| \le \Bound!$
for all $q \in \Cf$ (stored in the state);
\item a value $\lagm \in B^*$ such that $|\lagm| \le \Bound!$ (stored in the state).}
\end{itemize}
Furthermore, when $\strans$ is in a state of the separable mode, it will additionally store:
\begin{itemize}
\item a value $\theta \in B^*$ with $|\theta| = \Bound!$ (stored in the state);
\item for all $\pi = C_1 \cdots C_n \in \tree{\Cf}$ (note that $C_1 = \Cf$ by definition of  $\tree{\Cf}$):
\begin{itemize}
\item\old{a function $\nb{\pi}: C_n \fonc [0{:}\four]$ (stored in the state);}%
\cor{%
a function $\nb{\pi}: C_n \fonc [0{:}\four]$ (stored in the state);
}
\item the content of a register $\outi{\pi}$. For $\pi = \Cf$, we identify the register $\outi{\Cf}$ with $\out$;
\end{itemize}
\item a function $\last: \Cf \fonc B^*$ such that $|\last(q)| < \Bound!$
forall $q \in \Cf$ (stored in the state).
\end{itemize}
If a configuration of $\strans$ is clearly fixed, we abuse notations and denote by $\outi{\pi}$
(resp. $\nb{\pi}$, $\lag$, etc.) the value contained in register $\out_\pi$
 (resp. stored in the state) in this configuration.
In a given configuration of $\strans$, we say that $\pi \in \tree{\Cf}$ is  \emph{close} if
for all $\pi \prefneq \pi' \in \tree{\Cf}$, we have 
$\nb{\pi'} = 0$ and $\outi{\pi'} = \movi$
(intuitively, the subtree rooted in $\pi$ stores empty informations).

\subparagraph*{Invariants.} The main idea for building $\strans$ it the following. If $C^x_i$ is a non-separable
set, then the productions along the initial runs which end in $C^x_i$ are mutual
prefixes (by \cref{lem:mutual}) which only differ
from a bounded information. Hence the common part  $\com{}{}$
of these runs is stored to $\out$, and the $\adv{}{}$ are stored in the $\lag$.
If $C^x_i$ becomes separable, then these runs still produce mutual prefixes,
but two of them can differ by a large information. However by \cref{lem:sep-theta},
they iterate some value $\theta$. Hence the only
relevant information is the number of $\theta$ which were produced
along these runs.
Formally, our construction ensures that the following
invariants hold when $\strans$ has just read
$C^x_0 x[1] C^x_1 \cdots x[i] C^x_i$ for $i \ge 0$:
\begin{enumerate}
\item  \label{inv:C} $\Cf = C^x_i$;

\item \label{inv:step} $\Jf,x[1{:}i],\Cf$ is an initial step and $\pre{}{} = \pre{\Jf, \Cf}{}$

\item \label{inv:nosep} if $\Cf$ is not separable, then $\strans$ is in non-separable mode and:
\begin{enumerate}
\item $\out = \com{\Jf,\Cf}{x[1{:}i]}$;
\item $\lag(q) = \adv{\Jf,\Cf}{x[1{:}i]}(q)$ for all $q \in \Cf$.
\end{enumerate}

\item \label{inv:sep} if $\Cf$ is separable, then $\strans$ is in separable mode and:
\begin{enumerate}

\item\old{\label{inv:lag}the $\lag(q)$ for $q \in Q$ are mutual prefixes,
and so $\lagm \defined \bigvee_{q \in C} \lag(q)$ is defined.
Furthermore, there exists $q \in \Cf$ such that $\lag(q) = \movi$.
We say that some $q \in \Cf$ is \emph{lagging} if and only if $\lag(q) \prefneq \lagm$
(strict prefix), otherwise it is \emph{not lagging};}%
\cor{%
\label{inv:lag}for all $q \in Q$ $\lag(q) \pref \lagm$.
Furthermore, there exists $q \in \Cf$ such that $\lag(q) = \movi$.
We say that some $q \in \Cf$ is \emph{lagging} if and only if $\lag(q) \prefneq \lagm$
(strict prefix), otherwise we say that $q$ is \emph{not lagging};}

\item \label{inv:out} if $\pi \in \tree{\Cf}$ is such that $\pi \neq \Cf$ (i.e. $\outi{\pi} \neq \out$),
then $\outi{\pi} \in \theta^*$;

\item  \label{inv:last} for all $q \in \Cf$, $\last(q) \pref \theta^\omega$
(if furthermore $|\last(q)| < \Bound!$, then $\last(q) \prefneq \theta$);

\item \label{inv:lagging} if $q$ is lagging, then $\last(q) = \movi$
and  for all $\pi = C_1 \cdots C_n \in \tree{\Cf}$ such
that $q \in C_n$, we have $\nb{\pi}(q) = 0$ and, if $\pi \neq \Cf$, $\outi{\pi} = \movi$;

\item \label{inv:past} for all $\pi = C_1 \cdots C_n \in \tree{\Cf}$, 
 for $1 \le i \le n$  define $\pi_i \defined C_1 \cdots C_i$.
If $C_n = \{q\}$, then:
\begin{itemize}

\item $\val{\Jf,\Cf}{x[1{:}i]}(q) = \out~\lag(q)$ if $q$ is lagging;

\item $\val{\Jf,\Cf}{x[1{:}i]}(q) = \out~\lagm ~\theta^{\nb{\pi_1}(q)}
\left(\prod_{i=2}^n \outi{\pi_i} \theta^{\nb{\pi_i}(q)} \right) \last(q)$ 
 if $q$ is not lagging.
\end{itemize}

\item \label{inv:future} for all future steps $\Cf, u, D$ and for all $q \in D$,
$\val{\Jf,D}{x[1{:}i]u}(q)
\pref \out~ \lagm~\theta^\omega$;

%
\old{\item \label{inv:close}%
for all $\pi = C_1 \cdots C_n \in \tree{\Cf}$ not close,
let $J_n \defined \pre{\Jf,\Cf}{x[1{:}i]}(C_n) \subseteq \Jf$.
Then $J_n, x[1{:}i], C_n$ is an initial step,
which can be decomposed as an initial step
$J_n, x[1{:}j], E$  and a step $E, x[j{+}1{:}i], C_n$
such that $|\advm{J_n,E}{x[1{:}j]} | \ge \four \Bound !$.}%
\cor{%
\item \label{inv:close}
for all $\pi = C_1 \cdots C_n \in \tree{\Cf}$ not close,
let $J_n \defined \pre{\Jf,\Cf}{x[1{:}i]}(C_n) \subseteq \Jf$.
Then $J_n, x[1{:}i], C_n$ is an initial step,
which can be decomposed as an initial step
$J_n, x[1{:}j], E$  and a step $E, x[j{+}1{:}i], C_n$
such that $|\advm{J_n,E}{x[1{:}j]} | \ge \four \Bound !$.
}
\end{enumerate}
\end{enumerate}

\section{Description of the $1$-bounded \DSST{} for Subsection~\ref{ssec:compo-f}} 

\label{sec:invariants}

In this section, we finally describe how the \DSST{} $\strans$ can
preserve the invariants of Subsection~\ref{ssec:compo-f},
while being $1$-bounded and outputting $f(x)$ when $x \in \Dom{f}$.

Let us first deal with the initialization of $\strans$.
When reading the first letter $C^x_0$ of $g(x)$, $\strans$ stores
$\Jf \becomes C^x_0$, $\Cf \becomes C^x_0$
and $\lag(q) \becomes \movi$ for all $q \in C^x_0$.
\cor{There is no need to define $\lagm$ in this context.}
This is enough if  $C^x_0$ is not separable.
Otherwise, \old{we let $\theta$ be given by \cref{lem:sep-theta}}%
\cor{let $\theta$  (resp. $\lagm$) be given by the $\theta$ (resp. 
the $\tau$) of \cref{lem:sep-theta}}
(applied to the initial simulation $C^x_0, \movi, C^x_0$), $\nb{\pi}(q) \becomes 0$
and $\outi{\pi} \becomes \movi$ for all $\pi = C_1 \cdots C_n \in \tree{C^x_0}$
and all $q \in C_n$. \cor{We also let $\last(q) \defined \movi$ for all $q \in C^x_0$.}

\begin{claim} \label{claim:inv:init}
Invariants~\ref{inv:C} to~\ref{inv:sep}
(with $i = 0$) hold after this operation.
\end{claim}

Assume now that the invariants hold for some $x \in \Dom{f}$
and $i \ge 0$. We describe how $\strans$ updates its information
when it reads $x[i{+}1] C^x_{i+1}$. Let $a \defined x[i{+}1]$.

\subsection{If $C^x_i$ was not separable}

\label{ssec:origin-nosep}

In this case $\strans$
was in the non-separable mode.
We update $\pre{}{} \becomes \pre{}{} \circ \pre{C^x_i, C^x_{i+1}}{a}$,
$\Cf \becomes C^x_{i+1}$ and $\Jf \becomes \pre{}{}(\Cf)$.
Since $C^x_i, a, C^x_{i+1}$ was a pre-step,
then $\Jf, x[1{:}i{+}1], \Cf$ is an initial step.
For all $q \in C^x_{i+1}$,
let $\delta_q \defined \lag(\pre{C^x_i, C^x_{i+1}}{a}(q)) \val{C^x_i, C^x_{i+1}}{a}(q)$.
Now let $c \defined \bigwedge_{q \in Q} \delta_q$,
we update $\out \becomes \out c$ and 
define  $\alpha_q \defined c^{-1} \delta_q$ for all $q \in \Cf$.
It is easy to see that:

\begin{claim} \label{claim:retrouve}
$\out = \com{\Jf, \Cf}{x[1{:}i+1]}$
and $\alpha_q = \adv{\Jf, \Cf}{x[1{:}i+1]}(q)$ for all $q \in \Cf$.
\end{claim}
Finally we discuss two cases depending on the separability of
$\Cf = C^x_{i+1}$:
\begin{itemize}
\item if $\Cf$ is not separable, then $\strans$ stays in non-separable mode and it
updates $\lag(q) \becomes \alpha_q$ for all $q \in \Cf$
(note that $|\alpha_q| \le \Bound \le \Bound!$).
We easily see that invariants~\ref{inv:C}, \ref{inv:step} and~\ref{inv:nosep} hold. 
\item if $\Cf$ is separable, $\strans$ goes to separable mode.
\old{By applying \cref{lem:sep-theta} 
to $\Jf, x[1{:}i{+}1], \Cf$ we get $\tau \in B^*$
with $k \defined |\tau| \le 3 \Bound$. 
We update $\lag(q) \becomes \alpha_q[1{:}k]$ 
and $\last(q) \becomes  \alpha_q[k{+}1{:}]$
 for all $q \in \Cf$.
The $\theta$ is given by \cref{lem:sep-theta}
and we let $\nb{\pi}(q) \becomes 0$
and $\outi{\pi} \becomes \movi$ for all $\pi = C_1 \cdots C_n \in \tree{\Cf}$
(except for $\outi{\pi} = \out$ when $\pi = \Cf = C^x_{i+1}$) and all $q \in C_n$.}%
\cor{By applying \cref{lem:sep-theta} 
to $\Jf, x[1{:}i{+}1], \Cf$ we get $\tau \in B^*$
with $k \defined |\tau| \le \Bound!$ and $\theta \in B^*$
with $|\theta| = \Bound!$. 
We update $\lagm \becomes \tau$, $\lag(q) \becomes \alpha_q[1{:}k]$ 
and $\last(q) \becomes  \alpha_q[k{+}1{:}]$
 for all $q \in \Cf$.
We also let $\nb{\pi}(q) \becomes 0$
and $\outi{\pi} \becomes \movi$ for all $\pi = C_1 \cdots C_n \in \tree{\Cf}$
(except for $\outi{\pi} = \out$ when $\pi = \Cf = C^x_{i+1}$) and all $q \in C_n$.}
\old{\begin{lemma} \label{lem:premier}
Invariants~\ref{inv:C},
\ref{inv:step} and~\ref{inv:sep} hold  in $i{+}1$ after this operation.
Furthermore $|\theta| = \Bound!$, $|\lag(q)| \le 3 \Bound$
for all $q \in \Cf$,
and for all $\pi = C_1 \cdots C_n \in \tree{\Cf}$, 
$\nb{\pi} = 0$.
\end{lemma}}%
\cor{\begin{lemma} \label{lem:premier}
Invariants~\ref{inv:C},
\ref{inv:step} and~\ref{inv:sep} hold  in $i{+}1$ after this operation.
Furthermore $|\theta| = \Bound!$, $|\lagm| \le \Bound!$
for $q \in \Cf$,
and for all $\pi = C_1 \cdots C_n \in \tree{\Cf}$, 
$\nb{\pi} = 0$.
\end{lemma}}

Note that we may have $|\last(q)| \ge \Bound!$. In order to reduce their
sizes, we apply the tool detailed in Subsection~\ref{ssec:tool}
(it will push the $\last(q)$ into the $\nb{\pi}(q)$ and $\outi{\pi}$).
\end{itemize}

\subsection{Toolbox: reducing the size of $\last(q)$}

\label{ssec:tool}

\old{In this subsection, we assume that $\strans$ is in
its separable mode and that invariants~\ref{inv:step}
and~\ref{inv:sep} hold in some $i \ge 0$.
Furthermore, we suppose that $|\theta| = \Bound!$,
$ |\lag(q)| \le 3\Bound$ for all $q \in \Cf$,
and for all $C_1 \cdots C_n \in \tree{\Cf}$,
$\nb{C_1 \cdots C_n}: C_n \fonc [0{:}\four]$.
However the $\last$ may be longer than they should.
We are thus going to resize them.}%
\cor{In this subsection, we assume that $\strans$ is in
its separable mode and that invariants~\ref{inv:step}
and~\ref{inv:sep} hold in some $i \ge 0$.
Furthermore, we suppose that $|\theta| = \Bound!$,
$|\lagm|, |\lag(q)| \le \Bound!$ for all $q \in \Cf$,
and for all $C_1 \cdots C_n \in \tree{\Cf}$,
$\nb{C_1 \cdots C_n}: C_n \fonc [0{:}\four]$.
However the $\last$ may be longer than they should.
We are thus going to resize them.}

From invariant~\ref{inv:last}, there
exists $n: \Cf \fonc \Nat$ 
 such that
$\last(q) = \theta^{n(q)} \delta_q$ with $\delta_q \prefneq \theta$
for all $q \in \Cf$.
We update $\last(q) \becomes \delta_q$
and $\nb{\Cf}(q) \becomes \nb{\Cf}(q) + n(q)$
for all $q \in \Cf$. Now, we have $|\last(q)| < \Bound!$
and $\nb{\pi}(q) \le \four$ when $\pi \neq \Cf$.

\begin{algorithm}
\SetKw{KwVar}{Variables:}
\SetKwProg{Fn}{Function}{}{}
\SetKw{In}{in}
\SetKw{Out}{Output}

 \Fn{$\tnorm{\textbf{down}} (\pi)$}{
 
 	$C_1 \cdots C_n \becomes \pi$;

	 \tcc{1. Add the common part of the buffers to the local output}
	 
	 	$m \becomes \min_{q \in C_n} \nb{\pi}(q)$;
		
		$\outi{\pi} \becomes \outi{\pi} \theta^{m}$;
	
		$\nb{\pi}(q) \becomes \nb{\pi}(q) - m$ \tnorm{ for all } $q \in C'$;

	 \tcc{2. Check if some buffers $\nb{\pi}(q)$ are still more than $\four$}

	\For{$q \in C_n$}{
	
		\If{$\nb{\pi}(q) > \four$}{

			\For{$C' \in \Comp(C_n)$ \tnorm{\bfseries such that} $C' \neq C_n$
			 \tnorm{\bfseries and} $q \in C'$}{

					$\nb{\pi C'}(q) \becomes \nb{\pi C'}(q) + (\nb{\pi}(q){-}\four)$;

				}

			$\nb{\pi}(q) \becomes \four$;
			}
	
		}
		
	 \tcc{3. Recursive calls}
		
	\For{$C' \in \Comp(C_n)$ \tnorm{\bfseries with} $C' \neq C_n $}{
		$\tnorm{\textbf{down}} (\pi C')$;
		}
}
	
 \caption{\label{algo:down} Sending down values in $\tree{\Cf}$}
\end{algorithm}

In order to reduce the value $\nb{\Cf}$, we then apply 
the function $\tnorm{\textbf{down}}(\Cf)$ of Algorithm~\ref{algo:down}
which adds some $\theta$ in the $\outi{\pi}$.
Let us describe its base case informally. If
$\nb{\Cf}(q) > 0$ for all $q \in \Cf$, then no state is lagging by invariant~\ref{inv:lagging}.
Thus $\lag(q) = \lagm$ for all $q \in \Cf$, and so $\lagm = \movi$
by invariant~\ref{inv:lag}. With the notations of invariant~\ref{inv:past}
(note that $\pi_1 = \Cf$), we get $\val{x[1{:}i]}{\Jf,\Cf}(q) = \out~\theta^{\nb{\pi_1}(q)}
\left(\prod_{i=2}^n \outi{\pi_i} \theta^{\nb{\pi_i}(q)} \right) \last(q)$
for all $q \in \Cf$. Thus we can produce in
$\out$ the value $\theta^{m} \defined \bigwedge_{q \in \Cf} \theta^{\nb{\Cf}(q)}$
(i.e. $m \becomes \min_{q \in \Cf} \nb{\Cf}(q)$)
and remove $m$ to each $\nb{\Cf}(q)$.

\old{
\begin{lemma} \label{lem:sim:algo}
Algorithm~\ref{algo:down} is well defined.
After the operation described
in this subsection, invariants~\ref{inv:step} and~\ref{inv:sep}
hold, and furthermore we have
 $|\theta| = \Bound!$, $|\lag(q) | \le 3 \Bound$ and \mbox{$|\last(q)| < \Bound!$}
 for all $q \in \Cf$,
and for all $\pi = C_1 \cdots C_n \in \tree{\Cf}$, $\nb{\pi}: C_n \fonc [0{:}\four]$.
\end{lemma}}

\cor{
\begin{lemma} \label{lem:sim:algo}
Algorithm~\ref{algo:down} is well defined.
After the operation described
in this subsection, invariants~\ref{inv:step} and~\ref{inv:sep}
hold, and furthermore we have
 $|\theta| = \Bound!$, $|\lagm| \le \Bound!$ and \mbox{$|\last(q)| < \Bound!$}
 for all $q \in \Cf$,
and for all $\pi = C_1 \cdots C_n \in \tree{\Cf}$, $\nb{\pi}: C_n \fonc [0{:}\four]$.
\end{lemma}
}

\subsection{If $C^x_i$ was separable}

If $C^x_i$ is separable, then $\strans$
was in the separable mode by invariant~\ref{inv:sep}.
We first explain in Subsubsection~\ref{sssec:up-step} how to perform the update when
$\Cf, a, C^x_{i+1}$ is a step (it
corresponds to the ``easy case'' thanks to invariant~\ref{inv:future}
which deals with future steps). Then, we explain
in Subsubsection~\ref{sssec:prepro} how the other
case can be reduced to the first one, after a preprocessing
which selects a subset $C' \subseteq \Cf$ such that $C',a,C^x_{i+1}$ is a step.

\subsubsection{Updating when $\Cf,a,C^{x}_{i+1}$ is a step}

\label{sssec:up-step}

In the current subsubsection we assume that invariants~\ref{inv:step}
and~\ref{inv:sep} hold, that $\Cf \subseteq C^x_i$ is separable
\cor{(we may not have $\Cf = C^x_i$ because of the preprocessing 
of Subsubsection~\ref{sssec:prepro})},
and that $\Cf,a,C^{x}_{i+1}$ is a step. We show how
to update the information stored by $\strans$ in accordance with
this step. Note that $C^x_{i+1}$ is necessarily separable.

Since $\Cf$ will be modified, so will be $\tree{\Cf}$, hence
we begin with several register updates.
For $\pi = D_1 \cdots D_n \in \tree{C^x_{i+1}}$, we define
$C_i \defined \pre{\Cf, C^{x}_{i+1}}{a}(D_i)$ for $1 \le i \le n$.
Since we had a step then
$C_1 = \Cf$, $C_i \in \Comp(\Cf)$ and $C_1 \supseteq \cdots \supseteq C_n$.
But we may not have $C_1 \cdots C_n \in \tree{\Cf}$
due to possible equalities. Let $1 = i_1 < \cdots < i_m \le n$
be such that $C_{i_1}  = \cdots = C_{i_2-1} \supsetneq C_{i_2}$
and so on until $C_{i_m-1} \supsetneq C_{i_m} = \cdots = C_n$.
Then $\rho \defined C_{i_1} \cdots C_{i_m} \in \tree{\Cf}$ and:
\begin{itemize}
\item if $i_m = n$, we let $\nb{\pi} \becomes \nb{\rho}
\circ \pre{\Cf,  C^{x}_{i+1}}{a}$
and $\outi{\pi} \becomes \outi{\rho}$;
\item if $i_m < n$, we let  $\nb{\pi} \becomes 0$ and $\outi{\pi} \becomes \movi$.
\end{itemize}
For all $q \in C^x_{i+1}$, let $k_q \defined |{\lag(\pre{\Cf,  C^{x}_{i+1}}{a}(q))}^{-1} \lagm|$ and:
\begin{itemize}
\item $\lag(q) \becomes \lag(\pre{\Cf,  C^{x}_{i+1}}{a}(q)) (\val{\Cf, C^{x}_{i+1}}{a}(q) [1{:}k_q])$
(note that $\lagm$ remains unchanged);
\item $\last(q) \becomes \last(\pre{{\Cf,  C^{x}_{i+1}}}{a}(q))(\val{\Cf,  C^{x}_{i+1}}{a}(q) [k_q{+}1{:}])$.
\end{itemize}
\old{Now let $c \defined \bigwedge_{q \in \Cf} \lag(q)$. We
update $\lag(q) \becomes c^{-1} \lag(q)$ for all $q \in \Cf$
(therefore $\lagm$ becomes $c^{-1} \lagm$), $\out \becomes \out~c$,
$\Cf \becomes C^x_{i+1}$ and finally $\pre{}{} \becomes \pre{}{} \circ \pre{\Cf,  C^{x}_{i+1}}{a} $.}
\cor{
Now let $c \defined \bigwedge_{q \in \Cf} \lag(q)$. We
update $\lag(q) \becomes c^{-1} \lag(q)$ for all $q \in \Cf$
and $\lagm \becomes c^{-1} \lagm$, $\out \becomes \out~c$,
$\Cf \becomes C^x_{i+1}$ and finally $\pre{}{} \becomes \pre{}{} \circ \pre{\Cf,  C^{x}_{i+1}}{a} $.
}

\old{
\begin{lemma} \label{lem:sim:update}
After the operation described in this subsection,
invariants~\ref{inv:C}, \ref{inv:step} and~\ref{inv:sep} hold,
and $|\theta| = \Bound!$, $|\lag(q)| \le \Bound$ for all $q \in \Cf$,
and $\nb{\pi}: C_n \fonc [0{:}\four]$ for all $\pi = C_1 \cdots C_n \in \tree{\Cf}$.
\end{lemma}}

\cor{
\begin{lemma} \label{lem:sim:update}
After the operation described in this subsection,
invariants~\ref{inv:C}, \ref{inv:step} and~\ref{inv:sep} hold,
and $|\theta| = \Bound!$, $|\lagm| \le \Bound!$ for all $q \in \Cf$,
and $\nb{\pi}: C_n \fonc [0{:}\four]$ for all $\pi = C_1 \cdots C_n \in \tree{\Cf}$.
\end{lemma}
}
However, we may have $|\last(q)| \ge \Bound!$.
Thus we finally apply Subsection~\ref{ssec:tool} once more.

\subsubsection{Preprocessing when $\Cf, a, C^x_{i+1}$ is not a step}

\label{sssec:prepro}

In the current subsubsection we assume that invariants~\ref{inv:C}, \ref{inv:step}
and~\ref{inv:sep} hold in $i \ge 0$, that $\Cf = C^x_i$ is separable,
and that $C^x_i,a,C^{x}_{i+1}$ is \emph{not} a step. 
Then let $C' \defined \pre{C^x_i, C^x_{i+1}}{a}(C^x_{i+1}) \subsetneq \Cf$
(an equality would give a step)
and $\pi \defined \Cf C' \in \tree{C}$. Two cases
can occur.

\subparagraph{If $\pi$ is close.} In this case, we have for all $\pi \prefneq \pi' \in \tree{\Cf}$
that $\nb{\pi'} = 0$ and $\outi{\pi'} = \movi$.
Therefore by invariant~\ref{inv:past} we can 
describe the productions for all $q \in C'$ as follows:
\begin{itemize}
\item $\val{x[1{:}i]}{\Jf,\Cf}(q) = \out~\lag(q)$ if $q$ is lagging;
\item $\val{x[1{:}i]}{\Jf,\Cf}(q) = \out~\lagm ~\outi{\Cf C'}~\theta^{\nb{\Cf}(q) + \nb{\Cf C'}(q)}~\last(q)$ 
 if $q$ is not lagging.
\end{itemize}
Now two cases are possible, depending on whether there is a lagging state in $C'$ or not:
\begin{itemize}
\item if there exists $q' \in C'$ which is lagging, then we
must have $\outi{\Cf C'} = \movi$ by invariant~\ref{inv:lagging}.
For all $q \in C'$ let $\delta_q \defined \lag(q) \theta^{\nb{\Cf}(q) + \nb{\Cf C'}(q)}~\last(q)$
and let $c \defined \bigwedge_{q \in C'}\delta_q $. Then we
update  $\out \becomes \out~c~$ and
 define $\alpha_q \defined c^{-1} \delta_q $ for all $q \in C'$;
\item if each $q \in C'$ is not lagging, we define 
$\delta_q \defined \theta^{\nb{\Cf }(q) + \nb{\Cf C'}(q)} \last(q)$
and $c \defined \bigwedge_{q \in C'} \delta_q $.
Then we update $\out \becomes \out~\lagm~\outi{\Cf C'}~c~$
and define $\alpha_q \defined c^{-1}\delta_q $ for all $q \in C'$;
\end{itemize}
We finally update $\Jf \becomes \pre{}{}(C')$, $\Cf \becomes C'$
and $\pre{}{} \becomes \pre{}{}\hspace*{-0.06cm}|_{C'}$.
It is easy to see that $\Jf, x[1{:}i], \Cf$ is a step and 
furthermore that we have computed $\com{}{}$ and $\adv{}{}$,
as shown in \cref{cla:claclacla}.

\begin{claim} \label{cla:claclacla}
After this operation, $\out = \com{\Jf, \Cf}{x[1{:}i+1]}$
and $\alpha_q = \adv{\Jf, \Cf}{x[1{:}i+1]}(q)$ for all $q \in \Cf$.
\end{claim}
This result exactly corresponds to \cref{claim:retrouve}
from Subsection~\ref{ssec:origin-nosep} (replace $i{+}1$ by $i$).
Thus, to conclude, we just need to apply the operations
described after \cref{claim:retrouve} (i.e. determining
if the new $\Cf$ is separable or not, and building the structure accordingly).

\subparagraph{If $\pi$ is not close.}
\old{Let $c \defined \bigwedge_{q \in C'} \lag(q)$,
we update $\out \becomes \out~c~\outi{\Cf C'}$
and for all $q \in C'$, $\lag(q) \becomes c^{-1} \lag(q)$
and $\last(q) \becomes \theta^{\nb{\Cf}(q)} \last(q)$.}
\cor{Let $c \defined \bigwedge_{q \in C'} \lag(q)$,
we update $\out \becomes \out~c~\outi{\Cf C'}$
and for all $q \in C'$, $\lag(q) \becomes c^{-1} \lag(q)$,
$\last(q) \becomes \theta^{\nb{\Cf}(q)} \last(q)$, and
finally $\lagm \becomes c^{-1} \lagm$}.
Then, we update $\nb{C' \pi} \becomes \nb{\Cf C' \pi}$
and $\outi{C' \pi} \becomes \outi{\Cf C' \pi}$
for all $\pi \in (C')^{-1}\tree{C'}$ 
(except for $\pi = \movi$, in which case we 
have already updated $\outi{C'} = \out$ before).
We finally update $\Jf \becomes \pre{}{}(C')$, $\Cf \becomes C'$
and $\pre{}{} \becomes \pre{}{}\hspace*{-0.06cm}|_{C'}$.

\old{
\begin{lemma} \label{lem:sim:prune}
After the operation described in this subsection,
invariants~\ref{inv:step} and~\ref{inv:sep} hold,
and $|\theta| = \Bound!$, $|\lag(q)| \le \Bound$ for all $q \in \Cf$,
and $\nb{\pi}: C_n \fonc [0{:}\four]$ for all $\pi = C_1 \cdots C_n \in \tree{\Cf}$.
Furthermore $\Cf$ is separable.
\end{lemma}}%
\cor{
\begin{lemma} \label{lem:sim:prune}
After the operation described in this subsection,
invariants~\ref{inv:step} and~\ref{inv:sep} hold,
and $|\theta| = \Bound!$, $|\lagm| \le \Bound!$ for all $q \in \Cf$,
and $\nb{\pi}: C_n \fonc [0{:}\four]$ for all $\pi = C_1 \cdots C_n \in \tree{\Cf}$.
Furthermore $\Cf$ is separable.
\end{lemma}
}

\begin{remark}  Contrary to the former cases,
the main difficulty here is to show the preservation
of invariant~\ref{inv:future}. For this we critically rely on invariant~\ref{inv:close}:
the idea is to show that $\theta$ is still a suitable looping value,
even if we have chosen a subset of our compatible set
(observe that in \cref{lem:sep-theta}, the value $\theta$
depends on the compatible set chosen).
\end{remark}
Again, we may have $|\last(q)| \ge \Bound!$. Thus 
we finally apply Subsection~\ref{ssec:tool} once more.

\subsection{Boundedness and productivity of the construction}

\label{sec:bound}

We first claim that $\strans$ is a $1$-bounded \DSST{}, by construction.

\begin{lemma} \label{lem:strans-bound} The \DSST{} $\strans$ is $1$-bounded.
\end{lemma}
It follows from invariants~\ref{inv:C}, \ref{inv:step} and~\ref{inv:past} that for all $x \in \Dom{f}$,
$\out$ is always a prefix of $f(x)$ when $\strans$ reads $g(x)$.
To conclude the construction of $\strans$, it remains to see that $\out$ tends to an infinite word.
The key ideas for showing \cref{lem:infini} is to use
the fact that $\trans$ is productive, and that Algorithm~\ref{algo:down}
can only empty a buffer $\nb{\Cf}(q)$ if it outputs a word.

\begin{lemma} \label{lem:infini}
If $x \in \Dom{f}$, then $|{\out}| \fonc \infty$ when $\strans$ reads $g(x)$.
\end{lemma}

\section{Outlook}

This paper provides a solution to an open problem.
From a practical point of view, it allows to build
a copyless streaming algorithm from a rational specification
whenever it is possible (it is impossible when the 
rational function is not continuous).
We conjecture that the techniques
introduced in this paper can be extended to show
that any continuous \emph{regular} function is deterministic regular.
Furthermore, they may also
be used to study the rational or regular functions which are
uniformly continuous for the Cantor topology,
and capture them with a specific transducer model
(another open problem of~\cite{dave2020synthesis}).


\bibliography{rational-continuity}

\newpage

\appendix

\section{Proof of Remark~\ref{rem:not-sequential}}

We show in this section that the function $\double$ from
Example~\ref{ex:rational} cannot be computed by a 
one-way deterministic transducer. It implies that one-wayness
is not enough to compute continuous rational
functions by deterministic transducers,
even if the original function is total (contrary to what happens with the subclass
of rational functions considered in \cite{filiot2021synthesizing}).

Now, asssume that $\double$ is computed by a
deterministic one way transducer.
There exists $N,M \ge 0$,   $u_1,u_2, u_3, u_4 , u'_3, u'_4 \in  \{0,1,2\}^*$ such that
for all $n \ge 0$, $\double(0^{M + Nn} 1 0^\omega) = u_1u_2^nu_3u_4^\omega
= 0^{M+Nn} 1 0^\omega$
and $\double(0^{M + Nn} 20^\omega)
= u_1u_2^nu'_3{u'_4}^\omega
= 0^{2(M+Nn)} 2 0^\omega$.
A contradiction can easily be deduced.

\section{Proof of Lemma~\ref{lem:infinite-accepting}}

Let $\trans = (A,B,Q, I, F, \Delta, \lambda)$ be a trim \oNT{}
computing $f: A^\omega \parfonc B^\omega$.
By a pumping argument, it is easy to see
that $\trans$ is clean if and only if it has no
$q \runs{u | \movi} q$ for some $q \in F$.
Assume furthermore that $\trans$ is unambiguous.
We build a clean and unambiguous  (trimming can be done later) \oNT{}
$\trans' = (A,B,Q', I', F', \Delta', \lambda')$
as follows:
\begin{itemize}
\item $Q' \defined (Q \times \{0\}) \uplus (Q \times \{1\})$, $I' \defined I \times \{1\}$
and $F' \defined F \times \{1\}$;
\item for all $(q,a,q') \in \Delta$ we add:
\begin{itemize}
\item $((q,1), a, (q',1)) \in \Delta'$ if $q \not \in F$, and then
$\lambda'((q,1), a, (q',1)) \defined \lambda(q,a,q')$;
\item $((q,1), a, (q',0)) \in \Delta'$ if $q \in F$, and then
$\lambda'((q,1), a, (q',0)) \defined \lambda(q,a,q')$;
\item $((q,0), a, (q',0)) \in \Delta'$ if $\lambda'((q,0), a, (q',0)) \defined \lambda(q,a,q') = \movi$;
\item $((q,0), a, (q',1)) \in \Delta'$ if $\lambda'((q,0), a, (q',1)) \defined \lambda(q,a,q') \neq \movi$.
\end{itemize}
\end{itemize}
It is easy to see that $\trans'$ is unambiguous, clean, and computes $f$.

\section{Proof of Theorem~\ref{theo:2dt-dsst}}

\subsection{From \tDT{}s to $1$-bounded \DSST{}s}

We show how to transform a \tDT{} into a $1$-bounded \DSST{}.  The main idea is to keep track
of the right-to-right behavior of the \tDT{}  (the ``crossing sequence'') on
the prefix read so far. This proof is somehow standard for \DSST, and its main ideas originate
from \cite{shepherdson1959reduction} which first showed
how to transform a two-way automaton into a one-way
automaton.

Consider a \tDT{} $\trans = (A,B,Q,q_0, \delta, \lambda)$ computing a partial function
$f: A^\omega \parfonc B^\omega$. Let $x \in A^\omega$ . We denote by $\fonc$ the transition
relation of $\trans$ between configurations, and $\fonc^+$ its transitive closure.
Let $\Qubot \defined Q \uplus \{\bot\}$.
When reaching a position $i \ge 0$ of the input $\lmark x$, the \DSST{}
will keep track the following information (see Figure \ref{fig:cross}):
\begin{itemize}
  \item the state $\first^x_i \in \Qubot$ that is ``the state of
  $\trans$ the first time it goes to position $i{+}1$''.  More formally, $\first^x_i$
  is the state such that $ \phi^x_i \defined (q_0,0) \fonc^{+} (\first^x_{i}, i{+}1)$
 is the run which visits the position $i{+}1$ for the first time (and $\bot$ if such a run does not exist,
  which implies that $x \not \in \Dom{f}$).  This information is coded in the state of the \DSST{};
  \item the concatenation of the outputs by $\lambda$
  along the run $\phi^x_i$, stored in the register $\out$;
  \item a function $\nnext^x_i(q) \colon Q \fonc \Qubot$, which
  gives for each $q \in Q$ the state such that the run
  $\rho^x_i(q) \defined (q,i) \fonc^{+} (\nnext^x_i(q),i{+}1)$
  visits $i{+}1$ for  the first time when starting from $(q,i)$ ($\bot$ if it does not exist).  This
  information is coded in the state of the \DSST{};
  \item  the concatenation of the outputs by $\lambda$
  along the run $\rho^x_i(q)$ for all $q \in Q$.
  This information is stored in a register $\outi{q}$
  (it is empty if $\nnext^x_i(q) = \bot$).
\end{itemize}

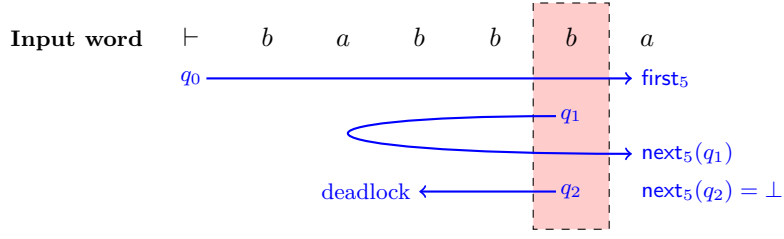
\begin{figure}[h!]

\begin{center}
\begin{tikzpicture}[scale=1]
	\newcommand{\theon}[1]{\footnotesize {\textbf{#1}}};
	\node (in) at (-2.5,0) [above,right]{\theon{Input word}};
	
	\draw[fill = red!20,dashed](4.5,0.5) rectangle (5.5,-2.5);
		
	\node[blue] (in) at (0,-0.5) []{\footnotesize  $q_0$};
	\draw[->,blue,thick](0.2,-0.5) to (5.8,-0.5);
	\node[blue] (in) at (5.8,-0.5) [right]{\footnotesize $ \first_5$};

	\node[blue] (in) at (5,-1) []{\footnotesize  $q_{1}$};
	\node[blue] (in) at (5.8,-1.5) [right]{\footnotesize  $\nnext_5(q_1)$};
	\draw[->,blue,thick] (4.8,-1) .. controls (1,-1) and (1,-1.5) .. (5.8,-1.5) ;
	
	\node[blue] (in) at (5,-2) []{\footnotesize  $q_{2}$};
	\node[blue] (in) at (5.8,-2) [right]{\footnotesize  $ \nnext_5(q_2) =  \bot$};
	\draw[->,blue,thick](4.8,-2) to (3,-2);
	\node[blue] (in) at (2.3,-2) []{\footnotesize  deadlock};	

	\node[above] (in) at (0,-0.2) []{  $\lmark$};
	\node[above] (in) at (1,-0.2) []{  $b$};
	\node[above] (in) at (2,-0.2) []{  $a$};
	\node[above] (in) at (3,-0.2) []{  $b$};
	\node[above] (in) at (4,-0.2) []{  $b$};
	\node[above] (in) at (5,-0.2) []{  $b$};
	\node[above] (in) at (6,-0.2) []{  $a$};
	\node[above] (in) at (7,-0.2) []{  $$};
       \end{tikzpicture}
\end{center}
       \caption{\label{fig:cross} Crossing sequences in a \tDT}
\end{figure}

\subparagraph*{Updates when reading a letter.} We have to show how the \DSST{} can update
this abstraction of the behavior of $\trans$.  Assume that the invariants
are computed correctly at position $i \ge 0$, we want to compute them
for $i{+}1$. Let $a \defined x[i{+}1]$. We recursively define  a function $\sigma^x_i: \Qubot \rightarrow
(\Qubot)^* \cup (\Qubot)^\omega$, which describes the sequence of states visited
in position $i$ when starting from $q$ in position $i{+}1$, and before reaching position $i{+}2$:
\begin{enumerate}
\item if $q = \bot$ or $\delta(q,a)$ is undefined, then $\sigma^x_i(q) \defined \bot^\omega$;
\item otherwise if $\delta(q,a) = (q', \rmove)$, then $\sigma^x_i(q) \defined \movi$;
\item \label{item:sew3} otherwise if $\delta(q,x[i{+}1]) = (q', \lmove)$, then $\sigma^x_i(q) \defined q' \sigma^x_i(\nnext^x_i(q'))$.
\end{enumerate}

\begin{claim} \label{claim:clam} Let $q \in Q$. Then $p_1 \cdots p_n \in Q^*$
is a prefix of $\sigma^x_i(q)$ if and only if
$(q,i{+}1) \fonc \rho^x_i(p_1) \fonc
\cdots \fonc \rho^x_i(p_n)$ is a run.
\end{claim}

\begin{proof} Follows directly from
the definitions of $\sigma^x_i(q)$ and $\rho^x_i(q)$.
\end{proof}

\begin{claim} \label{claim:nono} If $\nnext^x_{i+1}(q) = \bot$, then $|\sigma^x_i(q)| = \infty$.
\end{claim}

\begin{claim} \label{claim:nana}
If $\nnext^x_{i+1}(q) = q' \neq \bot$, then $\sigma^x_i(q) = p_1 \dots p_n$ with $n \le |Q|$.
Furthermore if $n=0$ then $\delta(q,a) = (q', \rmove)$,
if $n > 0$ then $\delta(\nnext^x_i(p_n),a) = (q', \rmove)$,
and  $\rho^x_{i+1}(q) = (q,i{+}1) \fonc \rho^x_i(p_1)  \fonc
\cdots \fonc \rho^x_i(p_n) \fonc (q',i{+}2)$
\end{claim}

\begin{proof}[Proof of claims~\ref{claim:nono} and~\ref{claim:nana}] 
We have $\sigma^x_i(q) \in Q^* \cup Q^\omega \cup Q^* \bot^\omega$.
If $\sigma^x_i(q) = p_1 \cdots p_n \bot^\omega$ then we
must have either $\nnext^x_i(p_n) = \bot$ or $\delta(\nnext^x_i(p_n),a)$ undefined. In
both cases this implies $\nnext^x_{i+1}(q) = \bot$. On the other hand,
if $p_1 \cdots p_n \in Q^*$ with $n > |Q|$ is a prefix of $\sigma^x_i(q)$,
then the run of Claim~\ref{claim:clam} above visits twice the same configuration.
Thus it loops infinitely before seeing position $i{+}2$,
which gives $\nnext^x_{i+1}(q) = \bot$.
Finally, the remaining case is when $\sigma^x_i(q) = p_1 \cdots p_n \in Q^*$
with $n \le |Q|$. This means that  $\delta(q, a) = (q', \rmove)$ or $\delta(\nnext^x_i(p_n), a) = (q', \rmove)$
which, when added to the run of Claim~\ref{claim:clam},
gives the first visit of position $i{+}2$ in $\rho^x_{i+1}(q)$
\end{proof}
By computing of the $|Q|$ first
letters of $\sigma(q)$, the \DSST{}
can update the $\nnext^x_{i+1}$ information.
Furthermore, we define the register updates as follows:
\begin{itemize}
\item if $\nnext^x_{i+1}(q) = \bot$, then $\out_q \becomes \movi$;
\item otherwise $\sigma_i^x(q) = p_1 \cdots p_n$, and then
$
\out_q \becomes  \alpha_0 \out_{p_1} \alpha_2 \cdots \out_{p_n} \alpha_n
$ where the $\alpha_j \in B^*$ are defined accordingly
to $\rho^x_{i+1}(q)$ in Claim~\ref{claim:nana}.
\end{itemize} 
It remains to deal with the update of $\first^x_{i+1}$ and $\out$. By a similar
argument, it is easy to see that   $\first^x_{i+1} \neq \bot$ if and only
if $\sigma(\first^x_i) = p_1 \cdots p_n \in Q^*$ with $n \le |Q|$,
and furthermore in that case
$\phi^x_{i+1} = \phi^x_i \fonc \rho^x_{i+1}(\first^x_i)$.
We can thus update $\out \becomes \out
\alpha_0 \out_{p_1} \cdots \out_{p_n} \alpha_n$
(where the $\alpha_i$ are the same as those
of the update of $\out_{\first^x_i}$).

\subparagraph*{Correctness of the construction.} Let $x \in \Dom{f}$,
then by definition of the semantics of a \tDT{}, we
have $\first^x_i \neq \bot$ for all $i \ge 0$ and furthermore
output labels along the run $\phi^x_i$
tends to $f(x) \in B^\omega$ when $i \rightarrow + \infty$.
Thus $\cro{\out}^x_i$ tends to $f(x)$.

Now assume that $x \not \in \Dom{f}$. Then either
$\first^x_i = \bot$ for some $i \ge 0$
(either because there is no infinite run of the \tDT{}
labelled by $x$, or because this run is not accepting, i.e. it contains
a loop), which will be detected by our \DSST{} that will stop its computation.
Or $\first^x_i \neq \bot$ for all $i \ge 0$,
but the output labels of  the $\phi^x_i$ do not tend to an infinite word,
and therefore we get $|\cro{\out}^x_i| \not \fonc \infty$.

\subparagraph*{Copyless?} One could ask if the resulting \DSST{}
is copyless. It is not the case, since a state $p \in Q$
may occur both in $\sigma^x_i(q) \in Q^*$ and $\sigma^x_i(q') \in Q^*$
(which implies that $\out_p$ is used both for $\out_{q}$
and $\out_{q'}$). However, in this case it means that either
 $\out_{q}$ or $\out_{q'}$
stores information which is never used in $\out$
(since otherwise, it would induce a looping behavior in $\trans$).

\subparagraph*{Bounded copies.} Let $(A,B, S, s_0, \delta', \Regs, \out, \lambda')$
be the \DSST{}  built along the previous
paragraphs (this does not modify its semantics).
Let us show by Claim~\ref{claim:exact-branch} that it is $1$-bounded.
Recall that $\lambda'^x_i$ is the substitution applied (when defined)
when reading $x[i]$ on input $x \in A^\omega$.

\begin{claim} \label{claim:exact-branch}
Let $x \in A^\omega$, $1 \le j \le i$ be such
that $\lambda'^x_i$ is defined, then:
\begin{enumerate}
\item \label{item:unun} for all $\reg, \regg \in \Regs$,
$\reg$ occurs at most once in $\lambda'^x_j \circ \cdots \circ\lambda'^x_i  (\regg)$;
\item \label{item:dede} if $\out_p$ occurs in
$\lambda'^x_j \circ \cdots \circ\lambda'^x_i (\out_q)$, then
$\nnext_i(q) \neq \bot$ and $(p,j)$ occurs in the run $\rho^{x}_{i}(q)$;
\item  \label{item:dede-out}
 if $\out_p$ occurs in $\lambda'^x_j \circ \cdots \circ\lambda'^x_i (\out)$, then
$\first_i(q) \neq \bot$ and $(p,j)$ occurs in the run $\phi^{x}_{i}$.
\end{enumerate}
\end{claim}

\begin{proof} We show the three items by induction on $i \ge 1$.
The base case for $i = j$ follows from the definition of $\sigma$ and of the updates.
Now consider the induction step from $i \ge 1$ to $i{+}1$. Item~\ref{item:unun}
holds for $\reg = \out$ by definition of the updates $\lambda'$. Assume that $\out_p$ occurs in
$\lambda'^x_j \circ \cdots \circ\lambda'^x_{i+1}(\out_q)$,we show simultaneously
that items~\ref{item:unun} and~\ref{item:dede} hold.
First note that $\nnext^{x}_{i+1}(q) \neq \bot$, since otherwise
the update in $x[i{+1}]$ is $\out_q \becomes \movi$.
Let $p_1 \dots p_n = \sigma^{x}_{i+1}(q)$, then there exists $\alpha_0, \dots, \alpha_n \in B^*$
such that $\lambda'^x_i(\out_q) = \alpha_0 \out_{p_1} \cdots \out_{p_n} \alpha_n$.
Then by Claim~\ref{claim:clam} the run $\rho^{x}_{i+1}(q)$ is
obtained by concatenating the $\rho^{x}_{i}(p_k)$ for $1 \le k \le n$
in a disjoint way. Since $j \le i{+}1$,
then $\out_p$ occurs $\lambda'^{x}_{j} \cdots \lambda'^{x}_{i}(\out_{p_k})$
for some $1 \le k \le n$. then by induction hypothesis
(item~\ref{item:dede}), $(p,j)$ occurs in
$\rho^{x}_{i+1}(q)$. But since there is no loop in $\rho^{x}_{i+1}(q)$,
then $\out_p$ only occurs in $\lambda'^{x}_{j} \cdots \lambda'^{x}_{i}(\out_{p_k})$.
By induction hypothesis (item~\ref{item:unun})
$\out_p$ occurs only once in this $\lambda'^{x}_{j} \cdots \lambda'^{x}_{i}(\out_{p_k})$,
and finally $\out_p$ occurs only once in $\lambda'^{x}_{j} \cdots \lambda'^{x}_{i}(\out_q) $.
The proof of item~\ref{item:dede-out} is similar. 
\end{proof}

\subsection{From copyless \DSST{}s to \tDT{}s}

Let $\trans = (A,B,Q, q_0,\delta, \Regs, \out, \lambda)$
be a copyless \DSST{}
computing a function $f: A^\omega \parfonc B^\omega$.
We first give a simple recursive algorithm to compute $f$, and then
show that this algorithm is correct and can
be implemented by a \tDT.

\subparagraph*{Algorithmic description.} Given $x \in \Dom{f}$,  
let $\rho \defined q_0^x \rightarrow q_1^x \rightarrow \cdots$
be the initial run of $\trans$ on $x$ (with the convention that
$q^x_i = \bot$ if it is undefined). For $i \ge 1$, we also
define $\lambda^x_i \defined \lambda(q_{i-1}^x,x[i])$ the substitution
applied when reading $x[i]$.
We give in  Algorithm~\ref{algo:dsst-2dt} a function $\aval(\alpha,i, \reg)$, producing
$\cro{\alpha}^x_i$ when $i \ge 0$ and $\alpha \in (B \uplus \Regs)^*$ as input
($\reg$ shall be used later, it is an additional information
used by a  \tDT{} implementing the function).
It makes recursive calls to compute the values of
the registers which occur in $\alpha$.

\begin{algorithm}[h!]
\SetKw{KwVar}{Variables:}
\SetKwProg{Fn}{Function}{}{}
\SetKw{In}{in}
\SetKw{Out}{Output}

 \Fn{$\aval (u,i, \reg)$}{
 
 	\tcc{$\alpha \in (B \uplus \Regs)^*$ to be computed, $i \ge 0$ current position}

		\For{$a$ \In $\alpha[1], \dots, \alpha[|\alpha|]$}{
		
			\eIf{$a \in B$}{
			
			\Out{$a$};
			\tcc{Letter $a \in B$ is output}
			
			}{
			
			\If{$i > 0$}{

				$\aval(\lambda^x_i(a),i{-}1,a)$;
				\tcc{Recursion on $\lambda^x_i(a)$}
			
			}
			
		}

	}
	
}

 \Fn{$\simu$}{
 
	\For{$i$ \In $\{1,\dots, \infty \}$}{
				
		$\out \alpha \becomes \lambda^x_i(\out)$;
		
		$\aval(\alpha, i, \out)$;
		}
}
	
 \caption{\label{algo:dsst-2dt} Computing recursively the values of a register}
\end{algorithm}

\begin{claim}[Algorithm~\ref{algo:dsst-2dt} is correct] \label{claim:indu-2}
Let $x \in A^\omega$ and $i \ge 0$
be such that $q^x_i \neq \bot$. Then for all
 $\alpha \in (\Regs \uplus B)^*$ and $\reg \in \Regs$, $\aval(\alpha, i, \reg)$
 terminates and outputs $\cro{\alpha}^x_i$.
\end{claim}

\begin{proof}
Immediate by induction on $i \ge 0$.
\end{proof}

\begin{remark} We shall assume that Algorithm~\ref{algo:dsst-2dt} blocks
if $q^x_i = \bot$ (since it happens if and only if $\lambda^x_i$ is undefined).
\end{remark}

We have also described in Algorithm~\ref{algo:dsst-2dt} a function $\simu$
which uses $\aval$. This function ranges over the positions
$1, 2, \dots $ of the input, and for each of them
it produces ``the value added in $\out$'' at this position.

\begin{claim} \label{claim:indu-3}
If $x \in \Dom{f}$, then
$\simu$ loops infinitely and outputs $f(x)$.
\end{claim}

\begin{claim} \label{claim:indu-3}
If $x \not \in \Dom{f}$, then either $\simu$
gets blocked at some point, or it produces a finite output.
\end{claim}

\begin{proof} Two cases are possible if $x \not \in \Dom{f}$. Either $q^x_i = \bot$
for some $i \ge 0$, and then $\simu$ gets blocked before position $i$.
Or  $q^x_i \neq \bot$  for all $i \ge 0$, but $\cro{\out}^x_i$ tends
to a finite word $w \in B^*$. It is easy to see that $\simu$
then produces $w$.
\end{proof}

\subparagraph*{Implementation
by a \tDT.} We now describe how to implement the
function $\simu$ by a \tDT{}. First note that the machine needs to
determine the substitution $\lambda^x_i$, hence
the state $q_{i-1}^x$ (possibly $\bot$) when in position $i \ge 1$. For this,
we add a \emph{lookbehind} feature to our \tDT,
which enables it to choose its transition
 depending on a regular property of the prefix up to the current position.
Over finite words, it is well known that given a \tDT{} with lookbehind,
one can build an equivalent \tDT{} (see e.g. the ``lookahead removal''
techniques of \cite{chytil1977serial}).
We claim that the very same proof can also be applied to
infinite words, since lookbehinds only
concern a finite prefix of the input.

We now show how this \tDT{} proceeds.
The main loop of $\simu$ on $i \ge 1$ is executed by moving right
on the input. At each position $i \ge 1$ the \tDT{} determines
the value $\alpha \in (\Regs \uplus B)^*$ such that $ \lambda^x_i(\out) = \out \alpha$
(if it is defined) by using its lookbehind.
Then it processes each character $a \in \alpha[1] \cdots \alpha[|\alpha|]$ (since $\alpha$
is a bounded information, this loop is hardcoded in the states without moving).
If $a \in B$, it is output. Otherwise, the machine determines
$\beta \defined \lambda^x_{i}(a)$, moves left
and executes $\aval(\beta,i{-}1,a)$ by doing recursive calls.
The case when $i = 0$ is detected by reading the initial letter $\lmark$.

However,  a \tDT{} (which has a bounded memory)
cannot keep track of the ``call stack'' for $\aval$. Thus how can
it determine the calling function when coming back from a
recursive call? In fact, due to the copyless behavior, this information
can be easily determined without a stack. Indeed, assume that the \tDT{}
has just finished executing $\aval(\beta,i{-}1,a)$, then it moves right:
\begin{itemize}
\item if $a = \out$, then the call to $\aval$ was done in $\simu$
which had computed the value $\lambda^x_i(\out) = \out \beta$.
In this case, the \tDT{}  pursues the execution of $\out$ in $i{+}1$;
\item otherwise there exists exactly one $\reg \in \Regs$
such that $a$ occurs in $\alpha \defined \lambda^x_i(\reg)$
(and this value can be determined).
Then $\aval(\beta,i{-}1,a)$ was called while executing
$\aval(\alpha,i,\reg)$ on position $i$
and $\beta=\lambda^x_i(a)$.
Since $a = \alpha[j]$ for a unique $j$,
then the index $j$ can be determined
and thus the \tDT{} can pursue the computation.
\end{itemize}
Furthermore if $x \in \Dom{f}$, this machine visits its whole input.

\subsection{From $K$-bounded \DSST{}s to copyless \DSST{}s}

Let $\trans = (A,B, Q,q_0, \delta, \Regs, \out, \lambda)$ be a $K$-bounded
\DSST{}, we show how to transform it in an equivalent copyless \DSST{}.
The proof techniques are adapted from those for \DSST{}
over finite words \cite{dartois2016aperiodic, doueneau2020register}
or $\omega$-streaming string transducers \cite{alur2012regular}.
However, these transformations usually add extra features to
the \DSST, such as non-determinism
or lookaheads, which we avoid here. Indeed, such features allow
to check $\omega$-regular properties of the input,
and it is precisely what we intend to get rid of in this paper.

\subparagraph*{Generic proof ideas.} 
In order to transform $\trans$ into a copyless \DSST,
the natural idea is to keep $K$ copies of each register.
However, we cannot maintain $K$ copies
all the time: suppose that $\reg$
is used to update both $\reg_1$ and $\reg_2$.
If we have $K$ copies of $\reg$,
we cannot produce, in a copyless way, $K$ copies of $\reg_1$
and $K$ copies of $\reg_2$.

This issue is solved
as follows. Recall that $q^x_i$ (resp. $\lambda^x_i$) is the state
reached by $\trans$ after reading $x[1{:}i]$ (resp. the transition
applied when reading $x[i]$) on input $x \in A^\omega$.
Let $\Reggs \defined \Regs \smallsetminus \{\out\}$
and $\reg \in \Reggs$.
If $i \ge 0$ and $x \in A^\omega$ are
such that $q^x_i$ is defined,
we want to maintain $\copies^x_i(\reg)$ copies of the value $\cro{\reg}^x_i$,
where $\copies_i^x(\reg) \le K$ is the number of times
that $\reg$ will be used in $\out$ after position $i$.
Formally we define the following:

\begin{definition} Let  $x \in A^\omega$ and $i \ge 0$
such that  $q^x_i$  is defined.
Given $\reg \in \Reggs$, we let:
\begin{equation*}
\copies^x_i(\reg) \defined \max
\{ | \lambda^x_{i+1} \circ \cdots \circ \lambda^x_j(\out)|_{\reg}
: j \ge i \tnorm{ and } \lambda^x_j \tnorm{ is defined} \}.
\end{equation*}
\end{definition}

\begin{claim} \label{claim:cxi}
$\copies^x_i(\reg) \le K$ since $\trans$ is $K$-bounded.
\end{claim}
We now describe an inductive relation for
the $\copies^x_i(\reg)$. Intuitively, Lemma~\ref{lem:calcule-cix}
means that if  $\copies^x_i(\reg)$ copies of $\reg$ will be needed,
then in the next transition these copies can be reparted
between the registers. We postpone
the proof of Lemma~\ref{lem:calcule-cix}
to Subsubsection~\ref{sssec:proof:cix}.

\begin{lemma} \label{lem:calcule-cix} Let  $x \in A^\omega$ and $i \ge 0$
such that $q^x_{i+1}$ is defined. 
Then for all $\reg \in \Reggs$:
\begin{equation*}
\copies^x_i(\reg) = |\lambda^x_{i+1}(\out)|_{\reg} + 
\sum_{\regg \in \Reggs} |\lambda^x_{i+1}(\regg)|_{\reg} \times \copies^x_{i+1}(\regg).
\end{equation*}
\end{lemma}
However, this number $c_i^x(\reg)$ cannot be determined
after reading only $x[1{:}i]$. It requires some information
about $x[i{+}1{:}]$ (that is is typically why we would
need a lookahead). Thus our copyless \DSST{} will have to
memorize a finite forest which
describes the possible non-deterministic choices
done to determine the values of $\copies^x_i(\reg)$.
The copyless \DSST{} will also keep track of the substitutions
applied along the branches of this forest.

\subparagraph*{Structure of the proof.} 
In order to make the construction of a copyless \DSST{}
more understandable, we first describe its behavior
in a high-level algorithmic fashion in Subsubsection~\ref{sssec:dsst-info}.
Then we justify in Subsubsection~\ref{sssec:dsst-how} that this
algorithm can be implemented by a copyless \DSST{}.
We finally give in Subsubsection~\ref{sssec:proof:cix} the proof
of  Lemma~\ref{lem:calcule-cix}.

\subsubsection{Algorithmic description of the copyless \DSST}

\label{sssec:dsst-info}

We describe informally  the 
behavior of the copyless \DSST, denoted $\trans'$. The goal of
this section is to show the main ideas of the construction,
without dealing with implementation details
which shall be explained in Subsubsection~\ref{sssec:dsst-how}.

The key idea is that $\trans'$, after reading position 
$i \ge 0$ of $x \in A^\omega$, keeps track of a \emph{decomposition}
(see Definition~\ref{def:decompose})
which describes the substitution $\cro{\cdot}^x_i$
as a composition of $K$-bounded substitutions
along the branches of a forest.
We use the usual vocabulary for describing trees
and forests: a \emph{leaf} is a node which has children;
the \emph{depth} of a node is defined inductively when starting from
the nodes with depth $0$ (also called the \emph{roots});
the \emph{height} of the forest is the maximal depth of a node.

\begin{definition}[Decomposition] \label{def:decompose}
Let  $x \in A^\omega$ and $i \ge 0$ be
such that $\cro{\cdot}^x_i$  is defined.
A \emph{decomposition} of $\cro{\cdot}^x_i$ is:
\begin{enumerate}
\item a sequence $0 = i_0 < \cdots < i_m = i$
of positions;
\item a sequence $\sigma_1,\dots, \sigma_m \in \subst{\Reggs}{B}$
of $K$-bounded substitutions such that for all 
$1 \le \ell \le m$, $\sigma_\ell =\lambda^x_{i_{\ell-1}+1} \circ \cdots \circ \lambda^x_{i_{\ell}}$
restricted to $\Reggs$;
\item a finite forest $F$ whose nodes 
are labelled by functions $g: \Reggs \fonc [0{:}K]$ such that:
\begin{enumerate}
\item \label{cond:tree-a} all the leaves have same depth $m$ and distinct labels;
\item  \label{cond:tree-c} there exists a leaf whose label is $\copies^x_i$;
\item \label{cond:tree-b} if $h$ labels a node of depth $1 \le \ell \le m$
and $g$ labels its parent, then for all $\reg \in \Reggs$:
\begin{equation}
\label{eq:cond:tree-b}
g(\reg) =  \sum_{\regg \in \Reggs} |\sigma_{\ell}(\regg)|_{\reg} \times h(\regg).
\end{equation}
\end{enumerate}
\end{enumerate}
\end{definition}

\begin{example} \label{ex:decompo}
Assume that  $\Reggs = \{\reg, \regg\}$, $K = 5$ and
$\lambda^x_1(\reg) = a \reg$, $\lambda^x_1(\regg) = b\regg$;
$\lambda^x_2(\reg) = \regg a \regg$, $\lambda^x_2(\regg) = \reg b$;
$\lambda^x_3(\reg) = \reg $, $\lambda^x_3(\regg) = \regg$.
We describe a possible decomposition of $\cro{\cdot}^x_3$:
\begin{enumerate}
\item the sequence $i_0 = 0$, $i_1 = 1$, $i_2=3$;
\item  the substitutions $\sigma_1, \sigma_2$ defined
by $\sigma_1(\reg)  = a\reg $, $\sigma_1(\regg) =\reg b$ and
 $\sigma_2(\reg) = \regg a \regg$,
 $\sigma_2(\regg)  = \reg b$;
\item the forest depicted in Figure~\ref{fig:guess-forest}
(note that Equation~\ref{eq:cond:tree-b} holds).
\end{enumerate}
\end{example}

\begin{figure}[h!]

\centering
\begin{tikzpicture}{scale=1}

	\node[left] at (0,2) {$\substack{\reg \mapsto 5 \\ \regg \mapsto 0}$};	
	\draw[->] (0,2) -- (1.2,1.5) node[midway,below] {$\substack{\sigma_1}$};
	\node[left] at (2,1.5) {$\substack{\reg \mapsto 1 \\ \regg \mapsto 4}$};
	\draw[->] (0,2) -- (1.2,2.5) node[midway,above] {$\substack{\sigma_1}$};
	\node[left] at (2,2.5) {$\substack{\reg \mapsto 3 \\ \regg \mapsto 2}$};
	\draw[->] (2,1.5) -- (5.2,1.5) node[midway,above] {$\substack{\sigma_2}$};
	\node[left] at (6,1.5) {$\substack{\reg \mapsto 2 \\ \regg \mapsto 1}$};	
	\draw[->] (2,2.5) -- (5.2,2.5) node[midway,above] {$\substack{\sigma_2}$};
	\node[left] at (6,2.5) {$\substack{\reg \mapsto 1 \\ \regg \mapsto 3}$};

	\node[left] at (0,0) {$\substack{\reg \mapsto 3 \\ \regg \mapsto 0}$};	
	\draw[->] (0,0) -- (1.2,0) node[midway,above] {$\substack{\sigma_1}$};
	\node[left] at (2,0) {$\substack{\reg \mapsto 1 \\ \regg \mapsto 2}$};	
	\draw[->] (2,0) -- (5.2,0) node[midway,above] {$\substack{\sigma_2}$};
	\node[left] at (6,0) {$\substack{\reg \mapsto 1 \\ \regg \mapsto 1}$};

\end{tikzpicture}

\caption{\label{fig:guess-forest}
The forest of the decomposition described in Example~\ref{ex:decompo}
(roots are on the left).}

\end{figure}
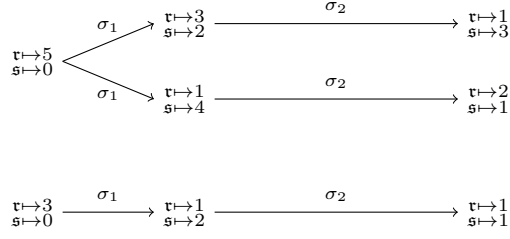

\subparagraph*{Information stored by $\trans'$.}
The configuration of $\trans'$ when in position $i \ge 0$ of $x \in A^\omega$
stores some decomposition $D^x_i$ of $\cro{\cdot}^x_i$ as follows,
with $m \le L \defined (K+1)^{|\Reggs|}$:
\begin{enumerate}
\item  the sequence of positions $i_0 < \cdots < i_m$
is not stored by $\trans'$ (its existence
is an invariant which will be preserved along the computation);
\item the sequence of substitutions $\sigma_1, \cdots, \sigma_m$
is stored depending on the forest $F$, see below;
\item the forest $F^x_i$ associated to $D^x_i$.
$F^x_i$  has depth $m \le L$, hence its structure
is a bounded information which can be stored in the state.
Let us now explain how we deal with the substitutions.
If $g$ labels a node of depth $1 \le \ell \le m$ in $F^x_i$, then
$\trans'$ ``stores $g(\reg)$ virtual copies of  $\sigma_{\ell}(\reg)$''
for all $\reg \in \Reggs$. What we mean by ``storing
virtual copies'' is explained in Subsubsection~\ref{sssec:dsst-how}
(the ideas  are those of \cite{dartois2016aperiodic, doueneau2020register}).
\end{enumerate}
Furthermore, the register $\out$ of $\trans'$ will
contain $\cro{\out}^x_i$.

\subparagraph*{Initialization of the decomposition.}
When $i=0$, the configuration of
$\trans'$ describes the following decomposition $D^x_0$
of $\cro{\cdot}^x_0$ with $m = 0 \le L$:
\begin{enumerate}
\item the sequence $i_0 = 0$ (not stored explicitly);
\item no substitutions since $m = 0$;
\item the forest of depth $0$ (i.e. it has only roots)
with $L = (K+1)^{|\Reggs|}$ nodes labelled by
all possible functions $\Reggs \fonc [0{:}K]$.
Conditions~\ref{cond:tree-a} and \ref{cond:tree-b}
of Definition~\ref{def:decompose}
are obviously true. Condition~\ref{cond:tree-c}
follows from Claim~\ref{claim:cxi}.
\end{enumerate}

\subparagraph*{Updates of the decomposition.}
Assume that $\lambda^x_{i+1}$ is defined (otherwise $\trans'$
gets blocked), and that $\trans'$ stores
a decomposition $D^x_i$ of $\cro{\cdot}^x_i$.
Then $\trans'$ transforms $D^x_i$ into a decomposition
of $\cro{\cdot}^x_{i+1}$ as follows:
\begin{enumerate}
\item the new sequence of positions is $i_0 = 0 < \dots < i_m = i < i_{m+1} = i{+}1$.
If $m{+}1 > L$, we first completely build this new decomposition,
and then we apply the \emph{merging operation} detailed in the
next paragraph in order to reduce  the depth to $L$;
\item the sequence $\sigma_1, \dots, \sigma_m, \sigma_{m+1}$
where $\sigma_{m+1} \defined \lambda^x_{i+1}$.
Note that the substitution $\sigma_{m+1} $ is $K$-bounded since $\trans$ is so;
\item before describing the new forest built from $F^x_i$, let us 
give some intuitions. Recall that $\sigma_{m+1}(\out) = \out \alpha$
for some $\alpha \in (B \uplus \Reggs)^*$.
We want $\trans'$ to add $\cro{\alpha}^x_i$ to $\out$ when reading $x[i{+}1]$.
For this purpose, it needs to determine the value
$\lambda^x_0 \circ \sigma_1 \circ \cdots \circ \sigma_m(\reg) = \cro{\reg}^x_i$
for each $\reg \in \Reggs$ which occurs in $\alpha$. We thus define
the functions $\used_{\ell}: \Reggs \fonc [0{:}K]$
for $0 \le \ell \le m$, which describe ``how many virtual copies'' of the
$\sigma_{\ell}(\reg)$ will be consumed to
compute $\cro{\alpha}^x_{i+1}$. They are built by a decreasing
induction:
\begin{itemize}
\item $\used_m(\reg) \defined  |\sigma_{m+1}(\out)|_{\reg}$ for all $\reg \in \Reggs$;
\item $\used_{\ell}(\reg) \defined 
\sum_{\regg \in \Reggs} |\sigma_{\ell{+}1}(\regg)|_{\reg} \times \used_{\ell +1}(\regg)$
for all $\reg \in \Reggs$, if $0 \le \ell \le m{-}1$.
\end{itemize}

We now want to subtract the $\used_{\ell}$ to the labels
of the nodes, since they describe the number of copies
that we need to ``consume'' to compute
$\lambda^x_0 \circ \sigma_1 \circ \cdots \circ \sigma_m(\alpha)$.

\begin{claim} \label{claim:tilde} Assume that
$h$ labels a node of depth $1 \le \ell \le m$ in 
$F^x_i$ and let $g$ label its parent.
Define $\new{h} \defined h {-} \used_{\ell} $
 and $\new{g} \defined g {-} \used_{\ell-1}$, then:
\begin{itemize}
\item if $\new{h} \ge 0$, then $\new{g} \ge 0$;
\item$\new{g}$ and $\new{h}$
verify Equation~\ref{eq:cond:tree-b} in Condition~\ref{cond:tree-b}.
\end{itemize}
\end{claim}

\begin{proof} By Condition~\ref{cond:tree-b} on $F$ we have  
$
g(\reg) =  \sum_{\regg \in \Reggs} |s_{\ell+1}(\regg)|_{\reg} \times h(\regg)
$,
therefore we get 
$g(\reg) - \used_{\ell}(\reg)=  \sum_{\regg \in \Reggs}
|s_{\ell+1}(\regg)|_{\reg} \times (h(\regg) - \used_{\ell+1}(\regg))$
by definition of $\used$.
\end{proof}

We now describe in three steps how to build the new forest from $F^x_i$:
\begin{description}
\item[Step 1: consuming the $\used_{\ell}(\reg)$.]
We replace each label $g$ in $F^x_i$ by
$\new{g}$ from Claim~\ref{claim:tilde},
which may create negative labels.
But since one leaf is labelled by $\copies^x_i$,
then by Lemma~\ref{lem:calcule-cix} we see that 
$\new{\copies^x_i} \ge 0$. Hence by Claim~\ref{claim:tilde},
there is a branch whose labels are nonnegative.
The copyless \DSST{} shall use the ``virtual copies of the $\sigma_{\ell}(\reg)$''
stored along this branch to output $\lambda^x_i \circ \sigma_1 \circ \cdots \circ \sigma_m(\alpha)$
in a copyless fashion (see Subsubsection~\ref{sssec:dsst-how});
\item[Step 2: adding level $m{+}1$.]
For each leaf now labelled by $\new{g}$, we
create several children labelled by the
$h: \Reggs \fonc [0{:}K]$ such that for all
$\reg \in \Reggs$ we have:
\begin{equation*}
\new{g}(\reg) =  \sum_{\regg \in \Reggs} |\sigma_{m+1}(\regg)|_{\reg} \times h(\regg).
\end{equation*}
For all $\reg \in \Regs$ and all created leaf labelled by $h$, the \DSST{}
creates $h(\reg)$ virtual copies of $\sigma_{m+1}(\reg) \in (\Reggs \uplus B)^*$
(which is a bounded information). Note that two 
created leaves cannot have the same label (otherwise
it would be the case for their parents).
Finally by Lemma~\ref{lem:calcule-cix}
 the node labelled by $\new{\copies^x_i}$ has a leaf labelled by $\copies^x_{i+1}$;

\item[Step 3: removing errors.] Now it remains to deal with the fact that
some nodes may have negative labels,
and some leaves may have depth $\ell < m{+}1$. We thus
remove all the nodes labelled by functions
which take negative values, and their descendants. Finally, we trim the 
resulting forest by removing all nodes which are not ancestors
of some leaf of depth $m{+}1$ (i.e. a leaf which has created in Step~2).
It is easy to see that conditions~\ref{cond:tree-a}, \ref{cond:tree-b} and \ref{cond:tree-c}
now hold.
\end{description}

\end{enumerate}

\subparagraph*{Merging operation: removing single children.}
Let us now explain how to reduce the height of the decomposition
obtained in the previous paragraph when  $m{+}1 > L$.

\begin{claim} If $m{+}1 > L$, there exists
 $1 \le \ell \le m$ such that all nodes
of depth $\ell$ have exactly one children
(in the forest build in the previous paragraph).
\end{claim}

\begin{proof} Assume that for all $1 \le \ell \le m$, some node of depth
$\ell$ has at least two children. Since $m{+}1 > L$ and all leaves have
the depth $m{+}1$, then our forest has
more than $L$ leaves, which contradicts the fact that
two leaves cannot have the same label.
\end{proof}
The main idea is to ``merge'' $\sigma_{\ell}$ and $\sigma_{\ell + 1}$
(which exists since $1 \le \ell \le m$).
The decomposition of the previous paragraph is updated
as follows to build $D^x_{i+1}$
\begin{enumerate}
\item the positions  become $i_0 < i_1 < \cdots < i_{\ell-1} < i_{\ell+1} < \cdots <  i_{m+1}$
\item  the substitutions become
$\sigma_1, \dots, \sigma_{\ell-1}, \sigma_{\ell} \circ \sigma_{\ell +1}, \sigma_{\ell+2}, \dots, \sigma_{m+1}$
(note that $\sigma_{\ell} \circ \sigma_{\ell +1}$ is $K$-bounded
and corresponds to the restriction of $\lambda^x_{i_{\ell-1}+1} \circ \cdots
\circ \lambda^x_{i_{\ell+1}}$);
\item before modifying the forest, we first show the following:
\begin{claim} \label{claim:deux-somm}
Assume that $h$ labels a node of depth $\ell{+}1$,
and let $g$ be the label of its grandparent
(for the forest built by the previous paragraph).
Then for all $\reg \in \Reggs$:
\begin{equation*}
g(\reg) =  \sum_{\regg \in \Reggs} |\sigma_\ell \circ \sigma_{\ell+1}(\regg)|_{\reg} \times h(\regg).
\label{eq:cond:tree-b}
\end{equation*}
\end{claim}

\begin{proof} By Claim~\ref{claim:cal-somm}, we have:
\begin{equation*}
\sum_{\regg \in \Reggs} |\sigma_\ell \circ \sigma_{\ell+1}(\regg)|_{\reg} \times h(\regg)
= \sum_{\regg \in \Reggs} \sum_{\reggg \in \Reggs}
|\sigma_\ell(\reggg)|_{\reg} \times |\sigma_{\ell+1}(\regg)|_{\reggg} \times h(\regg)
\end{equation*}
The result follows by permuting the sums and 
using condition~\ref{cond:tree-b} in $\ell$ and $\ell{+}1$.
\end{proof}

We thus transform the forest by merging
each node of depth $\ell$ with its single child
of depth $\ell {+}1$  (labelled by some $g$),
and labelling the resulting node by $g$.
Condition~\ref{cond:tree-b} still holds because
of Claim~\ref{claim:deux-somm}. Note that $\trans'$
also has to compute and store several copies
of $\sigma_\ell \circ \sigma_{\ell+1}(\reg)$ for $\reg \in \Regs$.
Subsubsection~\ref{sssec:dsst-how} describes how to 
perform this update in a copyless fashion when starting from
copies of $\sigma_\ell(\regg)$ and $\sigma_{\ell+1}(\reg)$.
\end{enumerate}

Finally $\trans'$ has the same domain than
$\trans$, since it gets blocked  at position $i \ge 0$
if $\lambda^x_i$ is undefined, and otherwise it stores $\cro{\out}^x_i$ 
in its output $\out$.

\subsubsection{Implementation details}

\label{sssec:dsst-how}

In the previous subsection, we have described
the behavior of the copyless \DSST{} without detailing
how, for each $g$ labelling a node of depth $1 \le \ell \le m$
and $\reg \in \Reggs$, this machine could
``store $g(\reg)$ virtual copies of $\sigma_{\ell}(\reg) \in (\Reggs \uplus B)^*$''.
We now explain it in detail.

\subparagraph*{Storing $K$-bounded substitutions.}
We describe how a copyless \DSST{} can
store $K$-bounded substitutions (the ideas
are mainly those of \cite{doueneau2020register}).
Let $\sigma \in \subst{\Reggs}{B}$ be a $K$-bounded substitution,
then for all $\reg \in \Reggs$ there exists
$n \le K|\Reggs|$ such that
$\sigma(\reg) = \alpha_0 \regg_1 \alpha_1 \cdots \regg_n \alpha_n$
with $\alpha_i \in B^*$, $\reggg_i \in \Reggs$.
We mainly have two informations in this expression:
\begin{itemize}
\item the sequence $\regg_1 \regg_2 \cdots \regg_n$
which describes where the former
registers must be used;
\item the sequence $\alpha_1 \cdots \alpha_n$ of 
(unbounded) words from $B^*$. Each of them
must be stored in a register. Furthermore,
we must keep track of the ``mapping'' between the
registers and the $\alpha_i$, which is a bounded information.
\end{itemize}
We can now explain what we mean by ``storing $g(\reg)$
virtual copies of $\sigma(\reg) \in (\Reggs \uplus B)^*$'':
it means that we store $g(\reg)$ copies (in $g(\reg)$
distinct registers) of each $\alpha_i$.
Note that if $g(\reg) \le K$, we need at most
$K^2|\Reggs|$ registers. The sequence 
$\regg_1 \regg_2 \cdots \regg_n$ is stored the state.

\subparagraph*{Composing $K$-bounded substitutions.}
With this representation, $\trans'$ is able to simulate 
the composition of two $K$-bounded substitutions,
when their composition is itself $K$-bounded (which is
always the case in the above merging operation).

\begin{claim} \label{claim:copyless-merge}
Assume that $\sigma,\sigma' \in \subst{\Reggs}{B}$ and $\sigma \circ \sigma'$
are $K$-bounded and that $\trans'$ stores:
\begin{itemize}
\item $g(\reg)$ virtual copies of $\sigma(\reg)$ for $\regg \in \Reggs$, where $g: \Reggs \fonc [0{:}K]$;
\item $g'(\regg)$ virtual copies of $\sigma'(\regg)$ for $\regg \in \Reggs$, where $g': \Reggs \fonc [0{:}K]$;
\end{itemize}
such that $g(\reg) =  \sum_{\regg \in \Reggs} |\sigma'(\regg)|_{\reg} \times g'(\regg)$
Then there exists a copyless update of $\trans'$ which
allows to store $g'(\regg)$ copies of $\sigma \circ \sigma'(\regg)$ for $\regg \in \Reggs$.
\end{claim}

\begin{proof} In order to compute $g'(\regg)$ copies of $\sigma \circ \sigma'(\regg)$,
we exactly need to used $ |\sigma'(\regg)|_{\reg} \times g'(\regg)$ copies of $\reg$.
The result follows by summing over all $\regg \in \Reggs$.
\end{proof}
Claim~\ref{claim:copyless-merge} justifies how $\trans'$ can update
is information in a copyless fashion when performing the
merging operation. We still have to justify how $\trans'$ can
compute $\lambda^x_0 \circ \sigma_1 \circ \cdots \circ \sigma_m(\alpha)$ when it has
to add something in $\out$ in the above Step~1. As for the proof
of Claim~\ref{claim:copyless-merge}, we exactly need 
to have $\used_{\ell}(\reg)$ copies of $\sigma_{\ell}(\reg)$
for each $\reg \in \Regs$. This copies are taken
along some root-to-leaf branch of the forest where removing the $\used_{\ell}$
does not create negative labels $\new{g}$
(such a branch exists because $\copies^x_i$ labels
a leaf, as shown in Step~1).
The remaining labels  $\new{g}$ exactly correspond to
the copies that were not used, hence they are still stored.

\begin{remark} To produce $\cro{\alpha}^x_i$, we only need to consume the copies along
one branch of $F^x_i$. However, to maintain a forest which is consistent with
our decomposition $D^x_i$, we remove copies along all the branches, even if only one branch
is truly used.
\end{remark}

\subsubsection{Proof of Lemma~\ref{lem:calcule-cix}}

\label{sssec:proof:cix}

We first give a way to count the copies
obtained when composing two substitutions.

\begin{claim} \label{claim:cal-somm} Let $\sigma,\sigma' \in \subst{\Regs}{B}$, then
for all $\reg, \regg \in \Regs$, $|\sigma \circ \sigma'(\reg)|_{\regg} 
= \sum_{\reggg \in \Regs} |\sigma(\reggg)_{\regg} |\times |\sigma'(\reg)|_{\reggg}$.
\end{claim}

\begin{proof}
We have $s'(\reg) = \alpha_0 \reggg_1 \alpha_1 \reggg_2 \dots \reggg_n \alpha_n$ 
for some words $\alpha_i \in B^*$ and registers $\reggg_i \in \Regs$.
Therefore $|\sigma \circ \sigma'(\reg)|_{\regg} = \sum_{i=1}^n |\sigma(\reggg_i)|_{\regg}
= \sum_{\reggg \in \Regs} |\sigma'(\reg)|_{\reggg} \times |\sigma(\reggg)|_{\regg}$.
\end{proof}
We then note that since $\out$ is always updated 
under the form $\out \alpha$, the number of copies of a given register in 
$\out$ can only grow.

\begin{claim} Given $\reg \in \Reggs$ and $i \ge 1$ such that $\lambda^x_i$ is defined,
the function which maps $j \ge i{-}1 \mapsto |\lambda^x_i \circ \cdots \circ \lambda^x_j(\out)|_{\reg}$
is increasing (on its domain).
\end{claim}

\begin{proof} By Claim~\ref{claim:cal-somm} we have
$|\lambda^x_i \circ \cdots \circ \lambda^x_{j+1}(\out)|_{\reg}
\ge |\lambda^x_i \circ \cdots \circ \lambda^x_{j}(\out)|_{\reg} \times |\lambda^x_{j+1}(\out)|_{\out}$.
\end{proof}
We are now ready to show Lemma~\ref{lem:calcule-cix}.
If $j \ge i{+}1$ is such that $\cro{\cdot}^x_{j}$ is defined, we have
the following for all $\reg \in \Reggs$ by Claim~\ref{claim:cal-somm}:
\begin{equation*}
|\lambda^x_{i+1} \circ \cdots \circ \lambda^x_j (\out)|_{\reg} = 
|\lambda^x_{i+1}(\out)|_{\reg}  \times 1 + 
\sum_{\regg \in \Reggs} |\lambda^x_{i+1}(\regg)|_{\reg} \times
|\lambda^x_{i+2}\circ \cdots \circ \lambda^x_j(\out)|_{\regg}.
\end{equation*}
Now let $j_0 \ge i{+}1$ be such that $|\lambda^x_{i+1} \circ \cdots \circ \lambda^x_{j_0} (\out)|_{\reg}$
is maximal. Since the function $j \mapsto |\lambda^x_{i+1} \circ \cdots \circ \lambda^x_j (\out)|_{\reg} $
is increasing by Claim~\ref{claim:cal-somm}, we get
for all $j \ge j_{0}$ (and when defined):
\begin{equation*}
\copies^x_i(\reg) = |\lambda^x_{i+1} \circ \cdots \circ \lambda^x_j (\out)|_{\reg} 
= |\lambda^x_{i+1}(\out)|_{\reg} + 
\sum_{\regg \in \Reggs} |\lambda^x_{i+1}(\regg)|_{\reg} \times
|\lambda^x_{i+2}\circ \cdots \circ \lambda^x_j(\out)|_{\regg}.
\end{equation*}
And by the same argument of Claim~\ref{claim:cal-somm}
applied to the $|\lambda^x_{i+2}\circ \cdots \circ \lambda^x_j(\out)|_{\regg}$,
we conclude:
\begin{equation*}
\copies^x_i(\reg) =
|\lambda^x_{i+1}(\out)|_{\reg} + 
\sum_{\regg \in \Reggs} |\lambda^x_{i+1}(\regg)|_{\reg} \times
\copies^x_{i+1}(\regg).
\end{equation*}

\section{Proof of Proposition~\ref{prop:dom-reg}}

We show the stronger result which follows.
\begin{lemma}
\label{lem:dom-reg}
If $f$ is deterministic
regular, then $\Dom{f}$ is Büchi deterministic.
Conversely, if $L \subseteq \Dom{f}$ is Büchi deterministic,
then $f$ restricted to $L$ is deterministic regular.
\end{lemma}

\begin{proof}
We first show that the domain of a deterministic regular function
is accepted by a Büchi deterministic automaton. For this let
$\trans$ be a copyless \DSST{}
computing a deterministic regular function
$f: A^\omega \parfonc B^\omega$.
We describe a deterministic
Büchi automaton which follows the states of
$\trans$ in a deterministic way (in particular, it gets blocked
if $\trans$ gets blocked). It also memorizes
for each $\reg \in \Regs$ if $\cro{\reg}^x_i = \movi$ or not (this is a bounded
information which can be updated with a bounded memory).
The automaton reaches an accepting state when $\trans$
adds a non-empty value into $\out$.
It is not hard see that it exactly recognizes $\Dom{f}$.

Now let $\auto$ be a deterministic Büchi automaton
accepting a language $L \subseteq \Dom{f}$.
We build the product of $\trans$ and $\auto$ and follow
the transition functions of both machines. The register updates
are the same as in $\trans$ for all $\reg \in \Regs \smallsetminus \{\out\}$.
For $\out$, we store temporarily the values which are added in a new register $\out'$,
which is added to $\out$ only when leaving an accepting state of $\auto$.
Thus the output (which is always a prefix
of the output of $\trans$) is infinite if and only if the input is accepted by $\auto$.
\end{proof}

\section{Proof of Theorem~\ref{theo:prec-reg}}

Let us consider a restricted \oNT{} computing  a
function $f: A^\omega \parfonc B^\omega$
and a copyless  \DSST{} computing a deterministic regular
function $g: B^\omega \parfonc C^\omega$.
By making their cascade product
(a standard construction for composing one way machines),
we can easily build a copyless \emph{restricted non-deterministic streaming string
transducer} computing $h \defined g \circ f$.

\begin{definition} \label{def:nsst}
A \emph{restricted non-deterministic streaming string transducer}
\emph{(restricted \NSST)} $\ntrans = (A,C, Q,I, \Delta, \Regs, \out, \lambda)$
consists of:
\begin{itemize}
\item a finite input (resp. output) alphabet $A$ (resp. $C$);
\item a finite set of states $Q$ with $I \subseteq Q$ initial;
\item a transition relation $\Delta \subseteq Q \times A \times Q$;
\item a finite set of registers $\Regs$ with a distinguished output register $\out \in \Regs$;
\item an update function $\lambda: \Delta \fonc \subst{\Regs}{C}$
such that for all $(q,a,q) \in \Delta$:
\begin{itemize}
\item $\lambda(q,a,q')(\out) = \out \cdots$;
\item there is no other occurence of $\out$ in $\{ \lambda(q,a,q')(\reg): \reg \in \Regs\}$.
\end{itemize}
\end{itemize}
\end{definition}
Given $x \in A^\omega \cup A^*$, a  (resp. initial, final, accepting)
run of $\ntrans$ labelled by $x$ is a (resp. initial, final, accepting)
run of the underlying one-way automaton $(A,Q,I, F, \Delta)$
where $F = Q$ (all states are accepting). Given an initial run
$\rho = q_0 \fonc q_1 \fonc \cdots$ composed of $n \in \Nat \uplus \{\infty\}$
transitions, we define for $0 \le i \le n$
the substitution $\lambda^{\rho}_i \in \subst{\Regs}{C}$
by $\lambda^{\rho}_0(\reg) = \movi$ for all $\reg \in \Regs$,
and $\lambda^{\rho}_i \defined \lambda(q_i,x[i])$ for $i \ge 1$
We define the substitution $\cro{\cdot}^{\rho}_i \in \subst{\Regs}{B}: \Regs \fonc B^*$
by $\cro{\cdot}^{\rho}_i = \lambda^{\rho}_0 \circ  \cdots \circ \lambda^{\rho}_i$.
By construction one has that $\cro{\out}^{\rho}_i \pref \cro{\out}^{\rho}_{i+1}$
(when defined).

The restricted \NSST{}  computes a function $h$ defined as follows.
Let $x \in A^\omega$ be such that there exists a unique accepting run
$\rho$ labelled by $x$, and such that $|\cro{\out}^{\rho}_i| \fonc \infty$.
Then $h(x) \defined  \bigvee_i \cro{\out}^x_i$.
Otherwise $h(x)$ is undefined.
We say that the restricted \NSST{} is \emph{copyless} if all the
$\lambda^{\rho}_i \circ \cdots \circ \lambda^{\rho}_j$
are copyless when $\rho$ is an initial run.
Equivalently (up to trimming), $\lambda(q,a,q')$
is always copyless when defined.

\subparagraph*{Building a $1$-bounded \DSST{}.}
We now build a $1$-bounded \DSST{}
$\trans$ which computes an extension of $h$
(i.e. whose domain contains $\Dom{h}$
and which coincides with $h$ on this set).
We shall deal with the domain in the next paragraph.
The general idea is to perform a subset
construction on $\ntrans$, while keeping track of
the forest of all initial runs and the outputs
that were produced along its branches.
Without loss of generality, we assume that
$I = \{q_0\}$ is a singleton. Thus the forest
of initial runs is a tree.
If a single accepting run exists, then the first node
of this tree which has at least two children
must move forward infinitely  often, which enables
to produce the whole output of $\trans$.
The fact that we obtain a $1$-bounded transducer
and not a copyless one follows from the fact
that the same value can be re-used when creating
new branches (but the branches will never be merged later).

\begin{definition}[Tree of runs] \label{def:Txi}
Given $x \in A^\omega$ and $i \ge 0$ we define
the \emph{tree of runs} $T^x_i$ as a tree whose nodes
are labelled by tuples $(i_1, \rho, i_2)$ with $0 \le i_1 \le i_2$
and $\rho$ is a run labelled  by $x[i_1{+}1{:}i_2]$.
It is built by induction as follows:
\begin{itemize}
\item $T^x_0$ consists of a single node labelled by
$(0,\rho,0)$ where  by $\rho = q_0$
is a run labelled by $\movi$;
\item if $T^x_i$ is built, we obtain $T^x_{i+1}$ by doing the following
operations where $a \defined x[i{+}1]$:
\begin{description}
\item[New transition:] below each leaf labelled
by $(i_1, \rho, i)$ with  $\rho = q_{i_1} \fonc \cdots \fonc q_{i}$,
and for all transitions $(q_i, a,q_{i+1}) \in \Delta$,
we add a leaf labelled by $(i, q_{i} \fonc q_{i+1}, i{+}1$);
\item[Removing ambiguity:] for all $q' \in Q$, if at least two
 leaves created by the previous step are such that $q' = q_{i+1}$, then we remove
all the created leaves such that $q' = q_{i+1}$
(indeed, if there exists an infinite run starting on $q'$
and labelled by $x[i{+}2{:}]$, then $x \not \in \Dom{f}$ since it
has two accepting runs);
\item[Trimming the tree:] we remove all nodes in $T^x_{i+1}$
which are not an ancestor of a new leaf;
\item[Merging single nodes:] for any node 
(labelled by $(i_1, q_{i_1}\fonc \cdots \fonc q_{i_2}, i_2)$)
which has only one child (labelled by $(i_2, q_{i_2} \fonc \cdots
\fonc q_{i_3}, i_3)$) we merge these two nodes
together with common label
$(i_1, q_{i_1}\fonc \cdots \fonc q_{i_2} \fonc \cdots \fonc q_{i_3}, i_3)$.
\end{description}
\end{itemize}
\end{definition}
The following properties immediately follow from the construction.

\begin{claim} There exists $L \ge 0$, such
that for all $x \in A^\omega$, $i \ge 0$, $T^x_i$ has at most $L$ nodes.
\end{claim}

\begin{claim} If $x \in \Dom{f}$, there exists a root-to-leaf branch of $T^x_i$
whose labels describe the beginning of the unique accepting run 
labelled by $x$.
\end{claim}
We also note that if $x \in \Dom{f}$, the run stored in the root of $T^x_i$
becomes longer and longer.

\begin{claim}\label{claim:calinf} If $x \in \Dom{f}$ and $(0, \rho_i, r_i)$ is the label of the 
root of $T^x_i$, then  $r_i \fonc \infty$.
\end{claim}
\begin{proof} Assume that $r_i$ is ultimately constant.
Then the root of $T^x_i$ always has at least two children, which implies
that there exists two distinct infinite runs labelled by $x$.
\end{proof}

Let us finally describe how $T^x_i$ is used to build 
a $1$-bounded \DSST{} $\trans$. When in position $i \ge 0$
of its input $x \in A^\omega$, the \DSST{} stores 
the following information:
\begin{itemize}
\item the structure of $T^x_i$, i.e. the tree  without its labels (stored in the state);
\item for each leaf of $T^x_i$ labelled by $(i_1, q_{i_1} \fonc \cdots \fonc q_i, i)$,
the last state $q_i$ (stored in the state);
\item for each node of $T^x_i$ labelled by $(j,\rho_n, j')$,
let $\rho \defined \rho_1 \cdots \rho_n$ be the run
labelling the branch of $T^x_i$ starting in the root and ending in $(j,\rho_n, j')$.
Let $\sigma_n \defined \lambda^{\rho}_{j{+}1} \circ \cdots \circ \lambda^{\rho}_{j'}$
be the substitution applied along $\rho_n$. Let
$\alpha \in ((\Regs \smallsetminus \{\out\})\cup C)^*$ be such
that $\sigma_n(\out) = \out \alpha$, then we store $\cro{\alpha}^{\rho_1 \cdots \rho_{n-1}}_{j}$
(it is the value which is added in $\out$ when following $\rho_n$).
Furthermore we do the following for the root and the leaves:
\begin{itemize}
\item if $(j,\rho_n, j')$ is the root of $T^x_i$, then $\cro{\alpha}^{\rho_1 \cdots \rho_{n-1}}_{j}$ is stored in
the register $\out$ of $\trans$;
\item if $(j,\rho_n, j')$ is a leaf, we also store for all $\reg \in \Regs \smallsetminus \{\out\}$,
the value $\cro{\reg}^{\rho}_{j'}$
(it is the value of $\reg$ after executing $\rho$).
\end{itemize}
\end{itemize}
It is clear that for $x \in \Dom{h}$, the value of $\out$ in $\trans$
in always a prefix of $h(x)$. Furthermore, this value
tends to an infinite word by Claim~\ref{claim:calinf} and the semantics of
restricted \NSST{}.
The updates of $\trans$ can be performed by following the operations
which build $T^x_i$ in Definition~\ref{def:Txi}. Furthermore, 
the machine can be built in a $1$-bounded way. Indeed,
$\ntrans$ was copyless, hence we can check that given a
branch of $T^x_i$, there is no need to make copies
in order to update the information stored by $\trans$ along this branch
(hence we only create copies if a leaf creates several children,
but they correspond to several distinct futures).

\subparagraph*{Domains.} We have built a 
$1$-bounded \DSST{} computing an extension of $h$.
We now show that its domain can be restricted to $\Dom{h}$. This result follows from
 lemmas~\ref{lem:dom-reg} and~\ref{lem:dom-compo}.

\begin{lemma} \label{lem:dom-compo}
If $h$ is computed by a restricted \NSST{},
then $\Dom{h}$ is Büchi deterministic.
\end{lemma}

\begin{proof} The idea is to build a deterministic Büchi automaton 
which keeps track of $T^x_i$ as $\trans$ does. However, checking that
$r_i \rightarrow \infty$ (see Claim~\ref{claim:calinf}) and that the
$\out$ of $\trans$ tends to an infinite value is not sufficient
for $x \in A^\omega$ to be in $\Dom{f}$. Indeed, there may exist
infinite runs which we removed in the operation ``Removing ambiguity'' of Definition~\ref{def:Txi}. 
Hence if $q'$ is defined as in ``Removing ambiguity'', we also have to check
all runs which start in $q'$ and are labelled by a prefix of $x[i{+}2{:}]$ are finite.
This can be done by simulating these runs, and reaching an accepting state only
when this simulation gets blocked.
\end{proof}

\section{Proof of Lemma~\ref{lem:continuity-loops}}

Let $x \in A^\omega$ and $\beta \in B^\omega$
be such that $q'_2 \runs{x|\beta}$ is final (such a run exists
since the transducer is trim and clean). Therefore, for all $n \ge 0$ we
have $f(u{u'}^nx) = \alpha_2 {\alpha'_2}^n \beta $.
On the other hand $f(u{u'}^\omega) = \alpha_1{\alpha'}_1^\omega$ because
$q_1 \in F$. By continuity in $u{u'}^\omega \in \Dom{f}$, for all $p \ge 0$
we have $|f(u{u'}^nx) \wedge f(u{u'}^\omega)| \ge p$ for $n$ large enough.
The result follows directly.

\section{Proof of Lemma~\ref{lem:make-productive}}

Let $\trans = (A,B,Q, I, F, \Delta, \lambda)$
be a trim, unambiguous and clean \oNT{} which computes
a continuous function $f: A^\omega \parfonc B^\omega$.
Note that Lemma~\ref{lem:continuity-loops} holds.

\begin{definition} We say that $q'_2 \in Q \smallsetminus F$ is \emph{constant}
if there exists $q_1, q_2 \in I$, $q'_1 \in F$,
$u \in A^*$, $u' \in A^+$, $\alpha_1, \alpha'_1, \alpha_2, \alpha'_2 \in B^*$
with $\alpha'_2 = \movi$ such that $q_i \runs{u | \alpha_i} q'_i \runs{u' | \alpha'_i} q'_i$.
\end{definition}
Since $\trans$ is clean, the existence of constant states
is clearly equivalent to the non-productivity of $\trans$.
Thus we want to avoid such states.
Let us now justify the ``constant'' terminology.

\begin{claim} We can compute the set of constant states.
Furthermore given a constant state $q$,
we can compute $\alpha_{q} \in B^*, \alpha'_{q} \in B^+$
such that for all
final run $q \runs{x|\beta}$, 
$\beta =  \alpha_{q} {\alpha'_{q}}^\omega$.
\end{claim}

\begin{proof}
If $q'_2$ is constant,
we can enforce $|u|, |u'| \le \bound$
(see e.g. Lemma~\ref{lem:carac-compat} which states
a more general result).
Hence we can decide if a state is constant.
Furthermore, if $q'_2$ is constant then by Lemma~\ref{lem:continuity-loops},
we have $\beta = (\alpha_2)^{-1}\alpha_1 {\alpha'_1}^\omega$
for all $q \runs{x|\beta}$ final.
\end{proof}
Given $q$ constant, we describe a transducer without $\movi$-loops,
which computes the function $x \mapsto \alpha_q \alpha'_q$ when $x$ labels a
final run starting in $q$. Let $\trans_q = (A,B, Q_q, I_q, F_q, \Delta_q, \lambda_q)$ be:
\begin{itemize}
\item $Q_q \defined \{q \} \uplus  \{q': q \rightarrow^* q' \} \times \{1\}$;
\item $I_q \defined \{q\}$ and $F_q \defined \{(q',1) \in Q_q: q' \in F\}$;
\item $\Delta_q  = \{(q,a, (q',1)): (q,a,q') \in \Delta \} \uplus
 \{((q',1),a, (q'',1)): (q',a,q'') \in \Delta \}$;
 \item $\lambda_q(q,a, (q',1)) = \alpha_q$ and $\lambda_q((q',1),a, (q'',1)) = \alpha'_q$.
\end{itemize}

\begin{claim} \label{claim:loooop} If $\new{q} \runs{u | \alpha} \new{q}$ in $\trans_q$,
then $u \neq \movi$ implies $\alpha \neq \movi$.
Furthermore $\trans_q$ is unambiguous and
it computes $f_q: A^\omega \parfonc B^\omega$
such that $x \in \Dom{f_q}$ if and only if there exists a (unique) final run $q \runs{x | \beta}$
in $\trans$, and then $f_q(x) \defined \beta$.
\end{claim}

\begin{proof} It is clear that $\trans_q$ has no loop
which produces $\movi$ since $\alpha'_q \neq \movi$.
Furthermore, an accepting run of $\trans_q$ labelled by $x$
is necessarily of the form $q \runs{x[1] | \alpha_q } (q_1,1) \runs{x[2] | \alpha'_q} (q_2, 1) \cdots$,
where $q \runs{x[1]} q_1 \runs{x[2]} q_2 \cdots$ is final in $\trans$.
The converse also holds. Hence $\trans_q$ is unambiguous
since $\trans$ is unambiguous and trim, and furthermore
it computes the function $f_q$.
\end{proof}

Finally, let us build $\trans' = (A,B, Q', I',F',\Delta', \lambda')$ which computes
$f$ and is productive and unambiguous (the trimming can be done after).
It consists of the disjoint union of $\trans$ and $\trans_q$ for $q$ constant.
Its initial states are those of $\trans$, and the final states are both those
of  $\trans$ and $\trans_q$ for $q$ constant.
Furthermore, for $q$ constant, all the outing transitions from $q$
in $\trans$ are removed, and $q$ is merged with the initial state of $\trans_q$
(we are forced to go in $\trans_q$).

\begin{claim} Let $x \in \Dom{f}$ and $\rho$
be the accepting run of $\trans$ labelled by $x$.
Assume that it never visits a constant state.
Then $\rho$ is also an accepting run in $\trans'$ (which stays in $\trans$)
with same label and output than in $\trans$. Conversely, if $\rho$
is an accepting run of $\trans'$ which stays in $\trans$, labelled by $x \in A^\omega$,
then $\rho$ is an accepting run of $\trans$ which never visits
a constant state, with same label and output than in $\trans'$.
\end{claim}

\begin{claim}  Let $x \in \Dom{f}$ and $\rho = q_0 \runs{x[1]} q_1 \cdots$
be the accepting run of $\trans$ labelled by $x$.
Assume that it visits a constant state for the first time in $q_i$ for $i \ge 0$.
Then $\rho' \defined q_0 \runs{x[1{:}i]} q_i \runs{x[i{+}1]} (q_{i+1},1) \cdots$
is also an accepting run in $\trans'$
with same label and output than $\rho$. Conversely, if $\rho'$
is an accepting run of $\trans'$ labelled by $x \in A^\omega$, which goes into some
$\trans_q$ at some point, then it stays in this $\trans_q$
and is of the form
$\rho' \defined q_0 \runs{x[1{:}i]} q_i \runs{x[i{+}1]} (q_{i+1},1) \cdots$
where $q_i = q$ is constant in $\trans$.
Then $\rho \defined q_0 \runs{x[1]} q_1 \runs{x[2]} q_2 \cdots $
is an accepting run in $\trans$ which visits a constant state,
and with the same output as $\rho'$.
\end{claim}
The two above claims enable us to show that
$\trans'$ is unambiguous and that it computes
the function $f$. It is still clean, we now check that
it is productive. Assume that
there exists $q_1, q_2 \in I'$, $q'_1 \in F'$, $q'_2 \in Q \smallsetminus F$,
$u \in A^*$, $u' \in A^+$, $\alpha_1, \alpha'_1, \alpha_2, \alpha'_2 \in B^*$
with $\alpha'_2 = \movi$ such that $q_i \runs{u | \alpha_i} q'_i \runs{u' | \alpha'_i} q'_i$.
Then $q'_2$ is not in one of the $\trans_q$, because these
machines have no $\movi$-loops by Claim~\ref{claim:loooop}.
Thus we had a run $q_2 \runs{u | \alpha_2} q'_2 \runs{u' | \alpha'_2} q'_2$
in $\trans$ and so $q'_2$ was not constant in $\trans$
(because otherwise it would be the initial state of $\trans_{q'_2}$).
Now:
\begin{itemize}
\item either $q_1 \runs{u | \alpha_i} q'_1 \runs{u' | \alpha'_i} q'_1$ stays in $\trans$.
Since $q'_1 \in F'$, $q'_1 \in F$ and this statement contradicts
the fact that $q'_2$ is not constant in $\trans$;
\item or $q_1 \runs{u | \alpha_1} q'_1 \runs{u' | \alpha'_1} q'_1$ goes in
$\trans_q$ at some point, for some constant state $q$. Assume without loss
of generality that $q_1 \not \in \trans_q$, then $q'_1 = (q''_1, 1)$ and
 $q_1 \runs{u | \beta_1} q''_1 \runs{u' | \beta'_1} q''_1$  is a run in $\trans$.
Then $q''_1 \in F$ since $q'_1 \in F_q$, hence $\beta'_1 \neq \movi$ because
$\trans$ was clean. This contradicts the fact that $q'_2$ is  not constant in $\trans$.
\end{itemize}

\section{Proofs of lemmas~\ref{lem:mutual} and~\ref{lem:sep-theta}}

The goal of this section is to show
lemmas~\ref{lem:mutual} and~\ref{lem:sep-theta}.
For this purpose, we establish several
properties of the \oNT{}  $\trans = (A,B,Q, I, F, \Delta, \lambda)$
which is assumed to compute a continuous function.
The two results will (respectively) be consequences of
lemmas~\ref{lem:ends} and~\ref{lem:sep-theta2}.

\subsection{Compatible sets and continuity}

We first give a pumping-like characterization
of compatible sets (which implies that one can
decide if a set is compatible).

\begin{lemma}[Characterization of compatibility]
\label{lem:carac-compat} The set $C$ is compatible
if and only if there exists a function $d: C \rightarrow Q$,
and $u,u'\in A^*$ such that the following holds:
\begin{itemize}
\item for all $q \in C$, $q \runs{u} d(q) \runs{u'} d(q)$;
\item there exists $q \in C$ such that $d(q)$ is accepting,
\item  $u' \neq \movi$ and $|u|, |u'| \le \bound$.
\end{itemize}
\end{lemma}

\begin{proof}
It is clear that this condition implies compatibility.
Now assume that there exists $ x \in A^\omega$
and infinite runs $\rho_q$ for $q \in C$ labelled by $x$ such that
$\forall q \in C$, $\rho_q$ begins in $q$, and
furthermore  $\rho_p$ is final for some $p \in Q$.
Therefore we have $\rho_p(i) \in F$ infinitely often,
and by a pigeonhole argument we get $i < j \in \Nat$ such that
$\rho_q(i) = \rho_q(j)$ for all $q \in C$,
and $\rho_p(i) \in F$. We define $d(q) \defined \rho_q(i)$
for $q \in C$. Finally we get
$q \runs{x[1{:}i]} d(q) \runs{x[i{+}1{:}j]} d(q)$ for all $q \in C$.

Now if $|x[i{+}1{:}j]| \ge \bound + 1$, by a similar
pumping argument we factor
$x[1{:}i]  = v v'v''$ such that $v,v' \neq \movi$
and $q \runs{vv''} d(q) \runs{x[i{+}1{:}j]} d(q)$ for all $q \in C$.
By induction we obtain $|u| \le \bound$ and 
a similar reasoning gives $1 \le |u'| \le \bound$.
\end{proof}
We now introduce the notion of \emph{end},
which allows to complete initials runs by some future.

\begin{lemma}[End] \label{lem:ends} Let $C \in \Comp$,
there exists a function $\en{C}: C \rightarrow B^\omega$
such that for all initial step $J,u,C$ and $p,q \in C$ we have
$\val{J,C}{u}(p) \en{C}(p) = \val{J,C}{u}(q) \en{C}(q)$.
\end{lemma}

\begin{proof} Since $C$ is compatible we get by Lemma~\ref{lem:carac-compat}
words $v \in A^*$, $v' \in A^+$ and a function $d: C \rightarrow Q$
such that for all $q \in C$, $q \runs{v | \alpha(q)} d(q) \runs{v'|\alpha'(q)} d(q)$
with $\alpha(q), \alpha'(q) \in B^*$.  Since $\trans$ is trim and clean,
for all $q \in C$ there exists $x(q) \in A^\omega, \beta(q) \in B^\omega$
such that $d(q) \runs{x(q) | \beta(q)} \infty$ is accepting. We define
$\en{C}(q) \defined \alpha(q) {\alpha'(q)}^\omega$ if $\alpha'(q) \neq \movi$
and $\en{C}(q) \defined \alpha(q) \beta(q)$ otherwise.
Let us now justify that $\en{C}$ verifies our equalities.
By Lemma~\ref{lem:carac-compat}, there is some $p \in C$ such that $d(p)$ is final.
By transitivity is it enough to show that  for all $q \in C$ we have
$\val{J,C}{u}(p) \en{C}(p) = \val{J,C}{u}(q) \en{C}(q)$,
which is a direct consequence of Lemma~\ref{lem:continuity-loops}.
\end{proof}
The notions of \emph{common} and \emph{advance}
are presented in~Definition~\ref{def:del}.
They enable us to reformulate Lemma~\ref{lem:ends}
as follows.

\begin{lemma} \label{lem:ends2b}
Let $J,u,C$ be an initial step and $C,v,D$ be a step. Then
for all $p,q \in D$:
\begin{equation*}
 \adv{J,C}{u}(\pre{C,D}{v}(q))  \val{C,D}{v}(q) \en{D}(q)
= \adv{J,C}{u}(\pre{C,D}{v}(q)) \val{C,D}{v}(p) \en{D}(p).
\end{equation*}
\end{lemma}

\begin{proof} By Lemma~\ref{lem:ends} we have
$\val{J,D}{uv}(p) \en{D}(p) = \val{J,D}{uv}(q) \en{D}(q)$
which by splitting gives
$\val{J,C}{u}(\pre{C,D}{v}(p)) \val{C,D}{v}(p) \en{D}(p)
= \val{J,C}{u}(\pre{C,D}{v}(q)) \val{C,D}{v}(q) \en{D}(q)$.
\end{proof}

\subsection{Separable compatible sets}

The notion of \emph{separable} compatible set
is presented in~Definition~\ref{def:separable-set}.
We first give a pumping-like characterization
of these sets (which implies that one can
decide if a set is separable).

\begin{lemma}[Characterization of separability]
\label{lem:carac-separ} A set $C \in \Comp$ is separable
if and only if there exists two functions $i: C \fonc I$ and $\ell: C \fonc Q$,
$u,u',u''\in A^*$ and three functions $\alpha, \alpha', \alpha'': C \fonc B^*$
such that the following holds:
\begin{itemize}
\item for all $q \in C$, $i(q) \runs{u | \alpha(q)} \ell(q) \runs{u' | \alpha'(q)} \ell(q) \runs{u'' | \alpha''(q)} q$;

\item $u' \neq \movi$ and $|u|, |u'|, |u''| \le \bound$;
\item there exists $p,q \in C$ such that $|\alpha'(p)| \neq |\alpha'(q)|$.
\end{itemize}
\end{lemma}

\begin{proof} If the conditions holds, then by iterating the loop
the set is separable. Conversely,
let $J,v,C$ and $p,q \in C$ be such that
$||\val{J,C}{v}(p)| - |\val{J,C}{v}(q)|| > M \bound$.
Suppose by symmetry that $ |\val{J,C}{v}(p)| > |\val{J,C}{v}(q)| + M \bound$.
Thus $|v| > |Q|^{|Q|}$. By pumping
we can  factor $v = uu'u''$ with $0 < |u'| \le |Q|^{|Q|}$
such that $i(q) \runs{u| \alpha(r)} \ell(r) \runs{u' | \alpha'(r)} \ell(r) \runs{u'' | \alpha''(r)} r$
for all $r \in Q$. Now, if  $|\alpha'(p)| = |\alpha'(q)|$, we can remove the loop
and get the result by induction since $|uu''| < |v|$
and $|\alpha(p) \alpha''(p)| > |\alpha(q) \alpha''(q)| + M \bound$.
Otherwise $|\alpha'(p)| \neq |\alpha'(q)|$ and we can enforce
 $|u|, |u''| \le \bound$ by a similar pumping argument.
\end{proof}
\begin{remark} Since $\Bound \defined M \bound$, observe that
$|\alpha(q)|, |\alpha'(q)|, |\alpha''(q)| \le \Bound$.
\end{remark}

We now state the strong version of Lemma~\ref{lem:sep-theta}.
Its proof is given in Subsection~\ref{ssec:proof:sep-theta2}.

\old{
\begin{lemma}[Looping futures - strong version] \label{lem:sep-theta2}
Let  $C \in \Comp$  be separable
and $J,w,C$ be an initial step. There exists 
$\tau, \theta \in B^*$ with $|\theta| = \Bound!$
and $|\tau| \le 3\Bound$, which can be uniquely determined
from $C$  and $\adv{J,C}{u}(q)$ for $q \in C$,
such that:
\begin{itemize}
\item $\tau \pref \advm{J,C}{w} \pref \tau  \theta^\omega$;
\item for all step $C,v,D$ and $q \in D$,
$\val{C,D}{w}(q) \en{D}(q) =  (\adv{J,C}{w}(\pre{C,D}{v}(q)))^{-1} \tau  \theta^\omega$.
\end{itemize}
\end{lemma}}

\cor{
\begin{lemma}[Looping futures - strong version] \label{lem:sep-theta2}
Let  $C \in \Comp$  be separable
and $J,w,C$ be an initial step. There exists 
$\tau, \theta \in B^*$ with $|\theta| = \Bound!$
and $|\tau| \le \Bound!$, which can be uniquely determined
from $C$  and $\adv{J,C}{u}(q)$ for $q \in C$,
such that:
\begin{itemize}
\item $ \advm{J,C}{w} \pref \tau  \theta^\omega$;
\item for all step $C,v,D$ and $q \in D$,
$\val{C,D}{w}(q) \en{D}(q) =  (\adv{J,C}{w}(\pre{C,D}{v}(q)))^{-1} \tau  \theta^\omega$.
\end{itemize}
\end{lemma}
}

\subsection{Proof of Lemma~\ref{lem:sep-theta2}}

\label{ssec:proof:sep-theta2}

\cor{Since $C$ is separable, we get $p,q \in C$ verifying the conditions
of Lemma~\ref{lem:carac-separ}. Assume by symmetry that
$0  \le |\alpha'(q)| <|\alpha'(p)| \le \Bound$.
Note that $i(C), u{u'}^n u'', C$ is an initial step for all $n \ge 0$.
From this observation, we deduce \cref{slem:beta-psi}.

\begin{sublemma} \label{slem:beta-psi} There exists $\beta, \theta \in B^*$ such that
$|\beta| \le \Bound$ and $|\theta| = \Bound!$, and $N,K \ge 0$ such that
for all $n$ large enough, we have:
\begin{equation}
\beta \theta^{n - K} \pref \left(\val{i(C),C}{u{u'}^{nN}u''}(q)\right)^{-1} \val{i(C),C}{u{u'}^{nN}u''}(p).
\end{equation}
Furthermore, the values $\beta$ and $\theta$ can be computed from
$i(C),C,u,u'$ and $u''$.
\end{sublemma}

\begin{proof} Let us first observe that for $n \ge 0$ large enough, we have:
\begin{equation*}
\begin{aligned}
\pi_n \defined \left(\val{i(C),C}{u{u'}^nu''}(q)\right)^{-1} \val{i(C),C}{u{u'}^nu''}(p)
= \alpha(p) \alpha'(p)^n \alpha''(p)[t_n{:}]
\end{aligned}
\end{equation*}
where $t_n \defined |\alpha(q)| +n \times |\alpha'(q)| + |\alpha''(q)|$.
If $|\alpha'(q)| = 0$ the result is clear. From now on, we assume that
$|\alpha'(q)| > 0$. Let us consider $n|\alpha'(p)|$ iterations of the loop, then:
\begin{equation*}
\begin{aligned}
\pi_{n |\alpha'(p)|} &=  \left( \alpha'(p)^{n |\alpha'(p)|} \alpha''(p)\right)[|\alpha(q)|+|\alpha''(q)| - |\alpha(p)|
+ n |\alpha'(p)| |\alpha'(q)|: ]\\
& = \left(\alpha'(p)^{{(n-M)} (|\alpha'(p)|  - |\alpha'(q)|)}
 \alpha''(p)\right)[|\alpha(q)|+|\alpha''(q)| - |\alpha(p)| + M|\alpha'(p)| |\alpha'(q)|:]\\
\end{aligned}
\end{equation*}
where $M$ is fixed such that $|\alpha(q)|+|\alpha''(q)| - |\alpha(p)| + M|\alpha'(p)| |\alpha'(q)| \ge 0$.
The result easily follows by choosing $\theta \defined \phi^{\Bound!/|\alpha'(p)|}$
where $\phi$ is defined as a conjugate of $\alpha'(p)$ (shifted of
$|\alpha(q)|+|\alpha''(q)| - |\alpha(p)| + M|\alpha'(p)| |\alpha'(q)|$), and $\beta = \movi$.
\end{proof}
}

\old{
Since $C$ is separable, we get $p,q \in C$ verifying the conditions
of Lemma~\ref{lem:carac-separ}. Assume by symmetry that
$0  \le |\alpha'(q)| <|\alpha'(p)| \le \Bound$.
Note that $i(C), u{u'}^n u'', C$ is an initial step for all $n \ge 0$.
Furthermore for $n \defined |\alpha'(p)|k$ and $k \ge 0$ large enough, we get:
\begin{equation*}
\label{eq:abo}
\begin{aligned}
\left(\val{i(C),C}{u{u'}^nu''}(q)\right)^{-1} \val{i(C),C}{u{u'}^nu''}(p)
&= \left(\alpha(q) \alpha'(a)^n \alpha''(q)\right)^{-1} \alpha(p) \alpha'(p)^n \alpha''(p)\\
&=  \left(\alpha(p) \alpha'(p)^n \alpha''(p)\right)[|\alpha(q)\alpha''(q)| + k |\alpha'(p)\alpha'(q)|: ]\\
&= \left(\alpha'(p)^{k(|\alpha'(p)| {-} |\alpha'(q)|)}\alpha''(p)\right) [K:] \\
& \substack{~~\tnorm{with}~~ K \defined |\alpha(q)\alpha''(q)| {-} |\alpha(p)|}\\
&= \phi^{k-K} \delta(k)
\end{aligned}
\end{equation*}
where $\phi \defined (\alpha'(p)^\omega) [K: K {+} |\alpha'(p)|]$
is a conjugate of $\alpha'(p)$ which is computable from $C$,
and $\delta(k) \defined
\left(\phi^{k-K}\right)^{-1} \left((\alpha'(p)^{k(|\alpha'p| - |\alpha'q|)} \alpha''(p))
[K{:}] \right)$.\\
Let $\psi \defined \phi^{\Bound! / |\phi|}$, note that $|\psi| = \Bound!$
and it can be computed from $C$.}

\cor{From this result, we now deduce that the futures have a looping behavior.

\begin{sublemma} \label{slem:rrr} For all step $C,v,D$ and for all $\new{r} \in D$,
if $r \defined \pre{C,D}{v}(\new{r})$ then we have
$\val{C,D}{v}(\new{r}) \en{D}(\new{r})
= (\adv{i(C),C}{u{u'}u''}(r))^{-1} (\adv{i(C),C}{u{u'}u''}(q)) \beta \theta^{\omega}$.
\end{sublemma}}

\old{%
\begin{sublemma} \label{slem:rrr} For all step $C,v,D$ and for all $\new{r} \in D$,
if $r \defined \pre{C,D}{v}(\new{r})$ then we have
$\val{C,D}{v}(\new{r}) \en{D}(\new{r})
= (\adv{i(C),C}{u{u'}u''}(r))^{-1} (\adv{i(C),C}{u{u'}u''}(q)) \psi^{\omega}$.
\end{sublemma}%
}

\cor{%
\begin{proof}
Let $\new{p}$, $\new{q}$ be such that
$\pre{C,D}{v}(\new{p}) = p$ and $\pre{C,D}{v}(\new{q}) = q$.
We get from \cref{lem:ends2b} applied to $i(C), uu'^nu'', C$ and $C,v,D$ that:
\begin{equation*}
\begin{aligned}
\val{C,D}{v}(\new{q}) \en{D}(\new{q})
&= \left(\val{i(C),C}{u{u'}^nu''}(q)\right)^{-1} \val{i(C),C}{u{u'}^nu''}(p) \val{C,D}{v}(\new{p}) \en{D}(\new{p}).
\end{aligned}
\end{equation*}
For $n$ large enough, \cref{slem:beta-psi} shows
$\beta \theta^{n-K}$$ \left(\val{i(C),C}{u{u'}^{nN}u''}(q)\right)^{-1} \val{i(C),C}{u{u'}^{nN}u''}(p) $.
Therefore, $\beta \theta^{n-K}\pref \val{C,D}{v}(\new{q}) \en{D}(\new{q})$.
Hence $\val{C,D}{v}(\new{q}) \en{D}(\new{q}) = \beta \theta^\omega$
because we can chose arbitrarily large $n \ge 0$.

Now, let us apply  \cref{lem:ends2b} to $i(C), uu'u'', C$ and $C,v,D$,
we get 
\begin{equation*}
\begin{aligned}
\val{C,D}{v}(\new{r}) \en{D}(\new{r})
= (\adv{i(C),C}{u{u'}u''}(r))^{-1} \adv{i(C),C}{u{u'}u''}(q)
\underbrace{\val{C,D}{v}(\new{q}) \en{D}(\new{q})}_{= \beta \theta^\omega}.
\end{aligned}
\end{equation*}
and the result follows immediately.
\end{proof}}

\old{%
\begin{proof}
Let $\new{p}$, $\new{q}$ be such that
$\pre{C,D}{v}(\new{p}) = p$ and $\pre{C,D}{v}(\new{q}) = q$.
We get from Lemma~\ref{lem:ends2b} applied to $i(C), uu'^nu'', C$ and $C,v,D$ that:
\begin{equation*}
\begin{aligned}
\val{C,D}{v}(\new{q}) \en{D}(\new{q})
&= \left(\val{i(C),C}{u{u'}^nu''}(q)\right)^{-1} \val{i(C),C}{u{u'}^nu''}(p) \val{C,D}{v}(\new{p}) \en{D}(\new{p})\\
& = \phi^{k-K} \delta(k) \val{C,D}{v}(\new{p}) \en{D}(\new{p})
\substack{\tnorm{~~for $k \ge 0$ large enough }}
\end{aligned}
\end{equation*}
By using arbitrarily large $k$, we see that
$\val{C,D}{v}(\new{q}) \en{D}(\new{q}) = \phi^{\omega} = \psi^\omega$.
Now if $\new{r}$ is such that $\pre{C,D}{v}(\new{r}) = r$, then by Lemma~\ref{lem:ends2b}
applied to $i(C), uu'u'',C$ and $C,v,D$ we get: 
\begin{equation*}
\begin{aligned}
\val{C,D}{v}(\new{r}) \en{D}(\new{r})
= (\adv{i(C),C}{u{u'}u''}(r))^{-1} \adv{i(C),C}{u{u'}u''}(q)
\underbrace{\val{C,D}{v}(\new{q}) \en{D}(\new{q})}_{\psi^\omega}.
\end{aligned}
\end{equation*}
and the result follows immediately.
\end{proof}}
\cor{%
Let us now consider what happens with the step $J,w,C$. Let
$r \in C$ (resp. $s \in C$) be such that $\adv{J,C}{w}(r) = \movi$
(resp. $\adv{J,C}{w}(s) = \advm{J,C}{w}$), i.e.
the run ending in $r$ (resp. in $s$) has the smallest
(resp. the longest) production.

Let  $C,v,D$ be a step and $\new{s} \in D$ (resp. $\new{r}$)
be such that $s = \pre{C,D}{v}(\new{s})$ (resp. $r = \pre{C,D}{v}(\new{r})$).
Note that there exists such a step (at least the empty one), and furthermore:
\begin{equation*}
\begin{aligned}
\advm{J,C}{w}  \val{C,D}{v}(\new{s}) \en{D}({\new{s}})
&= \adv{J,C}{w}(s)  \val{C,D}{v}(\new{s}) \en{D}({\new{s}})
& \substack{\tnorm{~by choice of $s$;}}\\
& =  \adv{J,C}{w}(r)  \val{C,D}{v}(\new{r}) \en{D}(\new{r})
& \substack{\tnorm{~by \cref{lem:ends2b};}}\\
& =  \movi  \val{C,D}{v}(\new{r}) \en{D}(\new{r}) 
 &\substack{\tnorm{~by choice of $r$;}}\\
&= (\adv{i(C),C}{u{u'}u''}(r))^{-1} (\adv{i(C),C}{u{u'}u''}(q))
\beta \theta^\omega
&\substack{\tnorm{~by Sublemma~\ref{slem:rrr}.}}\\
\end{aligned}
\end{equation*}
Let $m \defined |(\adv{i(C),C}{u{u'}u''}(q)| - |\adv{i(C),C}{u{u'}u''}(r)| + |\beta|$,
then $-3 \Bound \le m \le 4 \Bound$ (indeed $|\beta| \le \Bound$,
and furthermore $|\adv{i(C),C}{u{u'}u''}(q)| \le 3 \Bound$
and $|\adv{i(C),C}{u{u'}u''}(r)| \le 3 \Bound$ because 
$|uu'u''| \le 3 \bound$).

Now, observe that $|\theta| = \Bound! \ge 4 \Bound$
because $\Bound = M \bound$ and because we have chosen $M \ge 10$. We finally
make a case disjunction depending on the sign of $m$:
\begin{itemize}
\item if $m \ge 0$, we let $\tau \defined (\adv{i(C),C}{u{u'}u''}(r))^{-1} (\adv{i(C),C}{u{u'}u''}(q))
\beta$;
\item if $m < 0$, we let $\tau \defined (\adv{i(C),C}{u{u'}u''}(r))^{-1} (\adv{i(C),C}{u{u'}u''}(q))
\beta \psi$.
\end{itemize}
Note that $|\tau| \le \Bound!$ and that it only depends on $\beta$ (i.e. on
the step $i(C), uu'u'',C$) and on $r$ (thus on the advances of $J,w,C$),
but \emph{not} on the ``future'' step $C,v,D$ that we have selected.
Hence, for all step $C,v,D$ we have:
\begin{equation*}
\begin{aligned}
\advm{J,C}{w}  \val{C,D}{v}(\new{s}) \en{D}({\new{s}})
= \tau \theta^\omega.
\end{aligned}
\end{equation*}
and \cref{lem:sep-theta2} immediately follows.}%

\section{Proof of Lemma~\ref{lem:pre-compat}}

We denote by $\Parts$ the powerset $2^{Q}$.
We fix a total ordering $\pc$ on $\Parts$.

\subsection{Properties of compatible sets}

We begin this proof by giving some basic properties
of compatible sets.

\begin{definition} If $u \in A^*$ and $S \subseteq Q$,
we let $\push{u}{S} \defined \{q': \exists q \in S,
 q \runs{u|\alpha} q' \} \subseteq Q$.
\end{definition}

\begin{lemma}[Compatible sets cover the future] \label{lem:compat-fini} 
Let $x \in \Dom{f}$, $i \ge 0$ and $q^x_i \in S \subseteq Q$.
Then there exists
$C \in \Comp(S)$ and
$j \ge i$ such that $\push{x[i{+}1{:}j]}{C} = \push{x[i{+}1{:}j]}{S}$.
\end{lemma}

\begin{proof} Assume by contradiction that the property
of Lemma~\ref{lem:compat-fini} does not
hold. Let $P$ be the set of subsets $C \subseteq S$ 
such that $q^x_i \in C$, and for all $p \in C$ there exists an infinite
run starting in $p$ and labelled by $x[i{+1}{:}]$.
Then $P \subseteq \Comp(S)$ and $\{q_i^x\} \in P$.

Now consider a set $C \in P$ such that
$|C| = \max_{C' \in P} |C'|$.
Since $C \in \Comp(S)$
then $\forall j \ge i$ we must have
$\push{x[i{+}1{:}j]}{C} \neq \push{x[i{+}1{:}j]}{S}$,
thus $\push{x[i{+}1{:}j]}{(S \smallsetminus C)} \neq \varnothing$
(because $\push{x[i{+}1{:}j]}{S}= (\push{x[i{+}1{:}j]}{C}) \cup (\push{x[i{+}1{:}j]}{(S \smallsetminus C)})$).
Hence the tree of all runs starting from $S \smallsetminus C$
and labelled by $x$ is infinite, thus by Kruskal's lemma it has
a infinite branch, i.e. there exists a state $p \in S \smallsetminus C$
and a infinite run starting in $p$ and labelled by $x$.
Finally $C \uplus \{p\} \in P$, which contradicts
the maximality of $|C|$.
\end{proof}

Using this result, one can define
$\cov^{x}_i(S)$ to be the smallest
set for $\pc$ among the elements
of $\Comp(S)$ whose future completely covers the 
future of $S$ as quickly as possible. Formally we have
Definition~\ref{def:cov} (which makes sense by Lemma~\ref{lem:compat-fini}).

\begin{definition}[Time and cover] \label{def:cov}
Let $x \in \Dom{f}$, $i \ge 0$ and $q^x_i \in S \subseteq Q$,
 we define:
\begin{itemize}
\item $\ctime^x_i(S) \defined
\min \{j \ge i: \exists C \in \Comp(S) \text{ such that } \push{x[i{+}1{:}j]}{C} = \push{x[i{+}1{:}j]}{S}\}$
\item $\cov^{x}_i(S)$ the minimal element for $\pc$ in the (non-empty) set:\\
$\{C \in \Comp(S): \push{x[i{+}1{:}\ctime^x_i(S)]}{C} = \push{x[i{+}1{:}\ctime^x_i(S)]}{S}\}$.
\end{itemize}
\end{definition}
We now define the sequence $(\Good^x_i)_i \in \Comp^\Nat$
which is obtained by applying successively the functions $(\cov^x_{i})_i$
when starting from $I$.

\begin{definition}[Good] \label{def:good}
Let $x \in \Dom{f}$, the sequence $(\Good^x_i)_i \in \Comp^\Nat$
is defined by:
\begin{itemize}
\item $\Good_0 = \cov^x_0(I)$;
\item for $i \ge 0$, $\Good^x_{i+1} = \cov^x_{i+1}(\push{x[i{+}1]}{\Good^x_i})$.
\end{itemize}
\end{definition}

\begin{lemma} \label{lem:good}
For all $i \ge 0$, $\Good^x_i$
is well defined, $\Good^x_i \subseteq \push{x[1{:}i]}{I}$ and
$q^x_i \in \Good^x_i$.
\end{lemma}

\begin{proof} The result is shown by induction on $i \ge 0$.
Assume that it holds for $i\ge 0$ (the base case is very similar)
then $q^{x}_{i+1} \in \push{x[i{+}1]}{\Good^x_{i}}$ since
$q^{x}_i \in \Good^x_{i}$.
Therefore $\Good^x_{i+1}$
is well defined by Definition~\ref{def:cov}.
Furthermore, by induction hypothesis
and definition of $\cov^{x}_{i+1}$ we have
$\Good^x_{i+1} \subseteq \push{x[i{+}1]}{\Good^x_{i}} \subseteq
\push{x[i{+}1]}{\push{x[1{:}i]}{I}} = \push{x[1{:}i{+}1]}{I}$.

It remains to show that $q^x_{i+1} \in \Good^x_{i+1}$.
Indeed, by definition of $\cov^{x}_{i+1}$
there exists $j \ge i{+}1$ such that
$\push{x[i{+}2{:}{j}]}{\Good^{x}_{i+1}} = \push{x[i{+}2{:}{j}]}{\push{x[i{+}1]}{\Good^x_{i}}}$.
But we necessarily have $q^{x}_{j} \in \push{x[i{+}2{:}{j}]}{\push{x[i{+}1]}{\Good^x_{i}}}$ because
$q^{x}_{i+1} \in \push{x[i{+}1]}{\Good^x_{i}}$. Therefore we
can find a run of $\trans$ of the form
$p_{i{+}1} \rightarrow p_{i{+}2} \rightarrow \cdots \rightarrow p_j$
labelled by $x[i{+}2{:}j]$ where $p_{i+1} \in \Good^{x}_{i+1}$
and $p_j = q^{x}_j$. Since $\Good^{x}_{i+1} \subseteq  \push{x[1{:}i{+}1]}{I}$
and $\trans$ is unambiguous, one has $p_{i+1} = q^x_{i+1} $.
\end{proof}

\begin{remark}
Since $\Good^x_i \subseteq \push{x[1{:}i]}{I}$  and $\Good^x_{i+1} \subseteq \push{x[i{+}1]}{\Good^x_i}$
and because the transducer $\trans$ is trim and unambiguous,
we have that $\Good^x_i, x[i], \Good^x_{i+1}$ is a pre-step.
\end{remark}
In order to show Lemma~\ref{lem:pre-compat},
we prove that the sequence
$(\Good^x_i)_i$ can be computed by a restricted \oNT{}.

\begin{proposition} \label{prop:ont-good}
One can build a restricted \oNT{} $\mc{A}$ computing some
$f':A^\omega \parfonc (A \uplus \Comp)^\omega$ such that
if $x \in \Dom{f}$, then $f'(x) = \Good^x_0 x[1] \Good^x_1 x[2] \cdots$.
\end{proposition}

\begin{proof}[Proof of Lemma~\ref{lem:pre-compat}.]
Immediate by definition of $(\Good^x_i)_i$ and Lemma~\ref{lem:good}.
\end{proof}

The rest of this section is devoted to showing Proposition~\ref{prop:ont-good}.

\subsection{Description of the restricted \oNT}

We now describe how the restricted \oNT $\mc{A}$ of Proposition~\ref{prop:ont-good}
is built.

\subparagraph*{States.}
A state of $\mc{A}$ is a tuple  $(S,C, \Hi)$ where:
\begin{itemize}
\item $S \in \Parts$ and $C \in \Comp(S)$
(hint: we want $C$ to be $\Good^x_i$ after reading $x[1{:}i]$);
\item  $\Hi$ is a
set of tuples $(S^h,C^h,T,g)$ where:
\begin{itemize}
\item $S^h \in \Parts$, $C^h \in \Comp(S^h)$ and $T \in \Parts$;
\item $g: \Comp(S^h)\rightarrow 2^T$ (hint: it will store some
``history'' about the $C$ and $S$ visited).
\end{itemize}
\end{itemize}

\begin{definition}[Indicator function of compatible subsets] Let $S \in \Parts$, we define the function
$\Chi_S: \Comp(S) \rightarrow 2^{S}, X \mapsto \{X\}$.
\end{definition}
The initial states of $\mc{A}$ are those of the form
$(I,C,\{(I,C,I,\Chi_I)\})$ for $C \in \Comp(I)$. Intuitively,
these $C \in  \Comp(I)$ describe the
possible candidates for $\Good^x_0 = \cov^x_0(I)$.

\subparagraph*{Transitions.} 
Let us now describe formally the transitions.
If $a \in A$, there is an $a$-labelled transition
from $(S_1,C_1,\Hi_1)$ to $(S_2,C_2,\Hi_2)$ if
the following conditions hold:
\begin{enumerate}
\item \label{cond-trans1} $S_2 = \push{a}{C_1}$;
\item \label{cond-trans2} $\Hi_2 = \{(S^h, C^h,\push{a}{T},\push{a}{g}): (S^h, C^h,T,g) \in \Hi_1\}
\cup \{(S_2, C_2, S_2, \Chi_{S_2}) \}$
where:
\begin{equation*}
\push{a}{g}: \Comp(S^h) \rightarrow 2^{\push{a}{T}},
X \mapsto \push{a}{g(X)};
\end{equation*}
\item \label{cond-trans3} furthermore we add two restrictions which aim at ``forcing'' the choice of $\Good^x_i$:
\begin{enumerate}
\item \label{cond:a} if $S_1 \in \Comp$ then $C_1 = S_1$;
\item \label{cond:b} if $ (S^h, C^h,T,g) \in \Hi_1$ with $g^{-1}(T)= \vide$ but $(\push{a}{g})^{-1}(\push{a}{T}) \neq \vide$,
then $C^h$ is the minimal element for $\pc$ in $(\push{a}{g})^{-1}(\push{a}{T}) \neq \vide$.
\end{enumerate}
\end{enumerate}

\subparagraph*{Output.} The word produced on a $a$-transition
coming in $(S,C,\Hi)$ is $aC \in (A \uplus \Comp)^*$.
We add a specific production for the initial states
(which does not modify the construction).

\subsection{Correctness of the construction}

Since the output only describes the $C$ visited,
it remains to show the following result.

\begin{lemma} Let $x \in \Dom{f}$, then $\mc{A}$ has only one
final run labelled by $x$. It is of the form
$(S_0, \Good^x_0, \Hi_0) \rightarrow (S_1, \Good^x_1, \Hi_1) \rightarrow \cdots$
for some $S_0, S_1, \dots$ and $\Hi_0, \Hi_1, \dots$.
\end{lemma}

The rest of this subsection is devoted to the proof
of this result. We first give an extension of the
$\push{a}{g}$ defined to act on a function as in condition~\ref{cond-trans2}
of the transitions of $\mc{A}$.

\begin{definition}[Monoid action over functions] Let
$g: \Comp(S^h) \rightarrow 2^T$ and $u \in A^*$,
then we define
$ \push{u}{g}:  \Comp(S^h) \rightarrow 2^{\push{u}{T}},
X \mapsto \push{u}{X}.$
\end{definition}

\begin{claim} \label{claim:forward-fonc}
If $u \in A^*$ and $a \in A$ then $\push{a}{\push{u}{g}} = \push{ua}{g}$.
Furthermore $\push{\movi}{g} = g$.
\end{claim}
We now show that the $\Hi_i$ along an accepting
run store some ``history'' about the $S_i$ and $C_i$.

\begin{sublemma} \label{slem:shape-runs-A}
Let $(S_0,C_0, \Hi_0) \rightarrow (S_1,C_1, \Hi_1) \rightarrow \cdots$
be a accepting run  of $\mc{A}$ labelled by $x \in A^\omega$, then for all $n \ge 0$:
\begin{itemize}
\item $S_0 = I$ and $S_{n+1} = \push{x[n{+}1]}{C_{n}}$;
\item $\forall n\ge 0$, $C_n \in \Comp(S_n)$;
\item $\forall n \ge 0$ $\Hi_n = \{(S_i, C_i, \push{x[i{+}1{:}n]}{S_i}, \push{x[i{+}1{:}n]}{\Chi_{S_i}}: 0 \le i \le n \}$.
\end{itemize}
\end{sublemma}

\begin{proof}
The result is shown by induction on $n \ge 0$. 
The base case follows from the definition of initial states.
Assume that it holds for some $n \ge 0$, then
$S_{n+1} = \push{a}{S_n}$ by condition~\ref{cond-trans1} on transitions,
$C_{n+1} \in \Comp(S_{n+1})$ by definition of states,
and finally:
\begin{equation*}
\begin{aligned}
\Hi_{n+1} = &\{(S_i, C_i, \push{x[i{+}1{:}n{+1}]}{S_i}, \push{x[i{+}1{:}n{+}1]}{\Chi_{S_i}}): 0 \le i \le n \}\\
 &\cup  \{(S_{n+1}, C_{n+1}, \push{\movi}{S_{n+1}}, \push{\movi}{\Chi_{S_{n+1}}})\}.
 \end{aligned}
 \end{equation*}
by condition~\ref{cond-trans2} on transitions and Claim~\ref{claim:forward-fonc}.
\end{proof}
Given $x \in \Dom{f}$, we now completely describe  the
unique accepting run of $\mc{A}$ labelled by $x$.
Let $(S^x_n)_n \in \Parts^\Nat$ defined by
$S^x_0 \defined I$ and for $n \ge 0$ let 
$S^x_{n+1} \defined \push{x[n{+}1]}{\Good^x_n}$.
Furthermore for $n \ge 0$ we define
$\Hi^x_n \defined \{(S^x_i, \Good^x_i, \push{x[i{+}1{:}n]}{S^x_i},
\push{x[i{+}1{:}n]}{\Chi_{S^x_i}}): 0 \le i \le n \}$.

\begin{sublemma} \label{slem:accepting-run}
Given $x \in \Dom{f}$, then 
$\rho^x_x = (S^x_0, \Good^x_0, \Hi^x_0) \rightarrow (S^x_0, \Good^x_1, \Hi^x_1) \rightarrow \cdots$ 
is an accepting run of $\mc{A}$ labelled by $x$.
\end{sublemma}

\begin{proof} Each $(S^x_n, \Good^x_n, \Hi^x_n)$ is a state
since $\Good^x_n \in \Comp(S^x_n)$.
The first state is initial since
$S^x_0 = I$ and $\Good^x_0 \in \Comp(I)$,
$\Hi^x_0 = \{(I, \Good^x_0, I, \Chi_I) \}$.
Condition~\ref{cond-trans1} holds since $S^x_{n+1} = \push{x[n{+}1]}{\Good^x_n}$ for $n \ge 0$.
Condition~\ref{cond-trans2} holds, since by definition and Claim~\ref{claim:forward-fonc}:
\begin{equation*}
\begin{aligned}
\Hi^x_{n+1} = &\{(S^x_i, \Good^x_i, \push{x[n{+}1]}{\push{x[i{+}1{:}n]}{S^x_i}},
\push{x[n{+}1]}{\push{x[i{+}1{:}n]}{\Chi_{S^x_i}}}): 0 \le i \le n \}\\
 &\cup  \{(S^x_{n+1}, \Good^x_{n+1}, \push{\movi}{S^x_{n+1}}, \push{\movi}{\Chi_{S^x_{n+1}}})\}.
 \end{aligned}
 \end{equation*}
 
Now let us show that given $n \ge 0$, conditions~\ref{cond:a} and \ref{cond:b}
for transitions hold between
 $(S^x_n, \Good^x_n, \Hi^x_n)$ and $ (S^x_{n+1}, \Good^x_{n+1}, \Hi^x_{n+1})$. Indeed:
 \begin{itemize}
 \item if $S^x_n \in \Comp$, then $S^x_n$
 is the unique $C \in \Comp(S^x_n)$
 such that $\push{\movi}{C} = \push{\movi}{S^x_n}$.
 Hence $\ctime_x^n = n$ and $\Good^x_n = \cov^{x}_{n}(S^x_n) = S^x_n$.
 Thus condition \ref{cond:a} holds;
 \item assume that $\exists n \ge  i$ such that
 $(\push{x[i{+}1{:}n]}{\Chi_{S^x_i}})^{-1}(\push{x[i{+}1{:}n]}{S^x_i}) = \vide$
 and such that:
 \begin{equation*}
 \begin{aligned}
\underbrace{(\push{x[n{+}1]}{\push{x[i{+}1{:}n]}{\Chi_{S^x_i}}})^{-1}
(\push{x[n{+}1]}{\push{x[i{+}1{:}n]}{S^x_i}})}_
{ =(\push{x[i{+}1{:}n{+}1]}{\Chi_{S^x_i}})^{-1}
(\push{x[i{+}1{:}n{+}1]}{S^x_i})} \neq \vide
\end{aligned}
\end{equation*}
By definition of $\push{u}{g}$ and $\Chi$ we get for $j \ge i$:
 \begin{equation}
 \label{eq:chi-x}
(\push{x[i{+}1{:}j]}{\Chi_{S^x_i}})^{-1}(Y)
= \{Y \subseteq \Comp(S^x_i): \push{x[i{+}1{:}j]}{X}= Y\}
\end{equation}
therefore for $Y = \push{x[i{+}1{:}j]}{S^x_i}$ we have:
 \begin{equation*}
(\push{x[i{+}1{:}j]}{\Chi_{S^x_i}})^{-1}(\push{x[i{+}1{:}j]}{S^x_i})
= \{C \in \Comp(S^x_i): \push{x[i{+}1{:}j]}{C}  =\push{x[i{+}1{:}j]}{S^x_i} \}
\end{equation*}
Hence our hypotheses yield $\ctime^x_i(S^x_i) = n{+}1$
(see Definition~\ref{def:cov}) and furthermore $\Good^x_{n+1}$
is the minimal element for $\pc$ in the set
$(\push{x[i{+}1{:}n{+}1]}{\Chi_{S^x_i}})^{-1}(\push{x[i{+}1{:}n{+}1]}{S^x_i})$.
 \end{itemize}
Therefore the run is accepting.
\end{proof}

\begin{sublemma} Given $x \in \Dom{f}$,the run $\rho^x_x$
is the unique accepting run labelled by $x$.
\end{sublemma}

\begin{proof} Let us consider an accepting run
$\rho \defined (S_0, C_0, \Hi_0) \rightarrow (S_1, C_1, \Hi_1) \rightarrow \cdots$
and suppose that it coincides with $\rho^x_x$ until the state
$(S_{i-1}, C_{i-1}, \Hi_{i-1}) $. Then one has $S_i = S^x_i$ since:
\begin{itemize}
\item either $i=0$ and $(S_0, C_0, \Hi_0)$ is initial, thus $S_0 = I = S^x_0$;
\item or $i \ge 1$ and by Sublemma~\ref{slem:shape-runs-A}
we have $S_i = \push{x[i]}{C_{i-1}} = \push{x[i]}{\Good^x_{i-1}} = S^x_i$.
\end{itemize}
Furthermore $C_i \in \Comp(S^x_i)$. Now let $n \defined \ctime^x_i(S^x_i)$,
two cases may occur:
\begin{itemize}
\item either $n=i$ which means that $S^x_i \in \Comp(S^x_i)$
hence $\Good^x_i = \cov^x_i(S^x_i) = S^x_i$. By condition~\ref{cond:a} on transitions
(and since there is an  outing transition  from the state
$(S_i, C_i, \Hi_i) $) we conclude that $C_i = S^x_i = \Good^x_i$;
\item or $n > i$. Hence by Equation~\ref{eq:chi-x},
$(\push{x[i{+}1{:}n{-}1]}{\Chi_{S^x_i}})^{-1}(\push{x[i{+}1{:}n{-}1]}{S_i^x}) = \vide$
and $\Good^x_i$ is the smallest element for $\pc$
of $(\push{x[i{+}1{:}n]}{\Chi_{S^x_i}})^{-1}(\push{x[i{+}1{:}n]}{S_i^x})  \neq \vide$.
But by  Sublemma~\ref{slem:shape-runs-A} we have
$(S^x_i, C_i, \push{x[i{+}1{:}n{-}1]}{S^x_i}),
\push{x[i{+}1{:}n{-}1]}{\Chi_{S^x_i}}) \in \Hi_{n-1}$.
Considering the transition from $n{-}1$
to $n$, we must have by condition~\ref{cond:b} that
$C_i$ is the smallest element in
 $(\push{x[i{+}1{:}n]}{\Chi_{S^x_i}})^{-1}(\push{x[i{+}1{:}n]}{S_i^x}) $, that is $\Good^x_i$.
\end{itemize}
We conclude that $C_i = \Good^x_i$ and finally $\Hi_i = \Hi^x_i$
by Sublemma~\ref{slem:shape-runs-A}.
\end{proof}

\section{Invariant preservation in Section~\ref{sec:invariants}: correctness of $\strans$}

\subsection{Proof of Lemma~\ref{lem:premier}}

\old{
Invariants~\ref{inv:C} and~\ref{inv:step} are obvious.
 Invariant~\ref{inv:lag} follows since
 the $\adv{\Jf, \Cf}{x[1{:}i{+}1]}(q)$ were mutual prefixes
 and one of them is empty. Invariants~\ref{inv:out}
 and~\ref{inv:lagging} are obvious due to
 emptiness of the $\nb{\pi}$ and $\outi{\pi}$.
 For Invariant~\ref{inv:last} we use Remark~\ref{rem:sep-theta}.
 Invariant~\ref{inv:past} holds by definition of
 $\lag(q)$ and $\last(q)$ and emptiness of the $\nb{\pi}$ and $\outi{\pi}$.
For invariant~\ref{inv:future}
we use Lemma~\ref{lem:sep-theta} which gives
for all step $\Cf,u,D$ and $q \in D$,
$\val{\Cf,D}{u}(q) \pref  (\adv{\Jf,\Cf}{x[1{:}i{+}1]}(\pre{\Cf,D}{u}(q)))^{-1} \tau  \theta^\omega$.
Therefore by adding $\val{\Jf,\Cf}{x[1{:}i{+}1]}(\pre{\Cf,D}{u}(q))$
on both sides we get $\val{\Jf,D}{x[1{:}i{+}1]u}(q)
\pref  \com{\Jf,\Cf}{x[1{:}i{+}1]} \tau  \theta^\omega$.
We conclude because $\lagm = \tau$.
Finally invariant~\ref{inv:close} follows since all paths
$\pi \neq \Cf$ are close.}

\cor{
Invariants~\ref{inv:C} and~\ref{inv:step} are obvious.
 Invariant~\ref{inv:lag} follows since
 the $\adv{\Jf, \Cf}{x[1{:}i{+}1]}(q)$ were mutual prefixes
(hence by \cref{rem:sep-theta} they are prefixes of $\tau \theta^\omega$,
hence the $\lag(q)$ are prefixes of $\lagm$)
 and one of them is empty. Invariants~\ref{inv:out}
 and~\ref{inv:lagging} are obvious due to
 emptiness of the $\nb{\pi}$ and $\outi{\pi}$.
 For Invariant~\ref{inv:last} we use Remark~\ref{rem:sep-theta}.
 Invariant~\ref{inv:past} holds by definition of
 $\lag(q)$ and $\last(q)$ and emptiness of the $\nb{\pi}$ and $\outi{\pi}$.
For invariant~\ref{inv:future}
we use Lemma~\ref{lem:sep-theta} which gives
for all step $\Cf,u,D$ and $q \in D$,
$\val{\Cf,D}{u}(q) \pref  (\adv{\Jf,\Cf}{x[1{:}i{+}1]}(\pre{\Cf,D}{u}(q)))^{-1} \tau  \theta^\omega$.
Therefore by adding $\val{\Jf,\Cf}{x[1{:}i{+}1]}(\pre{\Cf,D}{u}(q))$
on both sides we get $\val{\Jf,D}{x[1{:}i{+}1]u}(q)
\pref  \com{\Jf,\Cf}{x[1{:}i{+}1]} \tau  \theta^\omega$.
We conclude because $\lagm = \tau$.
Finally invariant~\ref{inv:close} follows since all paths
$\pi \neq \Cf$ are close.
}

\subsection{Proof of Lemma~\ref{lem:sim:algo}}

Invariant~\ref{inv:step} is preserved along the operation,
since we never modify $\Jf$ nor $\Cf$.
It is clear that invariant~\ref{inv:sep}
holds before using the function $\tnorm{\textbf{down}}(\Cf)$.
Indeed, we do not modify the $\lag(q)$, thus the lagging
states are the same, and furthermore we do not
create non-close paths.

It is clear that Algorithm~\ref{algo:down} is well defined
since its recursive calls follow the definition of $\tree{C}$
(in a strictly decreasing way).
For all $\pi \in \tree{\Cf}$, let $n_{\pi}$ (resp. $o_{\pi}$)
be the value of $\nb{\pi}$ (resp. $\outi{\pi}$) before
launching the function $\tnorm{\textbf{down}}(\Cf)$.

\begin{sublemma}
Let $\pi = C_1 \cdots C_n \in \tree{\Cf}$. During the execution of $\tnorm{\textbf{down}}(\Cf)$,
the following invariants hold just before $\tnorm{\textbf{down}}(\pi)$ makes its recursive calls:
\begin{enumerate}[label ={\normalfont \textbf{\textcolor{darkgray}{\roman*.}}}]
\item \label{inv:algo:four} for all $ \pi'  = C_1 \cdots C_{n'} \pref \pi$, $\nb{\pi'}: C_{n'} \fonc [0{:}\four]$;
\item \label{inv:algo:untouched} for all $C' \in \Comp(C_n)$, $C' \neq C_n$
and all $\pi C' \prefneq \pi'$, we have $\nb{\pi'} = n_{\pi'}$
and $\outi{\pi'} = o_{\pi'}$;
\item \label{inv:algo:down} if $\tnorm{\textbf{down}}(\pi)$
has performed the update $\nb{\pi C'}(q) \becomes \nb{\pi C'}(q) + (\nb{\pi}(q){-}4)$,
for some $C'$ and $q \in C_n$, then for all $\pi' \pref \pi$,
we have ${\nb{\pi'}} (q) = \four$.
\item \label{inv:algo:lagging} if $q \in \Cf$ was lagging before launching
$\tnorm{\textbf{down}}(\Cf)$, then for all
$\pi' = C'_1 \cdots C'_{n'} \in \tree{\Cf}$ such
that $q \in C'_{n'}$, we have $\nb{\pi'}(q) = 0$ and, if $\pi' \neq \Cf$, $\outi{\pi'} = \movi$;
\item \label{inv:algo:past}
for all $\pi \pref \pi' = C_1 \cdots C_{n'} \in \tree{\Cf}$ such that $C_{n'} = \{q\}$,
 if $\pi_i$ denotes $C_1 \cdots C_i$ then:
\begin{equation*}
\theta^{\nb{\pi_1}} \left(\prod_{i=2}^{n'} \outi{\pi_i} \theta^{\nb{\pi_i}(q)} \right)  = 
\theta^{n_{\pi_1}} \left(\prod_{i=2}^{n'} o_{\pi_i} \theta^{n_{\pi_i}(q)} \right);
\end{equation*}
\item \label{inv:algo:close} for all $ \pi' = C_1 \cdots C_{n'} \pref \pi$
which is not close and was close before
launching $\tnorm{\textbf{down}}(\Cf)$, let
$J' \defined \pre{\Jf,\Cf}{x[1{:}i]}(C_{n'}) \subseteq \Jf$.
Then $J', x[1{:}i], C_{n'}$ is an initial step and
 $|\advm{J',C_{n'}}{x[1{:}i]} | \ge 4 \Bound !$.
\end{enumerate}
\end{sublemma}

\begin{proof} Invariant~\ref{inv:algo:four} is clear since
step $2$ forces $\nb{\pi} = \four$ if  $\nb{\pi} > \four$.
Invariants~\ref{inv:algo:untouched},
\ref{inv:algo:lagging}  and~\ref{inv:algo:past}
result from a simple but bureaucratic verification. For invariant~\ref{inv:algo:down},
note that if $\nb{\pi C''}(q) \becomes \nb{\pi C''}(q) + (\nb{\pi}(q){-}\four)$ is 
performed, then we had $\nb{\pi}(q) > \four$ (so now $\nb{\pi}(q) = \four$)
before executing $\tnorm{\textbf{down}}(\Cf)$.
Therefore using invariants~\ref{inv:algo:untouched}
and~\ref{inv:algo:down}, we conclude that we had
$\nb{\pi'}(q) = 4$ for all $\pi' \pref \pi$ before
lauching  $\tnorm{\textbf{down}}(\pi)$.

Let us show that invariant~\ref{inv:algo:close}
is preserved. First, if $\pi' \prefneq \pi$ is not close
before $\tnorm{\textbf{down}} (\pi)$ makes its recursive calls,
then it was not close before launching $\tnorm{\textbf{down}} (\pi)$.
Assume now that $\pi$ is not close,
but was close before launching $\tnorm{\textbf{down}} (\Cf)$.
By invariant~\ref{inv:algo:untouched} we conclude
that  $\pi$ was close before launching
$\tnorm{\textbf{down}} (\pi)$.
Hence it means that $\tnorm{\textbf{down}} (\pi)$
performed a $\nb{\pi C'}(q) \becomes \nb{\pi C'}(q) + (\nb{\pi}(q){-}\four)$.
Hence for all $\pi' \pref \pi$ we have ${\nb{\pi'}} (q) = \four$
by invariant~\ref{inv:algo:down}.
Before executing the first
``for'' loop of $\tnorm{\textbf{down}} (\pi)$, we had
${\nb{\pi}} (q') = 0$ for some $q' \in C'$,
and thus  ${\nb{\pi \{q'\}}} (q') = {\nb{\pi}} (q') = 0$ after this loop
(by construction and since $\pi$ was close).
Let $\pi_i \defined C_1 \cdots C_i$ for $1 \le i \le n$.
From invariant~\ref{inv:past}
of $\strans$, and  invariants~\ref{inv:algo:lagging} and~\ref{inv:algo:past},
we get:
\begin{equation*}
\val{\Jf,\Cf}{x[1{:}i]}(q) = \out~\lagm ~\theta^{\nb{\pi_1}(q)}
\left(\prod_{i=2}^n \outi{\pi_i} \theta^{\nb{\pi_i}(q)} \right) \theta^{\nb{\pi\{q\}}(q)} \last(q)
\end{equation*}
and similarly for $q'$ (which may be lagging):
\begin{equation*}
\val{\Jf,\Cf}{x[1{:}i]}(q') = \out~\lag(q) ~\theta^{\nb{\pi_1}(q')}
\left(\prod_{i=2}^n \outi{\pi_i} \theta^{\nb{\pi_i}(q')} \right) \theta^{\nb{\pi\{q'\}}(q')} \last(q')
\end{equation*}
Therefore we get:
\begin{equation*}
\begin{aligned}
|\advm{\Jf,\Cf}{x[1{:}i]}|
& \ge   \left| |\val{\Jf,\Cf}{x[1{:}i]}(q)| - |\val{\Jf,\Cf}{x[1{:}i]}(q')| \right|\\
&= \Big| \underbrace{|\lagm| - |\lag(q')|}_{\ge 0} + \underbrace{|\last(q)| - |\last(q')| + \Bound! |\nb{\pi\{q\}}(q)|}_{\ge 0}\\
&+ \Bound ! \times \sum_{i=1}^n (\four - \nb{\pi_i}(q'))\Big| \ge \Bound! (\four - \nb{\pi}(q')) \ge \four  \Bound! \\
\end{aligned}
\end{equation*}
which concludes the proof.
\end{proof}

Let us conclude that the invariants of $\strans$ hold 
after applying this algorithm. Invariants~\ref{inv:lag}
and~\ref{inv:last} hold because we did not modify the $\lag(q)$
and $\last(q)$. Invariant~\ref{inv:out} is clearly preserved along
Algorithm~\ref{algo:down}. Invariant~\ref{inv:lagging}
follows from invariant~\ref{inv:algo:lagging}.
Invariant~\ref{inv:past}
follows from invariant~\ref{inv:algo:past}
and invariant~\ref{inv:lagging}. Invariant~\ref{inv:future}
is obvious since $\Jf,x[1{:}i],\Cf$ is unchanged.
Finally, for invariant~\ref{inv:close},
we use invariant~\ref{inv:algo:close}
to deal with the paths $\pi \in \tree{\Cf}$
which became close when applying Algorithm~\ref{algo:down}.

Furthermore, $\nb{\pi}: C_n \fonc [0{:}\four]$
for all $\pi = C_1 \cdots C_n \in \tree{\Cf}$
which conclude the proof.

\subsection{Proof of Lemma~\ref{lem:sim:update}}

Since $\Cf$ was separable, then so is $C^x_{i+1}$
(indeed, if the lengths of the runs which end in $\Cf$ can differ
significantly, so are those obtained by adding an $a$-transition
to $C^x_{i+1}$). Furthermore, it is clear that invariants~\ref{inv:C}
and~\ref{inv:step} hold after this step.
We now use $\Jf$, $\lag$, etc. to denote the values
stored before the operation, and $\new{\Jf}$, $\new{\lag}$, etc. (with a bar)
after the operation.

\subparagraph*{Invariants~\ref{inv:lag}, \ref{inv:out} and \ref{inv:last}.}
To show invariant~\ref{inv:lag}, we note that
 $k_q \neq 0$ if and only if $\pre{\Cf,  C^{x}_{i+1}}{a}(q)$ was lagging.
\old{Thus we use invariants~\ref{inv:lag},~\ref{inv:past} and Lemma~\ref{lem:mutual} to
show that the $\lag$ are still mutual prefixes,
and furthermore $\new{\lagm} = c^{-1} \lagm$.}
\cor{Thus we use invariants~\ref{inv:lag},~\ref{inv:past}
and~\ref{inv:future} to
show that the $\new{\lag}$ are prefixes of $\new{\lagm}$.}
For invariant~\ref{inv:out}, it is sufficient to see that
the $\new{\outi{\pi}}$ are obtained from the  $\outi{\pi}$.
We just have to note that $\outi{\Cf}$ is not
used anywhere except for $\new{\outi{\new{\Cf}}}$.
Indeed (with the notations of Subsubsection~\ref{sssec:up-step}),
if $\rho = \new{\Cf}$ and $i_m = n$,
then $i_m = i_1 = 1 = n$ and by definition of $\tree{\Cf}$
we have $\pi = \Cf$.
For invariant~\ref{inv:last}, we use invariants~\ref{inv:past} 
and~\ref{inv:future} which imply that for all $q \in \new{\Cf}$:
\begin{itemize}
\item if $\pre{\Cf,  C^{x}_{i+1}}{a}(q)$ was lagging then
$\val{\Cf, \new{\Cf}}{a}(q) \pref (\lag(\pre{\Cf,  C^{x}_{i+1}}{a}(q)))^{-1} \lagm \theta^\omega$;
\item otherwise, $\val{\Cf, \new{\Cf}}{a}(q) \pref (\last(\pre{\Cf,  C^{x}_{i+1}}{a}(q)))^{-1} \theta^\omega$.
\end{itemize}
Thus $\val{\Cf,  C^{x}_{i+1}}{a}(q) [k_q{+}1{:}] \pref (\last(\pre{\Cf,  C^{x}_{i+1}}{a}(q)))^{-1} \theta^\omega$
by invariant~\ref{inv:lagging}, and so $\new{\last(q)} \pref \theta^\omega$.

\subparagraph*{Invariant~\ref{inv:lagging}.}
For invariant~\ref{inv:lagging}, let us consider
a lagging state $q \in \new{\Cf}$. Then
$\pre{\Cf,  C^{x}_{i+1}}{a}(q)$ was also lagging
and $|\val{\Cf,  C^{x}_{i+1}}{a}(q)| < k_q$.
Thus $\new{\last(q)} = \movi$. Now, let $\pi  = D_1 \cdots D_n
\in \tree{C^x_{i+1}}$ be such that $q \in D_n$,
with the notations of Subsubsection~\ref{sssec:up-step} we have:
\begin{itemize}
\item either $i_m < n$, then $\new{\nb{\pi}} = 0$ and $\new{\outi{\pi}} = \movi$;
\item or $i_m = n$. Since $C_{i_m} = C_n =\pre{\Cf,  C^{x}_{i+1}}{a}(D_n)$,
we get $\pre{\Cf,  C^{x}_{i+1}}{a}(q) \in C_{i_m}$
and so $\new{\nb{\pi}}  = 0$.
Furthermore, if $D_n \neq \new{\Cf}$,
then (see before) $C_{i_m} = C_n \neq \Cf$ thus $\new{\outi{\pi}}  = \movi$.
\end{itemize}

\subparagraph*{Invariant~\ref{inv:past}.} The proof is similar
to that of invariant~\ref{inv:last}. The important result is
that given $\pi = D_1 \cdots D_n \in \tree{\Cf}$ such that 
$D_n = \{q\}$, then (with the notations of Subsubsection~\ref{sssec:up-step})
if $\pi_i \defined D_1 \cdots D_i$ for $1 \le i \le n$
and $\rho_j \defined C_{i_1} \cdots C_{i_j}$ for $1 \le j \le m$, we get:
\begin{itemize}
\item $\new{\outi{\pi_i}} = \movi$ and $\new{\nb{\pi_i}} = 0$ if
$i \not \in \{i_1, \dots, i_m\}$;
\item $\new{\outi{\pi_{i_j}}} = \outi{\rho_j}$ for $1 < j \le m$
and $\new{\nb{\pi_{i_j}}} = {\nb{\rho_j}}  \circ \pre{\Cf, \new{\Cf}}{a} $
if $1 \le j \le m$
\end{itemize}
Hence $\theta^{\new{\nb{\pi_1}( q)}} \prod_{i=2}^n \new{\outi{\pi_i} }
\theta^{\new{\nb{\pi_i}(q)}}
= \theta^{\nb{\rho_1}(\pre{\Cf, \new{\Cf}}{a}(q))} \prod_{j=2}^m
\outi{\rho_j} \theta^{\nb{\rho_j}(\pre{\Cf, \new{\Cf}}{a}(q))}$.

\subparagraph*{Invariant~\ref{inv:future}.} Let $\new{\Cf},u, D$ be a step,
then $\Cf, au, D$ is also a step. Therefore we apply
invariant~\ref{inv:future} at the previous stage and
the result follows since $\out~\lagm~\theta^\omega = \new{\out}
\new{\lagm}~\theta^\omega$.

\subparagraph*{Invariant~\ref{inv:close}.}
Let $\pi = D_1 \cdots D_n \in \tree{\new{\Cf}}$ be not close.
With the notations of Subsubsection~\ref{sssec:up-step},
we get $\rho = C_{i_1} \cdots C_{i_{m}} \in \tree{\Cf}$.
Since $C_{i_{m}}, a, C_n$ is a step, it is sufficient to show
that $\rho$ was not close. Let $\pi' = \pi~D_{n+1} \cdots D_{n'} $
be such that $n' > n$ and $\new{\outi{\pi'}} \neq \movi$
or  $\new{\nb{\pi'}} \neq 0$. Then 
we have (with the notations of Subsubsection~\ref{sssec:up-step})
$\rho' = \rho~C_{i_{m+1}} \cdots C_{i_{m'}}$
and necessarily $i_{m'} = n'$ (which implies $m' > m$).
Hence $\outi{\rho'} = \new{\outi{\pi'}}$ and $\nb{\rho'} = \new{\nb{\pi'}}$,
showing that $\rho$ was not close.

\subsection{Proof of Lemma~\ref{lem:sim:prune}}

Let us show that invariants~\ref{inv:step} and~\ref{inv:sep}
hold after this operation. We use $\Jf$, $\out$, etc. to denote the values
stored before the operation, and $\new{\Jf}$, $\new{\out}$, etc. (with a bar)
to denote them after. The case of invariant~\ref{inv:step}
is trivial because $\new{\Cf} \in \Comp$ and $\new{\Jf} = \pre{x[1{:}i]}{\Jf, \Cf}(\new{\Cf})$.
Furthermore, $\new{\Cf}$ is separable by invariant~\ref{inv:close}
and Lemma~\ref{lem:carac-separ}.

\subparagraph*{Invariants~\ref{inv:lag} to~\ref{inv:lagging}.} \old{The $\new{\lag(q)}$ for $q\in \Cf'$ are mutual prefixes,
and furthermore $\new{\lag(q)} = \movi$ for some $q \in \new{\Cf}$, because
of the definition of $c$.}\cor{
The $\new{\lag(q)}$ for $q\in \Cf'$
are prefixes of $\new{\lagm}$
and furthermore $\new{\lag(q)} = \movi$ for some $q \in \new{\Cf}$, because
of the definition of $c$.} Thus invariant~\ref{inv:lag} holds.
Invariant~\ref{inv:out} holds because the $\new{\outi{\pi}}$
for $\pi \neq \new{\Cf}$ are obtained from the $\outi{\pi}$
for $\pi \neq \Cf$. Invariant~\ref{inv:last} is also trivial.

\old{Since $C' = \new{\Cf}$ was not close, there
was some  $p \in \new{\Cf}$ which was not lagging,
and so $\new{\lagm} = c^{-1}\lagm$.}
\cor{Since $C' = \new{\Cf}$ was not close, there
was some  $p \in \new{\Cf}$ which was not lagging.}
Therefore some $q \in \new{\Cf}$ is now lagging
if and only if it was lagging before the operation.
Let us consider such a $q \in \new{\Cf}$ if it exists,
and let $\pi = C_1 \cdots C_n \in \tree{\new{\Cf}}$ be such that $q \in C_n$.
The update gives $\new{\nb{\pi}} = \nb{\Cf \pi} $
which must be null since $q$ was lagging. The cases of
 $\outi{\pi}$ (for $\pi\neq \Cf$) and  $\last(q)$ are similar. Thus invariant~\ref{inv:lagging} holds.

\subparagraph*{Invariant~\ref{inv:past}.}
Let us now show invariant~\ref{inv:past}. Two cases occur given $q \in \new{\Cf}$
\begin{itemize}
\item either $q \in \new{\Cf}$ is lagging, then it was lagging
before and so we had $\val{\Jf, \Cf}{x[1{:}i]}(q) = \out~\lag(q)$.
Furthermore, since $q \in \new{\Cf}$ was lagging we necessarily had $\outi{\Cf \new{\Cf}} = \movi$
by invariant~\ref{inv:lagging}. Since the update gives
$\new{\out} = \out c \outi{\Cf \new{\Cf}}$
and $\new{\lag(q)} = c^{-1}\lag(q)$, we conclude that 
 $ \out~\lag(q)= \new{\out} \new{\lag(q)}$;
 \item or $q \in \new{\Cf}$ is not lagging, let
$\new{\pi} = C_1 \cdots C_n \in \tree{\new{\Cf}}$ be such that
$C_n = \{q\}$. Then let $\pi \defined \Cf C_1 \cdots C_n \in \tree{\Cf}$,
if $\pi_i \defined \Cf C_1 \cdots C_i$ we had:
\begin{equation*}
\begin{aligned}
\val{x[1{:}i]}{\Jf,\Cf}(q) = \out~\lagm ~\theta^{\nb{\Cf}(q)}
\left(\prod_{i=1}^n \outi{\pi_i} \theta^{\nb{\pi_i}(q)} \right) \last(q).
\end{aligned}
\end{equation*}
Now if $\new{\pi_i} \defined C_1 \cdots C_i$, the update gives
$\new{\outi{\new{\pi_i}}} = \outi{\pi_i} $ and $\new{\nb{\new{\pi_i}}} = \nb{\pi_i} $ if $i > 1$;
and $\new{\out} \becomes \out \outi{\pi_1}$ and $\new{\nb{\new{\pi_i}}} \becomes \nb{\pi_i} $ for $i = 1$.
Therefore we have:
\begin{equation*}
\begin{aligned}
\val{x[1{:}i]}{\Jf,\Cf}(q) = \out~\lagm ~
\left(\prod_{i=2}^n \new{\outi{\pi_i}} \theta^{\new{\nb{\new{\pi_i}}(q)}} \right) \theta^{\nb{\Cf}(q)} \last(q).
\end{aligned}
\end{equation*}
The result follows by update of $\new{\last(q)}$ and since $\new{\lagm} = c^{-1}\lagm$.
\end{itemize}

\subparagraph*{Invariant~\ref{inv:future}.} 
Let $\new{\Cf}, u, D$ be a step,
we show that for all $q \in D$:
$\val{\new{\Jf},D}{x[1{:}i]u}(q)
\pref \new{\out}~\new{\lagm}~\theta^\omega$.
By invariant~\ref{inv:close}, we can decompose the step
$\new{\Jf}, x[1{:}i], \new{\Cf}$ in an initial step $\new{\Jf}, x[1{:}j], E$
and a step $E, x[j{+}1{:}i], \new{\Cf}$ such that
$|\advm{\new{\Jf}, E}{x[1{:}j]}| \ge \four \Bound!$.
\old{By applying Lemma~\ref{lem:sep-theta}
to  $\new{\Jf}, x[1{:}j], E$, it follows that there exists
$\phi, \tau, \gamma$ with $|\phi| = \Bound!$
and $|\tau| \le 3 \Bound$ such that
$\advm{\Jf,E}{x[1{:}i]} = \tau \gamma$ (thus $|\gamma|  \ge \Bound!$),
$\gamma \pref \phi^\omega$ and for all step $\new{\Cf},u,D$ and $q \in D$:}
\cor{By applying Lemma~\ref{lem:sep-theta}
to  $\new{\Jf}, x[1{:}j], E$, it follows that there exists
$\phi, \tau, \gamma$ with $|\phi| = \Bound!$
and $|\tau| \le \Bound!$ such that
$\advm{\Jf,E}{x[1{:}j]} = \tau \gamma$ (thus $|\gamma|  = |\advm{\new{\Jf}, E}{x[1{:}j]}|-|\tau|  \ge \four\Bound! - \Bound!
\ge \Bound!$),
$\gamma \pref \phi^\omega$ and for all step $\new{\Cf},u,D$ and $q \in D$,
we obtain a step $E, x[j{+}1{:}i]u, D$ such that:}
\begin{equation*}
\begin{aligned}
\val{E,D}{x[j{+}1{:}i]u}(q) \pref  (\adv{\new{\Jf},E}{x[1{:}j]}(\pre{E,D}{x[j{+}1{:}i]u}(q)))^{-1} \tau  \phi^\omega.
\end{aligned}
\end{equation*}
and therefore by adding $\val{\new{\Jf},E}{x[1{:}j]}(\pre{E,D}{x[j{+}1{:}i]u}(q))$ on both sides:
\begin{equation}
\begin{aligned}
\label{eq:fff}
\val{\new{\Jf},D}{x[1{:}i]u}(q) \pref  \com{\new{\Jf},D}{x[1{:}j]} \tau  \phi^\omega.
\end{aligned}
\end{equation}
To conclude, it is thus sufficient to show the following result:
\begin{sublemma} $  \com{\new{\Jf},D}{x[1{:}j]} \tau  \phi^\omega
= \new{\out}~\new{\lagm}~\theta^\omega$.
\end{sublemma}

\begin{proof}

Since $\Cf \new{\Cf} \in \tree{\Cf}$ was not close,
there was some $\pi \neq \movi$ such that $\Cf \new{\Cf} \pi \in \tree{C}$
and $\nb{\pi} \neq 0$ or $\outi{\pi} \neq \movi$.
Therefore there is $q \in \new{\Cf}$ which was not lagging
and such that:
\begin{equation*}
\begin{aligned}
\val{\Jf,\Cf}{x[1{:}i]}(q) = \val{\new{\Jf},\new{\Cf}}{x[1{:}i]}(q) = \new{\out}~\new{\lagm} ~
\theta^{n(q)} \new{\last(q)}.
\end{aligned}
\end{equation*}
for some $n(q) > 0$. Let us consider the $q\in \new{C}$ such that $(n(q), |\last(q)|)$
is maximal. Then for all $p \in \new{C}$ we get
$\val{\new{\Jf},\new{\Cf}}{x[1{:}i]}(p) \pref \val{\new{\Jf},\new{\Cf}}{x[1{:}i]}(q)$,
thus (use Equation~\ref{eq:fff} for the right handside):
\begin{equation}
\label{eq:phi-theta}
\begin{aligned}
\com{\new{\Jf},E}{x[1{:}j]} \tau \gamma =  \com{\new{\Jf},E}{x[1{:}j]} \advm{\Jf',E}{x[1{:}j]}  \pref \new{\out}~\new{\lagm} ~
\theta^{n(q)} \new{\last(q)} \pref \com{\new{\Jf},D}{x[1{:}j]} \tau  \phi^\omega.
\end{aligned}
\end{equation}
Two cases can occur depending on the sign of $k \defined | \new{\out}~\new{\lagm} | - | \com{\Jf',D}{x[1{:}j]} \tau|  $
\begin{itemize}
\item either $k \ge 0$, then $ \new{\out}~\new{\lagm}  = \com{\Jf',D}{x[1{:}j]} \tau (\phi^\omega [1{:}k]) $
and $\theta  = \phi^\omega[k{+}1:k{+}1 {+}\Bound!]$;
\item or $k < 0$, then $ \new{\out}~\new{\lagm} (( \theta^{n(q)}  \new{\last(q)})[|k|{:}])
= \com{\Jf',D}{x[1{:}j]} \tau$
and $\gamma \pref \theta^{n(q)} \new{\last(q)}[|k|{+}1{:}]$.
But since $|\gamma| \ge \Bound!$ and $\gamma \pref \phi^\omega$,
we conclude that $\phi = \theta^\omega[|k|{+}1:|k|{+}1 {+}\Bound! ]$.
\end{itemize}
In both cases, we conclude that $  \com{\new{\Jf},D}{x[1{:}j]} \tau  \phi^\omega
= \new{\out}~\new{\lagm}~\theta^\omega$.
\end{proof}

\subparagraph*{Invariant~\ref{inv:close}.} 
Let us consider $\new{\pi} = C_1 \cdots C_n \in \tree{\new{\Cf}}$,
$\new{\pi}$ not close after the operation.
Then there exists $ \new{\pi}  \prefneq \new{\pi'}  =  C_1 \cdots C_{n'} \in \tree{\new{\Cf}}$ such that
$\new{\nb{\new{\pi'}}} \neq 0$ or $\new{\outi{\new{\pi'}}} \neq \movi$. 
Let $\pi' \defined \Cf C_1 \cdots C_{n'} $, then
then $\Cf \new{\Cf} \prefneq \pi'$, and by the updates
we get $\nb{\pi'} = \new{\nb{\new{\pi'}}}$
and $\outi{\pi'} = \new{\outi{\new{\pi'}}}$.
Finally, since invariant~\ref{inv:close} held
before the operation, it still holds.

\section{Proofs of section~\ref{sec:bound}: boundedness and productivity of $\strans$}

For $x \in A^\omega$ and $i \ge 0$, we denote by $\cro{\outi{\pi}}^x_{i}$, etc.
the values of the registers of $\strans$ after reading $C^x_0 \cdots x[i] C^x_i$
(when defined).

\subsection{Proof of Lemma~\ref{lem:strans-bound}: $1$-boundedness of $\strans$}

Given $\pi \in \tree{C^x_i}$, we say that $\cro{\outi{\pi}}^x_{i}$
is \emph{$1$-bounded} if for all $0 \le i' \le i$ and $\pi' \in \tree{C^x_{i'}}$,
$\outi{\pi'}$ is occurs at most once
in $\sigma(\outi{\pi})$, where $\sigma$ is the substitution
applied by $\strans$  when reading $x[i'{+}1]C^x_{i'+1} \cdots C^x_i$.
Given $\pi, \rho \in \tree{C^x_i}$, we say that
$\cro{\outi{\pi}}^x_{i}$ and $\cro{\outi{\rho}}^x_{i}$ have \emph{no shared memory} if
for all $0 \le i' \le i$ and $\pi' \in \tree{C^x_{i'}}$,
$\outi{\pi'}$ does not occur both in $\sigma(\outi{\pi})$ and $\sigma(\outi{\rho})$.
Lemma~\ref{lem:strans-bound} immediately follows from Sublemma~\ref{slem:indu-bound}.

\begin{sublemma} \label{slem:indu-bound}
For all $\pi \in  \tree{C^x_i}$, $\cro{\outi{\pi}}^x_{i}$ is $1$-bounded.
Furthermore, if $\rho \in  \tree{C^x_i}$ is such that $\pi \prefneq \rho$,
then $\cro{\outi{\pi}}^x_{i}$ and $\cro{\outi{\rho}}^x_{i}$
have no shared memory.
\end{sublemma}

The rest of this subsection is devoted to the proof of
Sublemma~\ref{slem:indu-bound} by induction on $i \ge 0$.
The result is obvious for $i=0$. Assume now by induction that it holds
for some $i \ge 0$.

First note that Subsection~\ref{ssec:tool} only adds constant values
in the register, hence its applications will always preserve our property.
Now, if $C^x_i$ was separable, the transition of $\strans$ uses Subsection~\ref{ssec:origin-nosep}.
The value $\cro{\out}^x_{i+1}$ is obtained using  $\cro{\out}^x_{i}$ once,
plus constant values. Furthermore, if  $C^x_{i+1}$ is not separable,
then each value $\cro{\outi{\pi}}^x_{i+1}$ is built from constant
values, hence they are $1$-bounded and they share no memory.
The result holds in $i{+}1$.

Now if $C^x_{i+1}$ is not separable, the transition may
first apply Subsubsection~\ref{sssec:prepro}.
If $\pi$ is close, then the situation is similar to that of
 Subsection~\ref{ssec:origin-nosep}, except that we may
use $\outi{\Cf C'}$ to update $\out$. However, since
$\cro{\outi{\Cf C'}}^x_{i}$ and $\cro{\out}^x_i = \cro{\outi{\Cf}}^x_i$
are $1$-bounded and have no shared memory
by induction hypothesis, then the resulting value of $\out$
is $1$-bounded. Now if $\pi$ is not close, the argument for
$\out$ is similar. Furthermore, we update $\outi{C' \pi} \becomes \outi{\Cf C C' \pi}$
for $\pi \neq \movi$, which clearly preserves
the fact that these registers are $1$-bounded and have no shared memory.
Furthermore, they also have no shared memory with $\cro{\out}^x_{i+1}$.

Let us finally consider the application of Subsubsection~\ref{sssec:up-step}.
The updates clearly preserve $1$-boundedness since
there are no concatenations.
Let $\pi_n \defined D_1 \cdots D_n \in \tree{C^x_{i+1}}$, we show that
if $\pi_{n'} \defined D_1 \cdots D_{n'}$ for $n' < n$, then
$\cro{\outi{\pi_n}}^x_{i+1}$ and $\cro{\outi{\pi_{n'}}}^x_{i+1}$
have no shared memory. Let $1 = i_1 < \cdots < i_m \le n$
be given by Subsubsection~\ref{sssec:up-step} for $\pi_n$.
If $i_m < n$ the result is clear since $\outi{\pi_{n}} \becomes \movi$
(and we may finally add constant values in it by Subsection~\ref{ssec:tool}).
Otherwise $\outi{\pi_{n}} \becomes \outi{C_{i_1} \cdots C_{i_m}}$.
If $n' \neq i_{m'}$ for some $1 \le m' \le m$, then the result is clear since
$\outi{\pi_{n'}} \becomes \movi$. Otherwise
$\outi{\pi_{n'}} \becomes \outi{C_{i_1} \cdots C_{i_{m'}}}$.
But necessarily $m' < m$ (because $n' < n'$) and so by induction hypothesis
 the former values  of $\outi{C_{i_1} \cdots C_{i_m'}}$
and  $\outi{C_{i_1} \cdots C_{i_m'}}$ shared
no memory. The result follows since Subsection~\ref{ssec:tool} only adds constant values.

\subsection{Proof of Lemma~\ref{lem:infini}: productivity of $\strans$}

Let us fix a word $x \in \Dom{f}$, we want
to show that when $\strans$
reads $g(x) = C^x_0 x[1] \cdots$, we have $|\cro{\out}^x_i| \fonc \infty$. Let us first suppose
that the transitions of $\strans$ on $g(x)$ use Subsection~\ref{ssec:origin-nosep}
or the ``close'' paragraph of Subsubsection~\ref{sssec:up-step}
infinitely often. Then for infinitely many $i \ge 0$
we have $\cro{\out}^x_i = \com{\Jf, \Cf}{x[1{:}i]}$,
and the $\alpha_q = \adv{\Jf, \Cf}{x[1{:}i]}(q) $
have a size bounded by $\Bound + 2 \times \four \Bound!$.
Since $\val{\Jf,\Cf}{x[1{:}i]}(q^x_i)$ tends to $f(x)$,
we conclude that $\out$ also tends to an infinite word.

Now assume that there exists $N \ge 0$ such
that when reading the suffix of its input
$C^x_N x[N] C^x_{N+1} x[N{+}1]  \cdots$,
$\strans$ only uses  Subsubsection~\ref{sssec:prepro}
in the ``non-close'' case, and Subsubsection~\ref{sssec:up-step},
when doing its transitions.
The rest of the proof is done by contradiction: we assume
that $\strans$ only adds empty words in $\out$ when performing these
transitions.

A key ingredient
to reach a contradiction will be the fact that $\trans$ is \emph{productive}.

\begin{sublemma}[Productivity.]
\label{slem:prod}
Let $J, u, S$ be an initial step and
$S, u', S$ be a step such that  $u' \neq \movi$, $\pre{S,S}{u'}: S \fonc S$
is the identity function, and some $q'_1 \in S$ is accepting.
Then $\val{S,S}{u'}(q) \neq \movi$
for all $q \in S$.
\end{sublemma}

\begin{proof} Let $q'_2 \in S$. If it is accepting the result follows
since a productive \oNT{} is clean. Otherwise, by definition of
steps there exists (uniques) $q_1, q_2 \in I$,
$\alpha_1, \alpha_2 \in B^*$
such that $q_i \runs{u | \alpha_i} q'_i \runs{u' | \val{S,S}{u'}(q'_i)} q'_i$
for $i \in \{1,2\}$. These are the conditions of Lemma~\ref{lem:continuity-loops},
hence since $\trans$ is productive we get $\val{S,S}{u'}(q'_2) \neq \movi$.
\end{proof}

\begin{definition}\label{def:lim}
For all $i \ge 0$,  let
$S^x_i \defined \bigcap_{i' \ge i} \pre{C^x_i, C^x_{i'}}{x[i{+}1{:}i']}(C^x_{i'})$.
\end{definition}

\begin{sublemma} For all $i \ge 0$, $q^x_i \in S^x_i$ and
$S^x_i, x[i{+}1], S^x_{i+1}$ is a step.
\end{sublemma}

\begin{proof}
For all $i' \ge i$, we
have $q^x_i = \pre{C^x_i, C^x_{i'}}{x[i{+}1{:}i']}(q^x_{i'})$
and $q^x_{i'} \in C^x_{i'}$. Hence $q^x_i \in S^x_i$.
Let us now show that $p \in S^x_i$
if and only if there exists a sequence
$(p_{i'})_{i' \ge i}$ such that $p_{i} = p$,
$p_{i'} \in C^x_{i'}$ and $p_{i'} = \pre{C^x_{i'}, C^x_{i'+1}}{x[i'+1]}(p_{i'+1})$
for all $i' \ge i$.  The ``if'' direction is obvious.
Conversely, if $p \in S^x_i$, then
for all $n \ge i$, there exists a finite
sequence $(p_{i'})_{i \le i' \le n}$
such that $p_{i} = p$,
$p_{i'} \in C^x_{i'}$ and
$p_{i'} = \pre{C^x_{i'}, C^x_{i'+1}}{x[i'+1]}(p_{i'+1})$.
By König's lemma (see $\pre{}{}$ as the ancestor relation
in a tree), we can build an
infinite sequence $(p_{i'})_{i' \ge i}$.
From this characterization, it follows that $S^x_i \in \Comp$,
and $\pre{C^x_i, C^x_{i+1}}{x[i{+}1]}(S^x_{i+1}) = S^x_i$,
which implies that $S^x_i,x[i{+}1], S^x_i$ is a step.
\end{proof}
We claim that the $S^x_i$ completely ``cover'' the set $C^x_{i'}$ at some
point in the future.

\begin{sublemma} \label{slem:cover-Sx} For all $i \ge 0$, there exists $i' \ge i$ such
that $\pre{C^x_{i}, C^x_{i'}}{x[i{+}1{:}i']}(C^x_{i'}) = S^x_i$. 
\end{sublemma}

\begin{proof} Since $C^x_{i'}, x[i'{+}1], C^x_{i'+1}$ is a
pre-step, $\pre{C^x_i, C^x_{i'}}{x[i{+}1{:}i']}(C^x_{i'}) \supseteq
\pre{C^x_i, C^x_{i'}}{x[i{+}1{:}i'{+}1]}(C^x_{i'{+}1}) $.
Hence $(\pre{C^x_i, C^x_{i'}}{x[i{+}1{:}i']}(C^x_{i'}))_{i' \ge i}$
is ultimately constant, and by Definition~\ref{def:lim} its limit is $S^x_i$.
\end{proof}
Finally, let us extract a sequence of positions where $S^x_i$ is constant.

\begin{sublemma} \label{slem:sequence} There exists a sequence $N \le \ell_1 < \ell_2 < \cdots $
such that $S^x_{\ell_1} = S^x_{\ell_2} = \cdots =: S$, $q^x_{\ell_1} = q^x_{\ell_2}  = \cdots \in F$
and $\pre{S,S}{ x[\ell_j{+}1{:}\ell_{j'}]}$ is the identity function
for all $j \le j'$.
\end{sublemma}

\begin{proof}
Since $(q^x_i)_{i \ge 0}$ is accepting,
one can extract an infinite sequence $N \le \ell_1 < \ell_2 < \cdots$
such that $q^x_{\ell_1} = q^x_{\ell_2} = \cdots \in F $.
Up to extracting a subsequence with Ramsey's theorem
for singletons (i.e. the pigeonhole principle), we can assume
that $S^x_{\ell_1} = S^x_{\ell_2} = \cdots = S$.
Up to extracting a subsequence using Ramsey's theorem for pairs
(color a pair $j \le j'$ by $\pre{S,S}{ x[\ell_j{+}1{:}\ell_{j'}]}$,
which is a permutation of $S$ since $S, x[\ell_j{+}1{:}\ell_{j'}], S$ is a step),
we can assume that  $\pre{S,S}{ x[\ell_j{+}1{:}\ell_{j'}]}$
for $j \le j'$ is the identity function.
\end{proof}

By sublemmas~\ref{slem:prod} and~\ref{slem:sequence}, we get
$ \val{C^x_{\ell_j}, C^x_{\ell_{j+1}}}{ x[\ell_j{+}1{:}\ell_{j+1}]}(q)
= \val{S,S}{ x[\ell_j{+}1{:}\ell_{j+1}]}(q) \neq \movi$
for all $j \ge 1$ and all $q \in S$.
Therefore $| \val{C^x_{\ell_1}, C^x_{\ell_{1+K}}}{ x[\ell_1{+}1{:}\ell_{1+K}]}(q)| \ge K$ for all 
$K \ge 1$ and $q \in S$.

Let us now fix $K \defined \four \Bound!$, $i \defined \ell_1$ and $i' \defined \ell_{1+K}$.
By Sublemma~\ref{slem:cover-Sx}, there exists $i'' \ge i$ such that
$\pre{C^x_{i'}, C^x_{i''}}{x[i'{+}1{:}i'']}(C^x_{i''}) = S^x_{i'} = S$.
Hence for all $q \in C^x_{i''}$, if $q' \defined \pre{C^x_{i'}, C^x_{i''}}{x[i'{+}1{:}i'']}(C^x_{i''})$ we get
$
|\val{C^x_{i}, C^x_{i''}}{x[i{+}1{:}i'']}(q)| 
= |\val{S, S}{x[i{+}1{:}i']}(q')  \val{C^x_{i'}, C^x_{i''}}{x[i'{+}1{:}i'']}(q)| \ge \four\Bound!
$.

The last operation which was applied by
$\strans$ when reading from $x[i''] C^x_{i''}$
is Algorithm~\ref{algo:down}. By definition of $m$
in $\tnorm{\textbf{down}}(C^x_{i''})$, there exists $q \in C^x_{i''}$ 
such that $\cro{\nb{C^x_{i''}}(q)}_{i''}^x = 0$.
Thus we can apply Sublemma~\ref{slem:infini}
with $i_1 = i$ and $i_2 = i''$, and Remark~\ref{rem:clefs}
yields a contradiction.

\begin{sublemma} \label{slem:infini} Let $N \le i_1 \le i_2$.
Let $q_2 \in C^x_{i_2}$ be such that $\cro{\nb{C^x_{i_2}}(q_2)}^x_{i_2} = 0$.
Let us define $q_1 \defined \pre{C^x_{i_1},C^x_{i_2}}{x[i_1{+}1{:}i_2]}(q_2)$,
then  $\cro{\nb{C^x_{i_1}}}^x_{i_1}(q_1) = 0$ and:

$
|\val{C^x_{i_1},C^x_{i_2}}{x[i_1{+}1{:}i_2]}(q_2)| = |\cro{\last}^x_{i_2}(q_2)| - |\cro{\last}^x_{i_1}(q_1)|
+ |\cro{\lag}^x_{i_2}(q_2)| - |\cro{\lag}^x_{i_1}(q_1)|.
$
\end{sublemma}

\begin{remark} \label{rem:clefs}
In particular, we get $\val{C^x_{i_1},C^x_{i_2}}{x[i_1{+}1{:}i_2]}(q_2)| < \four \Bound!$%
\cor{because the size of $\lag$ is bounded
by $\Bound!$ and the size of $\last$ is (strictly) bounded by $\Bound!$}.
\end{remark}

\begin{proof}
The proof consists in a decreasing induction on
$i_1 \le i_2$.
The base case being trivial, let us show
it for $i_0 \defined i_1{-}1$.
Let $q_2 \in C^x_{i_2}$ be such that $\cro{\nb{C^x_{i_2}}(q_2)}^x_{i_2} = 0$
and $q_1 \defined \pre{C^x_{i_1},C^x_{i_2}}{x[i_1{+}1{:}i_2]}(q_2)$.
By induction hypothesis, Sublemma~\ref{slem:infini} holds.
Now let us consider the transition of $\strans$
from $C^x_{i_0}$ to $C^x_{i_1}$. It obtained by possibly
applying Subsubsection~\ref{sssec:prepro}
(in the ``non-close'' case)
and then Subsubsection~\ref{sssec:up-step}.
We study the preservation of our property along
these operations, starting from the last one (we backtrack
on the computation).

\subparagraph*{Last operation: applying Subsection~\ref{ssec:tool} in Subsubsection~\ref{sssec:up-step}.}
Let $\nb{\pi}', \outi{\pi}'$, etc.
denote the configuration of $\strans$ right after Lemma~\ref{lem:sim:prune}
in Subsubsection~\ref{sssec:up-step}.
Since Algorithm~\ref{algo:down} has not modified $\out = \outi{C^x_{i_1}}$,
we had $m=0$ in $\tnorm{\textbf{down}}(C^x_{i_1})$. 
Thus we had $\nb{C^x_{i_1}}'(q_1) +n(q_1)= 0$
(because if $m=0$ then $\cro{\nb{C^x_{i_1}}}^x_{i_1}(q_1) = \min (\nb{C^x_{i_1}}'(q_1) +n(q_1), \four)$).
Therefore $\nb{C^x_{i_1}}'(q_1) = 0$. Furthermore
$n(q_1) = 0$ and so $\last'(q_1) = \cro{\last(q_1)}^x_{i_1}$.
Furthermore, $\lag'(q_1) = \cro{\lag(q_1)}^x_{i_1}$.

\subparagraph*{Previous operation: beginning of Subsubsection~\ref{sssec:up-step}.}
Now let $\nb{\pi}'', \outi{\pi}''$, etc.
 denote the configuration of $\strans$ before applying the whole
 Subsubsection~\ref{sssec:up-step}. Since by construction
$\Cf =  \pre{C^x_{i_0}, C^x_{i_1}}{x[i_1]}(C^x_{i_1})$, then
$\nb{C^x_{i_1}}' = \nb{\Cf}'' \circ \pre{\Cf, C^x_{i_1}}{x[i_1]}$.
Thus $ \nb{\Cf}''(q_0) = 0$ if $q_0 \defined \pre{\Cf, C^x_{i_1}}{x[i_1]}(q_1)$.

Finally, we note that $c = \movi$ since there is no output.
As a consequence, it is quite easy to see that
$|\val{C^x_{i_0},C^x_{i_1}}{x[i_1]}(q_1)| 
= |\val{\Cf,C^x_{i_1}}{x[i_1]}(q_1)| = |\last'(q_1)| - |\last''(q_0)|
+|\lag'(q_1)| - |\lag''(q_0)|$.

\subparagraph*{Previous operation: Subsubsection~\ref{sssec:prepro}.}
If Subsubsection~\ref{sssec:prepro} was not used,
the proof is completed.  Otherwise,  let $\nb{\pi}''', \outi{\pi}'''$, etc.
denote the information of $\strans$ right after Lemma~\ref{lem:sim:prune}).
By an analysis of Algorithm~\ref{algo:down} (similar to what we did above),
we see that $\nb{\Cf}'''(q_0) =0$, $\last'''(q_0) = \last''(q_0)$
and $\lag'''(q_0) = \lag''(q_0)$. Furthermore
$n(q_0) = 0$ thus $|\last'''(q_0)| < \Bound!$.

Let us finally consider the rest of Subsubsection~\ref{sssec:prepro}
in the ``non-close'' case.
Since there is no output, $c = \movi$ thus $\cro{\lag}^x_{i_0}(q_0) = \lag'''(q_0)$.
Since $|\last'''(q_0)| <  \Bound!$, then 
$\theta^{\cro{\nb{C^x_{i_0}}}^x_{i_0}(q_0)} = \movi$
thus $\cro{\nb{C^x_{i_0}}}^x_{i_0}(q_0) = 0$.
Furthermore  $\cro{\last}^x_{i_0}(q_0) = \last'''(q_0)$.
\end{proof}

\end{document}